\documentclass[%
 reprint,
superscriptaddress,
amsmath,amssymb,
aps,
pra,
floatfix,
]{revtex4-2}

\usepackage[english]{babel}
\usepackage{graphicx} 
\usepackage{physics}
\usepackage{mathtools}

\usepackage[colorlinks=true, allcolors=blue]{hyperref}

\usepackage[english]{babel}

\usepackage{amsmath,amssymb,amsthm}
\usepackage{algorithm}
\usepackage{algpseudocode}
\usepackage{tikz}
\usetikzlibrary{calc}
\usepackage{graphicx}
\usepackage{pgf}

\DeclareMathOperator*{\argmin}{arg\,min}
\DeclareMathOperator*{\poly}{poly}

\newtheorem{lemma}{Lemma}
\newtheorem{proposition}{Proposition}
\theoremstyle{definition}

\begin{document}

\title{Leveraging Landau-Zener-Stückelberg interference for accelerating diabatic quantum annealing}

\author{Matthias Werner}
\email{matthias.werner@qilimanjaro.tech}
\affiliation{Qilimanjaro Quantum Tech., Carrer de Vene{\c{c}}uela, 74, Sant Mart{\'i}, 08019 Barcelona, Spain}
\affiliation{Departament de F{\'i}sica Qu{\`a}ntica i Astrofísica (FQA), Universitat de Barcelona (UB), Carrer de Mart{\'i} i Franqu{\'e}s, 1, 08028 Barcelona, Spain}
\affiliation{Institut de Ci{\`e}ncies del Cosmos, Universitat de Barcelona, ICCUB, Carrer de Mart{\'i} i Franqu{\`e}s, 1, 08028 Barcelona, Spain}

\author{Mat\'ias Jonsson}
\noaffiliation

\author{Artur Garc\'ia-S\'aez}
\affiliation{Qilimanjaro Quantum Tech., Carrer de Vene{\c{c}}uela, 74, Sant Mart{\'i}, 08019 Barcelona, Spain}
\affiliation{Barcelona Supercomputing Center, Pl. Eusebi G{\"u}ell 1, 08032 Barcelona, Spain}

\author{Arnau Riera}
\affiliation{Qilimanjaro Quantum Tech., Carrer de Vene{\c{c}}uela, 74, Sant Mart{\'i}, 08019 Barcelona, Spain}

\author{Tameem Albash}
\email{tnalbas@sandia.gov}
\affiliation{Center for Computing Research, Sandia National Laboratories, Albuquerque NM, 87185 USA}

\begin{abstract}
    Diabatic quantum annealing with variationally optimized schedules can exhibit exponential speedups over conventional adiabatic quantum annealing, as was demonstrated numerically for a frustrated Ising ring model by C{\^o}t{\'e} et al. 
Here we identify Landau-Zener-St{\"u}ckelberg interference as the underlying mechanism for this speedup, and based on this insight we propose a variational schedule ansatz with far fewer parameters. This simplified ansatz allows us to show analytically that the classical optimization of the schedule parameters can be done in polynomial time and discuss conditions when we expect this type of mechanism to provide speedups over adiabatic annealing. 
Furthermore, we provide an analytical argument that coherence is an essential resource for this mechanism, which we verify numerically. 
We perform extensive numerical tests of the proposed ansatz and observe substantial improvements over adiabatic annealing and competitive performance in particularly challenging problem instances, including a well-studied MAXCUT instance commonly used for benchmarking. Our work shows that explicitly leveraging physical mechanisms can lead to more effective designs of variational annealing algorithms.
\end{abstract}

\maketitle

\section{Introduction}
Quantum annealing (QA)~\cite{Ray1989,Finnila1994,Kadowaki1998,Farhi2001,Santoro2002} is a heuristic quantum algorithm for tackling classical optimization problems. In its standard formulation, it relies on the principle of computation via adiabatic evolution~\cite{Farhi2000}.
The solution to a computational problem is encoded in the ground state of a target Hamiltonian, and the system Hamiltonian is slowly transformed towards this target from an easily prepared ground state of a simple Hamiltonian. The adiabatic theorem~\cite{Born1928,Kato1950,Jansen2007,Amin_2009a} guarantees that, if the evolution is slow enough, the system will remain in the ground state with high probability throughout the evolution, and the final state will yield the desired solution. For ground state evolution, ``slow'' is relative to the time scale set by the smallest energy gap between the instantaneous ground state and first excited state of the interpolating Hamiltonian connecting the initial and target Hamiltonians.

We consider the typical setting of QA for $N$ qubits. Here, the simple Hamiltonian is typically given by the uniform transverse field Hamiltonian:
\begin{equation}
    H_X = -\sum_{n=1}^N X_n \ ,
    \label{eq:DefHx}
\end{equation}
while the target Hamiltonian encoding the solution to an optimization problem is of the Ising form:
\begin{equation}
    H_Z = \sum_{n=1}^N h_n Z_n + \sum_{m>n}^N J_{mn} Z_m Z_n \ .
    \label{eq:DefHz}
\end{equation}
Here, $X_n$ and $Z_n$ are the Pauli-X and Pauli-Z operators acting on qubit $n$ with $+1$ eigenstates $\ket{+}$ and $\ket{0}$ respectively. In QA, the system is initially prepared in the ground state of $H_X$, i.e. $\ket{+}^{\otimes N}$. Then, the system evolves for time $T$ under the time-dependent Hamiltonian
\begin{equation}
    H(s) = (1-s(t)) H_X + s(t) H_Z \ ,
    \label{eq:DefFullH}
\end{equation}
where the time-dependent function $s: [ 0, T ] \rightarrow [ 0, 1 ] $ with $s(t=0) = 0$ and $s(t=T) = 1$ governs the anneal and is called the schedule. In conventional quantum annealing, one simply performs a linear ramp, i.e. $s(t) = t/T$. In the limit of large $T$ satisfying the adiabatic condition, the system will evolve to a state with high overlap with the ground state of $H_Z$, and a measurement of the qubits will give the solution to the optimization problem with high probability. 

Whether $T$ is sufficiently high for the evolution to be adiabatic is determined by the adiabatic theorem \cite{Born1928,Kato1950,Jansen2007,Amin_2009a, Albash_2018}. Let $E_i(s)$ be the $i$-th eigenvalue of the instantaneous Hamiltonian $H(s)$. A commonly used adiabatic condition~\cite{messiah1962quantum} states that the transition between the $i$-th and the $j$-th instantaneous energy level of $H(s)$ can be neglected if
\begin{equation}
    \frac{1}{T} \max_{s} \frac{|\langle E_i(s)| \partial_s H(s) | E_j(s) \rangle |}{|E_i(s) - E_j(s)|^2} \ll 1 \ .
    \label{eq:AdiabaticTheorem}
\end{equation}
Thus, if the spectral gap between ground and first excited state closes exponentially fast as a function of $N$, the anneal time $T$ must be chosen exponentially large in order to remain in the ground state throughout the evolution.

In many interesting optimization problems, the interpolating Hamiltonian $H(s)$ exhibits an avoided level crossing, also referred to as a perturbative anti-crossing, with an exponentially small gap~\cite{vanDam2001,Reichardt2004,Farhi2012}, which then requires exponentially long anneal times to reach the adiabatic limit. In some cases, such exponential slowdowns can be circumvented by choosing smart Hamiltonian interpolations~\cite{Seki2012} and encodings~\cite{Choi2010,Dickson2011a,Dickson2011b}, but it is also expected that some exponential slowdowns are unavoidable, as it can be seen as a signature of the problem's complexity~\cite{Kitaev2002}.

An alternative approach to avoiding exponentially long anneal times in these cases is to relax the requirement of adiabaticity and allow for diabatic transitions out of and back into the ground state~\cite{Somma2012,Crosson2014,Muthukrishnan2016,Crosson2021}, which allows for the possibility of faster evolution times but usually without provable guarantees of success. A concrete example of this was provided recently in Ref.~\cite{Cote_2023}, where the target Hamiltonian $H_Z$ is a 1-dimensional frustrated Ising ring. The parameters of the Ising Hamiltonian $H_Z$ are chosen such that the interpolating Hamiltonian $H(s)$ exhibits an exponentially closing energy gap with increasing system size~\cite{Roberts_2020}, despite the ground state of $H_Z$ being computationally trivial to find. Instead of an adiabatic linear anneal, which keeps the system in the instantaneous ground state throughout the evolution, Ref.~\cite{Cote_2023} considers non-adiabatic evolutions with annealing schedules optimized to minimize the energy of the final state. This optimization is performed over a set of variational schedules composed of concatenated linear ramps, with the parameters of the optimization corresponding to the different stopping points between ramps. With this approach, sufficiently optimized annealing schedules can avoid the exponential slowdown of the adiabatic protocol for the frustrated ring model, whereby the system is driven to the first excited state from the ground state just before the exponentially closing gap and the ramp through the minimal gap merely inverts the populations.

Here, we build on the observations of Ref.~\cite{Cote_2023} and address open questions about the mechanisms for the speedup and apply our conclusion to a broader context. Our contributions are the following:
\begin{itemize}
    \item We provide a theoretical explanation for the speedup observed in Ref.~\cite{Cote_2023}.
    \item Based on this understanding, we develop a simplified version of the variational schedules that maintains the exponential speedup while significantly reducing the optimization cost (from 100 to 7 optimization parameters). Our ansatz reproduces the simulation results from Ref.~\cite{Cote_2023} and even slightly exceeds them by achieving sub-quadratic scaling exponents on the frustrated Ising ring. Furthermore, we show that if we use a different interpolation method, the scaling exponents can be further reduced.
    \item Using this simplified ansatz, we can show that the computational cost of the optimization is at most polynomial in system size, thus showing that the observed speedup is genuine and not an artifact of an exponentially expensive optimization.
    \item We discuss sufficient conditions when we expect our ansatz to deliver a significant speedup over adiabatic QA. We analyze both theoretically and numerically the impact of decoherence on the ansatz and conclude that coherence is a crucial resource.
    \item We test our theoretical intuitions in extensive numerical studies, not only on the frustrated Ising ring, but also on MAXCUT instances, as well as a toy model of problem instances showing perturbative anti-crossings. Here we show that our ansatz is a viable strategy to overcome exponentially closing gaps under certain circumstances.
\end{itemize}

Some of the models we investigate numerically here have been studied in prior work. This makes them well-suited as benchmarks to put our results into the context of alternative approaches. In particular, Wang et al.~\cite{Wang_2025} recently applied an annealing-inspired digital approach to the frustrated Ising ring and observed linear scaling of the required computational time with the system size, not considering the cost of the optimization of the parameters. Arezzo et al.~\cite{Arezzo_2025} apply the Quantum Approximate Optimization Algorithm (QAOA)~\cite{Farhi2014} to the frustrated Ising ring and establish a quadratic bound on the required circuit depth. However, both approaches substantially increase the parameter count as the system grows. Our ansatz does not quite match the linear scaling by Wang et al., however, we observe a competitive scaling of the minimal anneal time, while requiring significantly fewer parameters to be optimized.

Similarly, Pecci et al.~\cite{Pecci_2024} investigated various QAOA techniques to solve a MAXCUT instance that displays a particularly small spectral gap. Here the QAOA circuit depth and thus the parameter number also needs to be increased in order to achieve lower residual energies. In comparison, as we demonstrate below, our proposed ansatz achieves comparable performance, while only requiring seven parameters to be optimized.

Our manuscript is structured as follows. First, we review a common cause of exponentially closing spectral gaps in Section~\ref{sec:CauseOfClosingGaps} and subsequently introduce the basics of Landau-Zener-St{\"u}ckelberg (LZS) interference in Section \ref{sec:LZSInterference}. In Section~\ref{sec:PopulationTransferMechanism} we discuss how LZS interference can be used to overcome exponentially closing spectral gaps in some cases, which will form the basis for our proposed annealing schedule ansatz in Section~\ref{sec:ScheduleIntroduciton}. We analyze the capabilities of the ansatz in Section~\ref{sec:SchedulesAsGates} and how it is affected by decoherence in Section~\ref{sec:DecoherenceTheory}. The theory section concludes in Section~\ref{sec:EfficientOptimization} with an argument that the optimization of the schedule parameters is efficient, and we discuss in Section~\ref{sec:OptimalityConditions} sufficient conditions under which we can expect the ansatz to provide an advantage over linear ramp QA.

In Section~\ref{sec:Numerics} we discuss our numerical results. We apply the ansatz to the frustrated Ising ring in Section \ref{sec:IsingRingSimulation}, and in Section \ref{sec:ToyModel} we study a toy model that allows us to distinguish qualitatively different problem instances with and without a perturbative anti-crossing resulting in an exponentially closing spectral gap. This insight is then applied to an instance of the MAXCUT problem in Section \ref{sec:Instance1} and a numerical verification of our theoretical discussion of the impact of decoherence on the working mechanism of the ansatz in Section \ref{sec:OpenSystem}. The numerical results conclude in Section \ref{sec:CubicSplines} with a discussion of an additional speedup on the frustrated Ising ring by choosing cubic spline interpolation, instead of linear interpolation. Finally, we summarize our findings in the Conclusion in Section \ref{sec:Conclusion}.
\section{Theoretical findings}
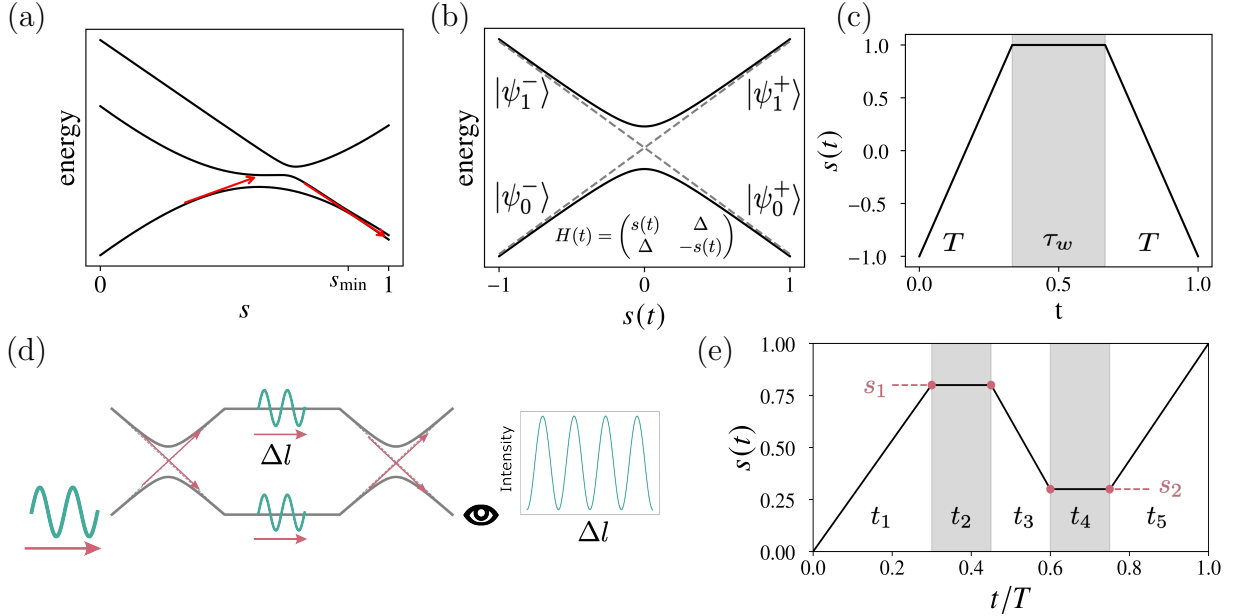
\begin{figure*}
    \centering
    \begin{tikzpicture}
        \node[anchor=north west, inner sep=0] (fig) at (0,0)
            {\includegraphics[scale=0.3125]{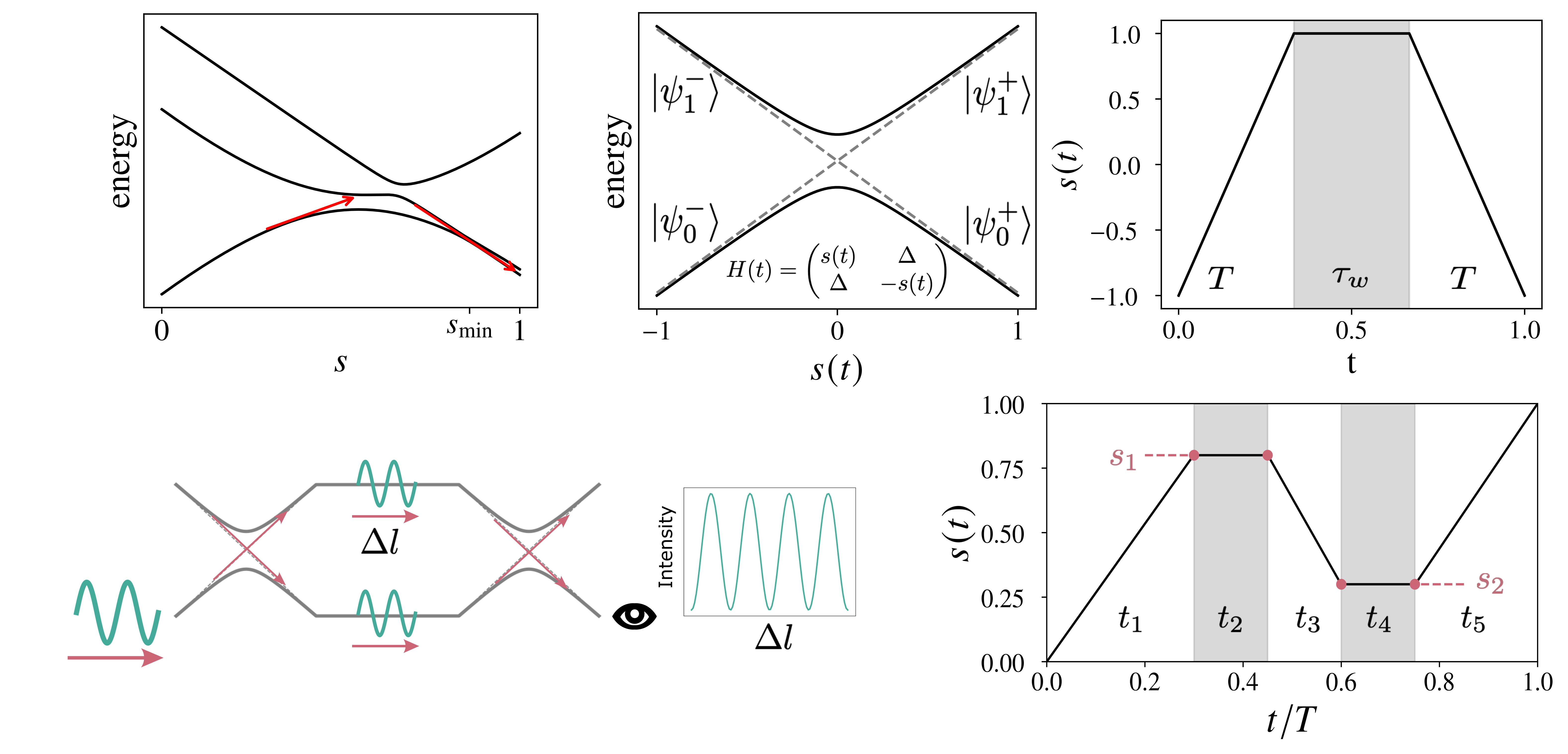}};
        \node[font=\large] at ($(fig.north west) + (fig.north west)!0.045!(fig.north east)$) {(a)};
        \node[font=\large] at ($(fig.north west)!0.38!(fig.north east)$) {(b)};
        \node[font=\large] at ($(fig.north west)!0.7!(fig.north east)$) {(c)};
        \node[font=\large] at ($(fig.north west)!0.55!(fig.south west) + (fig.north west)!0.045!(fig.north east)$) {(d)};
        \node[font=\large] at ($(fig.north west)!0.55!(fig.south west) + (fig.north west)!0.59!(fig.north east)$) {(e)};
    \end{tikzpicture}
    \caption{(a) Diagram of the three lowest energy eigenvalues of $H(s)$, red arrows indicating the diabatic transitions that allow to reach the ground state without slowing down for the minimal gap at $s_{\min}$. (b) Eigenvalues of the two-level system $H(t)$ to illustrate Landau-Zener transitions. (c) Two consecutive sweeps in time $T$ with a waiting period of $\tau_w$ between them. (d) Illustration of an interferometer, where the two distinct paths a wave can take are analogous to the eigenstates of the two-level system. One path is longer by a difference $\Delta l$, resulting in interference. The interference manifests as oscillations in the intensity, or equivalently as oscillations in the population of the instantaneous eigenstates after two transition processes. (e) Simplified piece-wise linear schedule ansatz, where two schedule plateaus are introduced at locations $s_1$ and $s_2$. The locations of the plateaus, their durations $t_{2}, t_{4}$, as well as the durations of the linear ramps $t_{1}, t_{3}, t_{5}$ form the seven parameters of the proposed ansatz. Note that for the numerical experiments below, we use the total time $T=\sum_{i=1}^5 t_i$ as an optimization parameter, so only four of the five $t_i$ are independent.}
    \label{fig:DiagrammeIntuition}
\end{figure*}
\subsection{Cause of exponentially closing gaps and how to circumvent them}
\label{sec:CauseOfClosingGaps}
Anderson localization has been identified as a major cause of exponentially closing spectral gaps \cite{Altshuler_2010}. At the beginning of the anneal, the instantaneous ground state is close to $|+\rangle^{\otimes N}$, which is delocalized in the computational basis. The first transition the ground state undergoes is from this delocalized state in
the computational basis to a localized one. As the anneal continues, the instantaneous ground state eventually localizes in the computational basis close to a local minimum of the optimization function encoded in $H_Z$. As a consequence, there is an avoided level crossing before reaching the end of the anneal, and it is this transition that typically results in a rapidly closing spectral gap.  An evolving state must tunnel between the local minimum and the global minimum of $H_Z$~\cite{Amin_2009b,Werner_2023}, assuming the system follows the ground state adiabatically. 

One way to bypass the exponentially closing gap in the localized phase of $H(s)$ is to use diabatic transitions to populate the first excited state before $s_{\min}\coloneqq\argmin_{s\in[0,1]} E_1(s) - E_0(s)$. Then by interpolating sufficiently fast through the minimum gap at $s_{\min}$  the ground state can be repopulated. This mechanism is depicted in Figure~\ref{fig:DiagrammeIntuition}(a). It is not obvious though by which mechanism the first excited state can be populated in a controlled manner. We propose LZS interference as a mechanism that can be exploited to this end.
\subsection{Landau-Zener-St{\"u}ckelberg interference}
\label{sec:LZSInterference}
As observed by C{\^o}t{\'e} et al.~\cite{Cote_2023}, the dynamics under their optimized schedule $s(t)$ is restricted to the two lowest instantaneous energy eigenstates. Therefore, it is reasonable to discuss the observed effect in terms of a two-level system. In this effective two-level system, the optimizer identified schedules that reliably populate the first excited state $|E_1(s)\rangle$ before the minimal gap at $s_{\min}$. We claim that the mechanism responsible for this is LZS interference.

LZS interference is closely related to the theory of Landau-Zener transitions~\cite{Landau1932,Zener1932,Majorana1932,Stueckelberg1932}. In these transition, a system sweeps through an avoided level crossing, as depicted for a simple two-level system in Figure \ref{fig:DiagrammeIntuition}(b). With a certain transition probability $r$, given by the Landau-Zener formula, the system transitions between the eigenstates of the two-level system. Let $|E_0 (s) \rangle$ and $|E_1(s) \rangle$ be the instantaneous eigenstates of the two-level system, and let $|\psi(t)\rangle$ be the state of the system. At $t=0$, let us assume $|\psi(t=0)\rangle = |E_0(s=0)\rangle$. Sweeping linearly in time $T$, i.e. $s(t) = t/T$, at the end of the sweep we find
\begin{equation}
    \begin{aligned}
        &|\psi(t=T) \rangle = \\
        &\sqrt{1-r} e^{i\phi_{00}}|E_0(s=1)\rangle + \sqrt{r} e^{i\phi_{01}}|E_1(s=1)\rangle \ ,
    \end{aligned}
\end{equation}
where $\phi_{ij}$ are the phases associated with the transitions from eigenstate $i$ to $j$ as a result of the transition.

If the sweep is repeated, e.g.~by sweeping back from $s=1$ to $s=0$, some part of the population transitions again, leading to
\begin{equation}
    \begin{aligned}
        &|\psi(2T) \rangle =\\
        & \left( (1-r)e^{i(\phi_{00} + \phi_{00}')} + r e^{i(\phi_{01} + \phi_{10}')} \right) |E_0(s=0)\rangle \\
        &+ \left( e^{i(\phi_{01} + \phi_{11}')} + e^{i(\phi_{00} + \phi_{01}')} \right) \sqrt{r(1-r)} |E_1(s=0)\rangle  \ ,
    \end{aligned}
\end{equation}
where the phases $\phi_{ij}$ are the phases due to the first sweep and $\phi_{ij}'$ are the respective phases for the second sweep.
The population in the ground state is the sum of the population that did not transition in either sweep and the part of the population that transitioned in both sweeps, while the coefficient of $|E_1(s=0)\rangle$ is given by the sum of populations transitioning only once.

If between the two sweeps the system sits idle at $s=1$, the part of $|\psi(T) \rangle$ that is in the instantaneous excited state will pick up an additional relative phase $\phi_w = \tau_w \cdot (E_1(1) - E_0(1))$, where $\tau_w$ is the waiting time. Such a schedule is depicted in Figure \ref{fig:DiagrammeIntuition}(c). After the sweep-wait-sweep sequence, the system will be in the state
\begin{equation}
    \begin{aligned}
        &|\psi(2T+\tau_w) \rangle \\
        &= \left( (1-r)e^{i(\phi_{00} + \phi_{00}')} + r e^{i(\phi_{01} + \phi_{10}' + \phi_w)} \right) |E_0(s=0)\rangle \\
        &+ \left( e^{i(\phi_{01} + \phi_{11}' + \phi_w)} + e^{i(\phi_{00} + \phi_{01}')} \right) \sqrt{r(1-r)} |E_1(s=0)\rangle \ .
    \end{aligned}
\end{equation}
If the phases are such that
\begin{equation}
    \phi_{01} + \phi_{10}' +\phi_w = \phi_{00} + \phi_{00}' \mod 2\pi \ ,
\end{equation}
then we have $|\psi(2T+\tau_w) \rangle = e^{i(\phi_{00} + \phi_{00}')}|E_0(s=0)\rangle$, where the system returns completely to the ground state. In contrast, choosing
\begin{equation}
    \phi_{01} + \phi_{10}' +\phi_w = \phi_{00} + \phi_{00}' - \pi \mod 2\pi \ ,
\end{equation}
the system remains in the instantaneous ground state only with probability $(1-2r)^2$. So, if the transition probability in a single sweep is smaller than 0.5, the interference of the populations that transition at each sweep can be used to destroy or amplify the net transition rate. Note that if $r=0.5$, any population transfer can be accomplished with an appropriate choice of $\phi_w$.

In fact, this effect is analogous to an interference experiment with a beam splitter, as seen in Figure \ref{fig:DiagrammeIntuition}(d), where a wave is split and travels on two distinct paths. The paths differ by a length $\Delta l$. After the two beams are re-combined, the beams acquire a relative phase proportional to $\Delta l$, such that, depending on the value of $\Delta l$, the beams constructively or destructively interfere. The analogs to the two paths are the two energy levels, while the path length difference $\Delta l$ is analogous to the waiting time $\tau_w$. In the end, the intensity / level populations are periodic functions of the phase difference.
\subsection{Population transfer mechanism}
\label{sec:PopulationTransferMechanism}
Having introduced the basic intuition of LZS interference and its potential to prepare any population in a two-level system by appropriately choosing the waiting time between ramps, we can give an intuition how LZS interference can be used to circumvent the exponentially closing spectral gap due to a perturbative anti-crossing.

As we have discussed above, the ground state wave function undergoes a transition from a delocalized state in the computational basis to a localized state.  Depending on the model, the transition from a delocalized to a localized ground state happens at a quantum phase transition (PT) from a paramagnet to a ferromagnet. This transition is also associated with a closing gap distinct from the avoided anti-crossing in the localized phase, so if the schedule moves too fast at the transition a certain amount of ground state population transfers to the excited state. If the gap at this delocalized-localized transition closes more slowly than the spectral gap between the localized states, the optimal strategy might be to sweep two or more times through the delocalized-localized transition. LZS interference then allows us to amplify the population in the excited state, thus increasing the success probability if we now sweep rapidly through the localized-localized transition and restore the population to the ground state.

We claim that this is the underlying mechanism of the speedup over linear ramp QA observed in Ref.~\cite{Cote_2023}. We leverage this intuition to design a schedule ansatz and successfully apply it to various optimization tasks. We also show in Appendix~\ref{sec:AppendixLZSDemo} the LZS interference by inserting waiting times $\tau$ to an optimal annealing schedule from Ref.~\cite{Cote_2023} and observe oscillations of the cost function and the eigenstate populations as functions of $\tau$, as illustrated in Figure~\ref{fig:DiagrammeIntuition}(d).
\subsection{Variational schedules and the Landau - Zener - St{\"u}ckelberg interference}
\label{sec:ScheduleIntroduciton}
In this section, we introduce the variational schedules that are the cornerstone of our work. We assume a Hamiltonian $H(s)$ with a single controllable parameter $s$, although this could easily be generalized to multiple parameters.
The variational schedules $s(t)$ are defined as a concatenation of two fundamental operations:
\begin{itemize}
    \item \emph{Linear ramps}. A linear interpolation between two points in the schedule, $(s_a,t_a)$ and $(s_b, t_b)$, is performed:
    \begin{equation}
        s(t)=\frac{s_b-s_a}{t_b-t_a}(t-t_a)+s_a \hspace{0.5cm}\textrm{for } t\in [t_a,t_b) \, .
    \end{equation}
    
    \item \emph{Waiting times}. The linear ramp is followed by a waiting time $\Delta t$ from the time $t_b$ to some new time $t_c= t_b+\Delta t $ in which the Hamiltonian parameter remains constant, i.e.
    \begin{equation}
        s(t)=s_b  \hspace{0.5cm} \textrm{for }t\in [t_b,t_c) \, .
    \end{equation}
\end{itemize}
The role of the \emph{linear ramps} is to change the Hamiltonian parameter under which the system evolves. Depending on the rate of change, this will either cause the system to adiabatically follow the instantaneous eigenstates or stimulate transitions between them. The purpose of the \emph{waiting time} is to introduce a phase difference between the instantaneous ground state and first excited state of the system. The phase shift produced in the waiting time and in previous and subsequent linear ramps can produce LZS interference. This is why we also denote this type of schedules as LZS schedules.

In this work, we focus mainly on the concatenation of three linear ramps interrupted by two waiting times, as depicted in Figure~\ref{fig:DiagrammeIntuition}(e). This class of schedules is characterized by a total of seven parameters: the location of the waiting plateaus $s_1$, $s_2$ and the length of the five time intervals $t_i$. These simple three-sweep schedules already capture the qualitative behavior observed by C{\^o}t{\'e} et al.~\cite{Cote_2023}, as we show below. However, the schedule is simple enough to derive three key results, which we discuss in more detail in their respective sections:
\begin{enumerate}
    \item Given sufficiently many concatenations of ramps and waiting times, any two-level state can be prepared in the low-energy subspace of the instantaneous Hamiltonian. Thus, if the dynamics is sufficiently constrained to the lowest two levels, this schedule class provides a powerful tool to overcome challenges due to closing gaps (see Section \ref{sec:SchedulesAsGates}).
    \item The simplified schedule allows us to gain an intuition about the impact decoherence can have on the working mechanism of the schedule. We verify this intuition later with numerical simulations. This shows that the speedup obtained by the LZS-schedules is a coherent effect (see Section \ref{sec:DecoherenceTheory}).
    \item Since the three-sweep schedules only have a finite number of variational parameters, we can show that the optimal schedule can be found in polynomial time, assuming that a solution to the optimization problem exists with $T = \poly(N)$. This answers a relevant concern regarding speedups obtained by variational methods, namely that potentially any speedup in the quantum computation is negated by an explosion of the time spent to optimize the variational parameters (see Section \ref{sec:EfficientOptimization}).
\end{enumerate}
We do not expect every concatenation of linear ramps and plateaus to be effective, but rather those in which the value of $s$ in consecutive plateaus alternates back and forth. As we will demonstrate, these schedules capture the essence of LZS interference and are the ones that achieve a significant speedup compared to linear schedules on problem instances with perturbative anti-crossings.

We also aim to understand these variational schedules as the analog version, in the annealing context, of the variational quantum algorithms from gate-based computations~\cite{Cerezo2021}. In a digital computer, the parameterization is in terms of the rotation angles of a fixed set of single-qubit and two-qubit gates. In the analog case, the parameterization is in terms of the stopping points and waiting times. In the digital case the complexity of a circuit increases by adding more gates, while in the analog picture the complexity increases by concatenating more ramps and waiting times, as well as by decreasing the rate of change in the linear ramp sections, resulting in longer total anneal times $T$.

It is well known that if we can perform sufficiently many gates on a digital computer, any unitary (and state) can be generated~\cite{Dawson2006}.  More rigorously, we denote a universal gate set as a set of single qubit and two-qubit gates that allows us to approximate any quantum gate to any desired precision.
In a similar spirit, as we argue below, these variational schedules allow us to prepare any superposition of the ground state and the first excited states as long as the low energy subspace is protected from the rest of the spectrum by an energy gap.
\subsection{Variational schedule as quantum gates}
\label{sec:SchedulesAsGates}
By decomposing the variational schedule into distinct pieces, we can understand their respective role. We make the assumption that the evolution is constrained to the two lowest energy levels so that we can describe the system in the instantaneous eigenbasis $|E_i(s) \rangle$ with $i=0,1$. Projecting onto the lowest two instantaneous energy levels using $P_-(s)$, we define the projected Hamiltonian $H_-(s)$ for simplicity
\begin{equation}
    \begin{aligned}
        H_-(s) &= P_-(s) H(s) P_-(s) \\
        &= P_-(s) \mathcal{U}(s) \Lambda(s) \mathcal{U}^\dagger(s) P_-(s) \ ,
    \end{aligned}
\end{equation}
where $\Lambda(s)$ is the diagonal operator with the eigenvalues $E_i(s)$ on the diagonal and $\mathcal{U}(s)$ is the unitary given by the instantaneous eigenbasis. Transforming to  the instantaneous eigenbasis, the Hamiltonian $H_-(s)$ transforms to
\begin{equation}
    H_-(s) \rightarrow \tilde{H}_-(s) = \Lambda(s) - i\mathcal{U}^\dagger(s) \dot{\mathcal{U}}(s) \ ,
\end{equation}
where $\dot{\mathcal{U}}(s) \equiv \frac{d}{ds}\mathcal{U}(s)$.

In the following, we consider the unitary in the low-energy subspace of the instantaneous eigenbasis
\begin{equation}
    \tilde{U}_- = \mathcal{T} \exp \left( -i\int_0^T \tilde{H}_-(t) dt \right) \ .
\end{equation}
As for any unitary in SU(2), $\tilde{U}_-$ can can be decomposed as
\begin{equation}
    \tilde{U}_- = \exp \left( -i \phi_r Z \right) \begin{pmatrix}
        \sqrt{1-r} & \sqrt{r} \\ -\sqrt{r} & \sqrt{1-r}
    \end{pmatrix} \exp \left( -i \theta_r Z \right) \ .
    \label{eq:SU2Gate}
\end{equation}
The transition rate $0 \leq r \leq 1$ and the phases $\phi_r$, $\theta_r$ depend on the schedule, or the piece of the schedule, that is integrated. For example, a waiting plateau segment contributes with $r=0$, while a linear ramp segment contributes with $r > 0$ . By introducing waiting intervals before and after the ramp, we can effectively control the phases
\begin{equation}
    \begin{aligned}
        &\tilde{U}_- = U(\phi, \theta) \\
        = &\exp \left( -i (\phi - \phi_r) Z \right) \begin{pmatrix}
        \sqrt{1-r} & \sqrt{r} \\ -\sqrt{r} & \sqrt{1-r}
    \end{pmatrix} \\
    &\times \exp \left( -i (\theta - \theta_r) Z \right) \\
    = & \exp \left( -i (\phi - \phi_r) Z \right) \begin{pmatrix}
        \cos \lambda & \sin \lambda \\ -\sin \lambda & \cos \lambda
    \end{pmatrix} \\
    &\times \exp \left( -i (\theta - \theta_r) Z \right) \ ,
    \end{aligned}
    \label{eq:WaitPulseWaitGate}
\end{equation}
where for the second equality we expressed the transition rates $r$ by a rotation $\sin \lambda$, i.e. a rotation around the y-axis by an angle $2\lambda$ on the Bloch sphere. The transition rate depends on the interval that is swept. In the context of diabatic annealing, in order to stimulate a transition with sufficient probability $r$ certain conditions must be fulfilled, such as a slowly vanishing gap along the anneal. We discuss this in more detail in Section \ref{sec:OptimalityConditions}.

Note that we choose a decomposition in Eq.~\eqref{eq:WaitPulseWaitGate} such that $\sin \lambda \geq 0$, so we can assume $0 \leq \lambda \leq \pi / 2$. While it is clear that $\tilde{U}_{-}$ is a universal $\text{SU(2)}$-gate if $\lambda$ could be chosen freely, here we will assume $\lambda$ to be fixed. This reflects that a given annealing instance dictates the ramp speeds, and therefore the transition rates, via its spectral structure at least to some degree. The phase parameters controlled via waiting times, however, are easy to control.

It is instructive to analyze what happens if we concatenate gates of the form of $U(\phi, \theta)$. Consider a quantum state $|\psi_n \rangle = \sqrt{1-P_1^{(n)}} |E_0(s)\rangle + e^{i\phi_n}\sqrt{P_1^{(n)}} |E_1(s) \rangle$ after $n$ applications of the gate, where $P_1^{(n)}$ is the population of the excited state and $\phi_n$ is the relative phase. Applying one more iteration of $U(\phi, \theta)$, we obtain the state
\begin{equation}
    |\psi_{n+1} \rangle = U(\phi, \theta) |\psi_n \rangle \ .
\end{equation}
We can calculate the excited state population $P_1^{(n+1)}$ as
\begin{equation}
    \begin{aligned}
        P_1^{(n+1)} =& | \langle E_1 | \psi_{n+1} \rangle |^2 \\
        =& | \langle E_1 | U(\phi, \theta) | \psi_{n} \rangle |^2 \\
        =& (1-r)P_1^{(n)} + r\left(1-P_1^{(n)} \right) \\
        &- 2 \cos (\tilde{\phi}) \sqrt{r(1-r)} \sqrt{P_1^{(n)} \left(1-P_1^{(n)} \right)} \ ,
    \end{aligned}
    \label{eq:PopulationTransferStep}
\end{equation}
where we have absorbed all the phases into one controllable phase $\tilde{\phi}$. Since $\tilde{\phi}$ is a free parameter, Eq.~\eqref{eq:PopulationTransferStep} gives a range of possible excited state populations $P_1^{(n+1)}$ after the next $U(\phi, \theta)$-gate for every transition rate $r$. These ranges can be plotted as in Figure \ref{fig:AccessiblePopulations}.

The expression in Eq.~\eqref{eq:PopulationTransferStep} and its associated plot in Figure \ref{fig:AccessiblePopulations} allow for insightful interpretations. Firstly, the accessible populations indicated by the colored sets in Figure~\ref{fig:AccessiblePopulations} form ellipses in the $P_1^{(n)}-P_1^{(n+1)}$-plane, and their major axes coincide with the line $P_1^{(n)} = P_1^{(n+1)}$, if $r < 0.5$, or $P_1^{(n)} = 1-P_1^{(n+1)}$, if $r>0.5$. If $r=0.5$, the reachable set forms a circle. Therefore, the expression in Eq.~\eqref{eq:PopulationTransferStep} is invariant under exchange of $P_1^{(n)}$ and $P_1^{(n+1)}$, which is a consequence of the invertibility of $U(\phi, \theta)$. For a given $P_1^{(n)}$ and an appropriate choice of $\tilde{\phi}$, $P_1^{(n+1)}$ is given by any point in the ellipse on a vertical line above $P_1^{(n)}$ in Figure~\ref{fig:AccessiblePopulations}. Due to the invariance, the accessible $P_1^{(n+2)}$ are now all points in the ellipse on a horizontal line at $P_1^{(n+1)}$. Continuing this procedure, it becomes apparent that by concatenating the gates $U(\phi, \theta)$, we can traverse the ellipse in perpendicular lines, thus reaching any point, i.e. any desired population of the excited state. This is illustrated in Figure~\ref{fig:AccessiblePopulations}(a) for $r=0.3$, where we show how an initially small excited state population $P_1^{(1)} \approx 0.05$ can be pumped to $P_1^{(4)} \approx 1$ after three LZS-gates.

By choosing the parameter $\phi$ of the last $U(\phi, \theta)$-gate and $\theta$ of the first, we can for any given state create both any desired population and phase, indicating that a sufficient concatenation of $U(\phi, \theta)$ is universal in SU(2). This graphical intuition can be formalized in the following Proposition:
\begin{proposition}
    \label{lemma:UniversalSU2}
    Let $U(\phi, \theta)$ as in Eq. \eqref{eq:WaitPulseWaitGate} and $\lambda_{\text{eff}} = \min\left\{\lambda, \frac{\pi}{2} - \lambda \right\}$. For any $V \in \text{SU}(2)$, there is a $N_0 \in \mathbb{N}$ with
    \begin{equation}
        N_0 = 2^{\left\lceil \log_2 \frac{\pi}{2\lambda_{\text{eff}}} \right\rceil} \ ,
    \end{equation}
    such that for any $N\geq N_0$ there is a sequence of angles $(\phi_n, \theta_n)$ with
    \begin{equation}
        \prod_{n=1}^N U(\phi_n, \theta_n) = V \ .
    \end{equation}
\end{proposition}
The proof of Proposition~\ref{lemma:UniversalSU2} can be found in Appendix~\ref{sec:ProofUniversalSU2}. Note that in the limit of small transition rates $r$, we find $N_0 = 2^{\lceil \log_2 \frac{\pi}{2\lambda} \rceil} = \mathcal{O} \left( r^{-1/2} \right)$. The key takeaway is that if sufficiently many sweeps and waiting times are concatenated, any state of the two-level system can be prepared. In the case where the transition rate $r= \sin^2 \lambda =\frac{1}{2}$, we have $N_0 = 2$, i.e. two sweeps are sufficient to reach universality. 

Although the exact value of $r$ depends on the individual Hamiltonian, an argument can be made that if there is a closing gap elsewhere along the anneal, which closes slower than the minimal gap as the system size grows, then $r=\frac{1}{2}$ can be approximated by slowing the ramps down and remaining close to adiabatic. In this scenario, the required anneal time for the LZS-schedules would be governed by the scaling of the larger gap.
A class of instances where we would expect this type of gap profile to occur are the instances discussed in \cite{Amin_2009a, Werner_2023}. Here, there is a slowly closing gap due to a quantum phase transition from a phase with a delocalized ground state in the computational basis to a phase with a localized ground state, followed by another avoided level crossing between two localized eigenstates.

There is a second takeaway from Eq.~\eqref{eq:PopulationTransferStep}. The first two terms are essentially the transitions one would observe in a classical Markov model with the transition matrix
\begin{equation}
    M = \begin{pmatrix}
        1-r & r \\ r & 1-r
    \end{pmatrix} \ .
\end{equation}
Only the last term contains a dependence on the phases $\tilde{\phi}$ and can therefore be interpreted as the ``quantum part" of the population transfer. Interestingly, the populations of a classical model evolving under the matrix $M$ will equilibrate to $P_1^{(n)} = 1-P_1^{(n)} = 0.5$. In order to pump up the population (or, equally, in order to decrease it) the coherent effect of the phases is needed. If coherence is compromised, we expect there to be an effect on the accessible populations, as we discuss in the next section.

We make one final remark. The number of sweeps required to reach any two-level state according to Proposition~\ref{lemma:UniversalSU2} for small transition probabilities $r$ scales as $\mathcal{O}(1/\sqrt{r})$. Without exploiting the coherent phase interferences to populate the desired state, e.g. by ramping once and measuring the state at the end, the system would transition to the excited state with probability $r$, implying that one would require $\mathcal{O}(1/r)$ repetitions of the ramp to observe the excited state. Thus, the exploitation of coherent effects can deliver a quadratic speedup in this simple preparation of the excited state, similar to Grover search~\cite{Roland_2002}.

The idea of exploiting LZS interference for population inversion has been investigated in the quantum control literature before~\cite{Teranishi_1998, Teranishi_1999, Hsiao_2025}. Specifically, Proposition~\ref{lemma:UniversalSU2} formally proves that universal control is possible, which has recently been anticipated by Hsiao et al.~\cite{Hsiao_2025}.

\subsection{The impact of decoherence}
\label{sec:DecoherenceTheory}
To get an intuition of what happens to the coherent transfer mechanism that we have discussed above, it is instructive to decompose the gate $U(\phi, \theta)$ into its population transfer part $U_r$ and the phase parts $U_\theta$, $U_\phi$. Instead of state vectors $|\psi_n \rangle$ and $|\psi_{n+1} \rangle$, we now consider density matrices $\rho_n$ and $\rho_{n+1}$. The unitaries are applied consecutively, but between the respective unitary channels $D_r$ and $D_\theta$ (the last phase-gate $U_\phi$ can be omitted for this consideration, as it does not impact the population) we insert a decoherence channel $D_\mu$, where $\mu$ is the error rate, i.e.
\begin{equation}
    \rho_{n+1} = D_r \left( D_\mu \left( D_\theta \left(\rho_n \right) \right) \right) \ .
    \label{eq:SimpleNoiseModel}
\end{equation}
We make the simplifying assumption that $\rho_n = |\psi_n \rangle \langle \psi_n|$ is pure. We consider the three major noise channels, dephasing noise, amplitude damping and depolarizing noise. For each, the excited state population can be determined and is shown in Table \ref{tab:NoiseChannelSummary}.

Again, we obtain a range of possible populations $P_1^{(n+1)}$ for different values of $\tilde{\phi}$, which we plot in Figure \ref{fig:AccessiblePopulations}(b), where we chose an example transition rate of $r=0.3$. We can see that different noise channels distort the ellipses in different ways, such that the invariance under inversion of $P_1^{(n)}$ and $P_1^{(n+1)}$ is lost. This corresponds to the loss of invertibility of the noisy quantum gate.

Additionally, Figure \ref{fig:AccessiblePopulations}(b) shows that the accessible populations do not cover the whole range from 0 to 1. This means that the range of population $P_1^{(n+1)}$ of the first excited state that can be reached by the noisy gate is restricted, no matter what the previous population $P_1^{(n)}$ was. Therefore, noise impedes the mechanism of control, which is used by our proposed schedule class to reach the excited state. Based on this argument, we would expect that in an open quantum system, there is an upper bound on the population inversion that can be achieved. We verify this numerically later.

\begin{figure*}
    \centering
    \begin{tikzpicture}
        \node[anchor=north west, inner sep=0] (fig) at (0,0)
            {\includegraphics[width=0.85\linewidth]{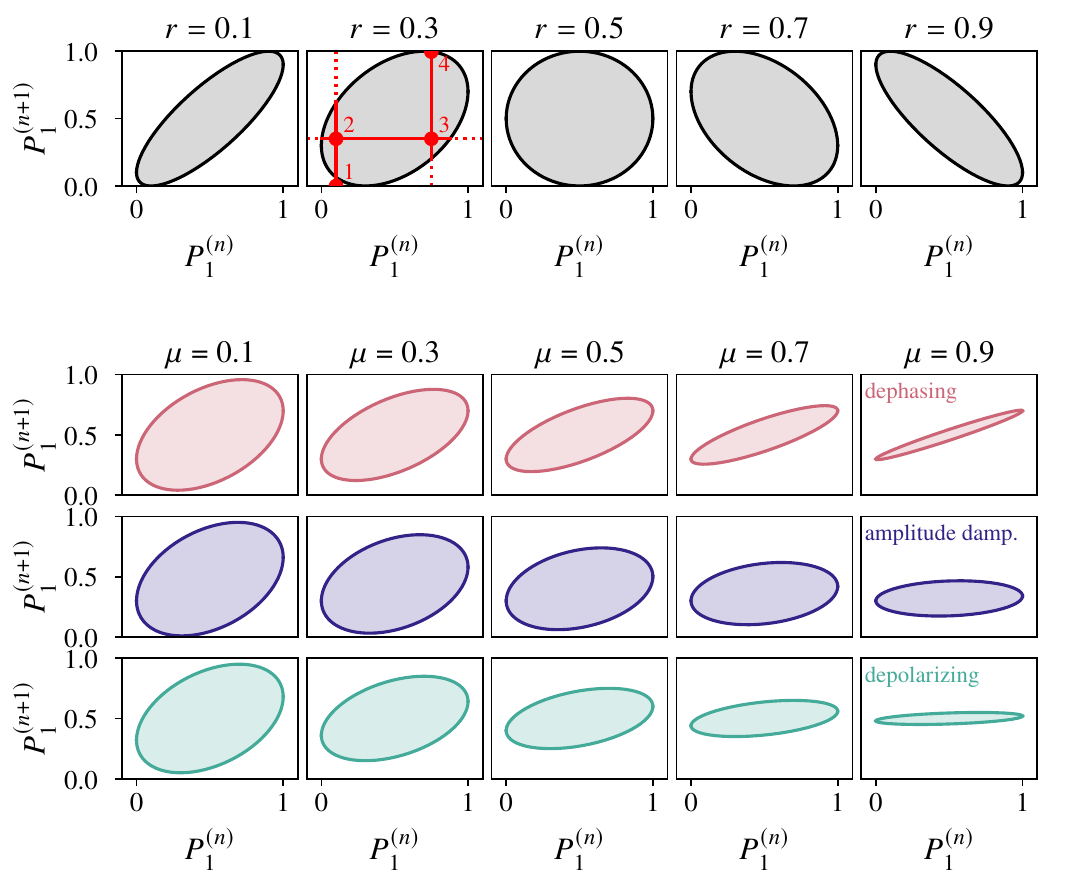}};
        \node[font=\large] at (fig.north west) {(a)};
        \node[font=\large] at ($(fig.north west)!0.35!(fig.south west)$) {(b)};
    \end{tikzpicture}
    \caption{(a): Visualization of accessible excited state populations $P_1^{(n+1)}$ as per Eq.~\eqref{eq:PopulationTransferStep} of the first excited state after $n+1$ LZS-sweeps assuming the initial population $P_1^{(n)}$. For each $P_1^{(n)}$, there is a set of accessible $P_1^{(n+1)}$, which are determined by the choice of phase $\tilde{\phi}$, shown here as the shaded ellipses. For $r=0.3$, we indicate how a sequence of sweeps with appropriately chosen phases between them can pump the excited state population. The red solid lines indicate which populations are accessible after the next LZS-gate, while the red points indicate the populations at each step. (b): Same accessible populations as in (a) under the influence of different types of noise for a fixed transition rate $r=0.3$. The respective expressions are listed in Table \ref{tab:NoiseChannelSummary}.}
    \label{fig:AccessiblePopulations}
\end{figure*}

\begin{table*}
    \centering
    \begin{tabular}{|c|c|}
        \hline
        Noise channel $D_\mu$ & \rule{0pt}{20pt} $P_1^{(n+1)} = \langle E_1 | \rho_{n+1} | E_1 \rangle = \langle E_1 | D_r \left(D_\mu \left(D_{\theta} \left(\rho_{n}\right)\right)\right) | E_1 \rangle $  \\[6pt]
        \hline \hline
         Dephasing & \rule{0pt}{20pt} $(1-r)P_1^{(n)} + r\left(1-P_1^{(n)}\right) - 2\cos(\tilde{\phi}) (1-\mu) \sqrt{r(1-r)P_1^{(n)} \left(1-P_1^{(n)}\right)}$\\[6pt]
         \hline
         Amplitude damping & \rule{0pt}{20pt} $r\left(1-P_1^{(n)} \right) + \left[ r\mu + (1-r)(1-\mu) \right] P_1^{(n)} - 2\cos(\tilde{\phi})\sqrt{r(1-r)P_1^{(n)} \left(1-P_1^{(n)} \right)(1-\mu)}$\\[6pt]
         \hline
         Depolarizing & \rule{0pt}{20pt} $(1-\mu) \left[ (1-r)P_1^{(n)} + r\left(1-P_1^{(n)} \right) - 2\cos(\tilde{\phi})\sqrt{r(1-r) P_1^{(n)} \left( 1-P_1^{(n)} \right)} \right] + \frac{\mu}{2}$\\[6pt]
         \hline
    \end{tabular}
    \caption{Population $P_1^{(n+1)}$ given $P_1^{(n)}$ after another application of a LZS-gate Eq.~\eqref{eq:WaitPulseWaitGate} for different decoherence channels modeled as in Eq. \eqref{eq:SimpleNoiseModel}. The resulting populations are depicted in Figure \ref{fig:AccessiblePopulations}(b).}
    \label{tab:NoiseChannelSummary}
\end{table*}

\subsection{Efficient optimization is possible}
\label{sec:EfficientOptimization}
A natural question that arises in the context of variational algorithms is the question of optimization efficiency. While an algorithm using optimized parameters can show substantial speedups on certain problems, it is often unclear if the observed speedup is genuine or an artifact of out-sourcing computational effort to the classical optimizer. For example, if the observed speedup is exponential, it can happen that the resulting parameter-optimization task is NP-hard.

Often times the exponential blow-up of required optimization resources is associated with a growing parameter space. In contrast, since our proposed class of schedules requires only a fixed number of parameters, we are able to show that the optimization can be done in polynomial time.
\begin{proposition} (informal)
    Let $\| H_X \|, \|H_Z \| = \poly(N)$. Assume that the optimal LZS-schedule $s(t)$ prepares the ground state of the target Hamiltonian with sufficient probability in polynomial time, $T=\poly(N)$. Then the optimal parameters can be determined by a classical optimizer in a polynomial number of calls to the quantum annealer and each annealing run takes at most $\poly(N)$ time.
    \label{prop:efficiency}
\end{proposition}
The proof is given in Appendix \ref{sec:EfficiencyProof}. The key caveat here is that the proposition only guarantees finding an efficient schedule in polynomial time if it exists. From the assumptions of Proposition \ref{prop:efficiency}, it follows that the optimal parameters are at most polynomially far away from the origin and the cost function is only inversely polynomially sensitive. This allows us to show that the optimal parameters can be approximated up to an error $\epsilon$ by a brute-force grid search algorithm. This shows that optimization is possible in polynomial time, but we do not claim any optimality of the brute-force algorithm used in the proof. Indeed, the polynomial runtime bound we prove in Appendix \ref{sec:EfficiencyProof} is of high order for practical problems.

While we can show that a polynomial-time solution can also be found by an optimizer in polynomial time, the next question would be under which conditions the LZS-schedule is optimal. In the next section, we will give a broad answer to this question.

\subsection{Conditions for optimality of LZS-schedules}
\label{sec:OptimalityConditions}
Here we give some intuition about when we expect our ansatz to be able to solve an annealing problem faster than adiabatic evolution. We assume that the time evolution can be well described by considering only the instantaneous ground and first excited states. Let $r_{\text{leak}}$ be the leakage rate, i.e. the probability of the system transitioning out of the ground and first excited state subspace. As before, $r$ is the transition rate between ground and first excited state. We require $r_{\text{leak}} \ll r$. This can be guaranteed for example if the low energy subspace is energetically separated from the rest of the spectrum by an energy gap that scales at least inverse-polynomially (in the system size) and the ramp times scale polynomially. Alternatively, the low energy subspace can also be isolated via selection rules, as per the numerator of the adiabatic condition (Eq. \eqref{eq:AdiabaticTheorem}). We also require there to be an interval in the spectrum where transitions between the ground and first excited state can be stimulated at a sufficient rate. This can occur at avoided level crossings where the gap closes polynomially but at a faster rate than the separation between the first excited state and higher excitations or, more generally, on intervals where the eigenstates undergo sufficient change. We call this stretch of the spectrum the ``control interval,'' where sweeping through e.g.~a closing gap will stimulate sufficiently large transitions within the low-energy subspace.  

We now make use of a result by Callision et al.~\cite{Callison_2021} to characterize the control interval more precisely. Let $\ket{\psi_0^-}$ and $\ket{\psi_1^-}$ be the instantaneous ground and excited states on the left side of the control interval $s\in [s^-, s^+]$, and let $\ket{\psi_0^+}$ and $\ket{\psi_1^+}$ be the eigenstates on the right. The defining property of the control interval is that ramps in polynomial time stimulate transitions between ground and first excited state at a rate $r$, while at the same time the dynamics is constrained to the lowest two energy levels. Let us define $H^+$ and $H^-$ to be the Hamiltonians on the right and left limits of the interval respectively. Callison et al.~argued that if the system is initially in the ground state of $H^-$, i.e. $\ket{\psi(t=0)} = \ket{\psi_0^-}$, then after a monotonic ramp to $H^+$ in time $T$ we have
\begin{equation}
    \bra{\psi(t=T)} H^+ \ket{\psi(t=T)} \leq \bra{\psi_0^-} H^+ \ket{\psi_0^-} \ .
\end{equation}
In a two-level system, this directly implies an upper bound on the transition rate $r$:
\begin{equation}
    r = \left| \langle \psi(t=T) | \psi_1^+ \rangle \right|^2 \leq \left| \langle \psi_0^-| \psi_1^+ \rangle \right|^2 \ ,
\end{equation}
where the bound is saturated in the limit of the sudden approximation. This shows that specific conditions need to be met and that not any interval can be used to stimulate transitions between ground and first excited state. As an example, if on a given interval the eigenstates barely change, then we have $|\langle \psi_0^- | \psi_1^+ \rangle|^2 \approx 0$ and there will be no population transfer between ground and first excited state. If, on the other hand, the control interval covers an avoided crossing, then we can describe this scenario locally as a Landau-Zener model, where $\ket{\psi_0^-} \approx \ket{\psi_1^+}$ and vice versa, as depicted in Figure \ref{fig:DiagrammeIntuition}(b). In this case we find $r \leq |\langle \psi_0^- | \psi_1^+\rangle|^2 \approx 1$ and potentially all the population can transition to the excited state in a single sweep, if the sweep is fast enough. However, since the sweep needs to be slow enough such that $r_{\text{leak}} \ll 1$, a slower ramp and thus a smaller $r$ might be necessary. As a third example, if the effective low-energy Hamiltonian can be described in such a way that $H^- \sim X$ and $H^+ \sim Z$, then $r \leq 1/2$. In this case, more than one sweep is necessary for population inversion. As we see in Section~\ref{sec:IsingRingSimulation}, these conditions, $r \approx 1/2$, $r_{\text{leak}} \approx 0$, can be met at quantum phase transitions.

In the examples we consider numerically in Section \ref{sec:Numerics}, the optimal solution often is populating the first excited state in order to avoid adiabatically evolving through the exponentially closing gap. However, inverting the population might not be sufficient. For example, there could be more than one location along the anneal where populations mix non-trivially and their phases interfere. The existence of a control interval allows in principle to overcome a broader class of challenges, as we can see from Proposition~\ref{lemma:UniversalSU2}.

Applying Proposition~\ref{lemma:UniversalSU2} allows us to conclude that in a small number of sweeps we can realize an arbitrary unitary on the low-energy subspace by sweeping through the control interval. Considering that the time evolution before and after the interval will act as some unitary on this subspace, assuming leakage to higher excitations is suppressed adiabatically, the proposed schedule class can prepare whichever state that will evolve into the ground state at $s=1$. Note that if the ramp can be slow enough such that $r = 0.5$ when sweeping through the control interval, then according to Proposition \ref{lemma:UniversalSU2}, $N_0 = 2$ sweeps are sufficient to prepare an arbitrary state in the two-level subspace. If there is a control interval and the dynamics is sufficiently constrained to the instantaneous ground and first excited state, coherent obstructions to the anneal such as the spectral gap between the two lowest eigenstates closing exponentially along the interpolation one or more times can be overcome.

Note that in the numerical examples we consider in Section \ref{sec:Numerics} we focus on the ansatz with three sweeps and two plateaus, as shown in Figure \ref{fig:DiagrammeIntuition}(e). Depending on the problem at hand, ans\"atze with more sweeps and plateaus may be required in order to populate the excited state or to reach universality. If the number of sweeps and plateaus remains constant with the system size, then these ans\"atze inherit the efficient optimization we show in Proposition~\ref{prop:efficiency}, as the proof also applies for a larger, but fixed, number of parameters.

We show in Figure \ref{fig:ControllableIntervalIlustration} an illustration of a spectrum where we expect our ansatz class to provide a substantial speedup over adiabatic evolution. We show the spectral gap between the first excited and ground state $\gamma_{01}$ and the second excited and ground state $\gamma_{02}$ along the interpolation parameter $s$. We annotate the control interval in green, where $\gamma_{01}$ is still inversely polynomially large, and the second excited state is sufficiently far away in energy. Along the entire anneal, the ground and first excited state are separated from the rest of the spectrum, except at the beginning of the anneal, where the system is in the ground state and the second and first excited states may approach closely. This is typically the case if e.g.~the transverse-field Hamiltonian (Eq.~\eqref{eq:DefHx}) is the initial Hamiltonian. In this case, an LZS-schedule $s(t)$ may not enter the red shaded area. While the numerical examples discussed below display one exponentially closing gap, in this example we show two since, if universality on the ground and first excited subspace is achieved, a state that evolves into the target ground state can be prepared on the control interval independent of closing $\gamma_{01}$ in other locations along the anneal.
\begin{figure}
    \centering
    \includegraphics[scale=0.625]{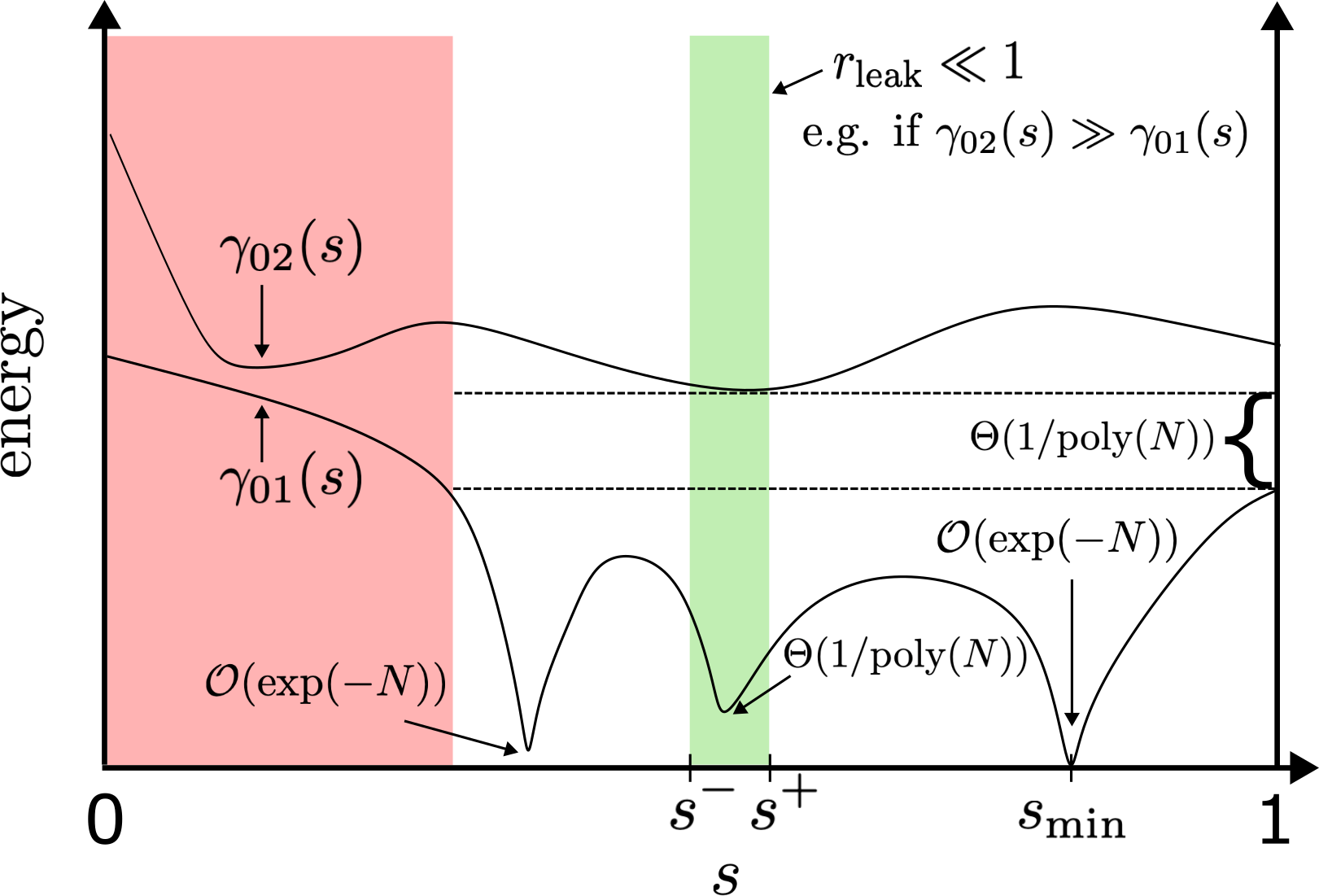}
    \caption{Illustration of a spectrum where we expect our ansatz to provide a speedup over adiabatic evolution. We show the spectral gap between the first excited and ground state $\gamma_{01}$ and the second excited and ground state $\gamma_{02}$ along the interpolation parameter $s$. Green shaded area: control interval, where $\gamma_{01} = \Theta(1/\poly(N)) \ll \gamma_{02}$, such that a transition rate $r$ can be obtained to prepare an arbitrary state. Red shaded area: part of the annealing spectrum where the first and second excited state may not be sufficiently separated. However, since in the beginning of the anneal, the system will be in the ground state, it will be protected from transitioning to the higher levels, while at later times the optimal schedule $s(t)$ could avoid this region.}
    \label{fig:ControllableIntervalIlustration}
\end{figure}

We note that the conditions discussed here are rather strict and may only be found to hold in artificially-constructed or rare problem instances. 
We also do not make a statement on what these conditions imply for the computational hardness of the instances where they hold. On the other hand, while we expect these condition to be sufficient, they may not be necessary, and our ansatz might be more broadly applicable. In this sense, we consider the theoretical discussion here to be a foundation for an ansatz that could be applied heuristically to a larger class of problems.
\section{Numerical results}
\label{sec:Numerics}
The class of annealing schedules we introduced above has seven parameters: two locations $s_{1,2} \in [0, 1]$ of the waiting times along the anneal and five time intervals $t_{1...5} \geq 0$, as depicted in Figure \ref{fig:DiagrammeIntuition}(e). We perform an outer optimization loop over the total time $T = \sum_{i=1}^5 t_i$ and optimize the remaining six free parameters in an inner loop for a fixed $T$. For the inner-loop optimization we use APOSMM \cite{Larson_2018} with COBYLA \cite{Powell_1994} as a gradient-free local optimizer. 

\subsection{Reproduction of result by C{\^o}t{\'e} et al.~on the frustrated Ising ring}
\label{sec:IsingRingSimulation}
We first test our proposed simplified schedule by applying it to the frustrated Ising ring, as was investigated in Refs.~\cite{Cote_2023, Wang_2025}. The target Hamiltonian $H_Z$ for the Ising ring is shown in Figure \ref{fig:FrustratedIsingRingProblem}(a). For a particular choice of the $ZZ$-coupling strengths, the spectrum of the full Hamiltonian $H(s)$ displays a perturbative anti-crossing at $s\approx0.899$. The spectral gap is shown for various system sizes in Figure~\ref{fig:FrustratedIsingRingProblem}(b). The spectral gap closes exponentially as a function of the system size at the perturbative anti-crossing. However, as the system size increases, the gap also appears to close at  $s\approx0.514$ with increasing system size. This closing of the gap is associated with the paramagnetic to ferromagnetic phase transition.
\begin{figure}
    \centering
    \begin{tikzpicture}
        \node[anchor=north west, inner sep=0] (a) at (-0.0cm,0)
            {\includegraphics[width=0.95\linewidth]{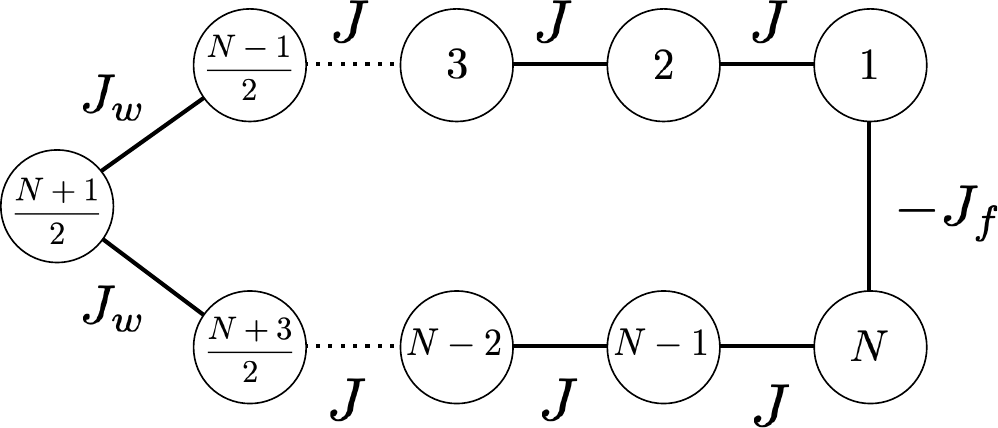}};
        \node[anchor=north west, inner sep=0] (b) at (0, -4.25cm)
            {\includegraphics[width=0.95\linewidth]{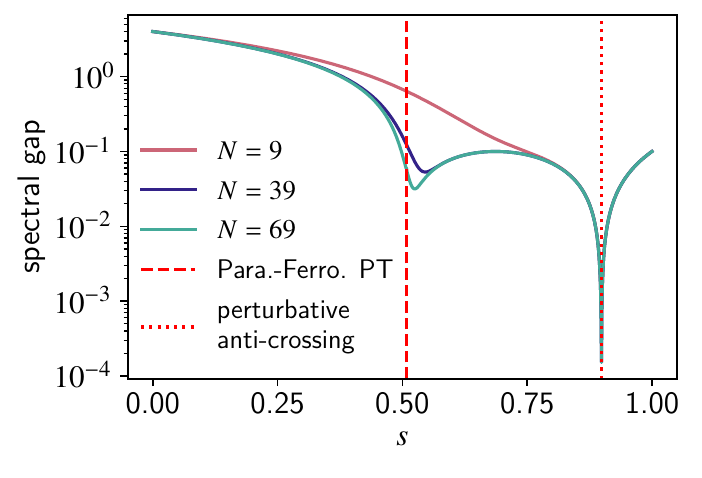}};
        \node[font=\large] at ($(a.north west) + (0.0cm,0)$) {(a)};
        \node[font=\large] at ($(b.north west) + (0, 0)$) {(b)};
    \end{tikzpicture}
    \caption{(a) Connectivity graph of the target Hamiltonian for the frustrated Ising ring. (b) Spectral gap for the Hamiltonian $H(s)$ associated with (a) for coupling values $J=1$, $J_w = 0.5$ and $J_f = 0.45$. Also shown are the approximate positions of the para- to ferromagnet transition at $s\approx 0.514$ and the perturbative anti-crossing at $s \approx 0.899$. }
    \label{fig:FrustratedIsingRingProblem}
\end{figure}

We optimize the schedule introduced above for various system sizes $N$. For each $N$, we employ the bisection search method used in Ref.~\cite{Cote_2023} to find the optimal value of $T$, while for the simulation of the quantum system we use the Nambu formalism, as discussed in Ref.~\cite{Wang_2025}. The Nambu formalism makes use of the fact that the frustrated Ising ring can be described by a system of non-interacting fermions, so the time-evolution only needs to be computed in a $2N$-dimensional vector space, instead of an exponentially large one. This allows us to simulate comparatively large system sizes of the frustrated Ising ring.

At each step of the bisection search method, the algorithm fixes a time $T$ and runs the optimization of the inner loop. The inner-loop optimization is considered successful if the optimizer is able to find a schedule that prepares a state $|\psi(T; \theta) \rangle$ such that
\begin{equation}
    \langle \psi(T; \theta) | H_Z | \psi(T; \theta) \rangle \leq E_0(s=1) + c \cdot \Delta E(s=1) \ ,
    \label{eq:IsingRingCriterium}
\end{equation}
where $E_0(s=1)$ and $\Delta E(s=1)$ are the ground state energy and spectral gap of $H_Z$ respectively, and $0 < c < 1$ is a small number that determines the accuracy of the ground state preparation, i.e.\@ it tunes the difficulty of the optimization task. If the inner-loop optimization is successful, the bisection algorithm continues the search for the optimal time below the current $T$. If the inner-loop optimizer terminates after $N_{\text{maxiter}} = 10^4$ calls to the annealer and the optimizer does not find a schedule that satisfies Eq.~\eqref{eq:IsingRingCriterium}, the inner-loop optimization is unsuccessful, and the search continues with a larger $T$.

The optimized $T$ as a function of system size $N$ is shown in Figure \ref{fig:FrustratedIsingRingResult}(a). Here we use $c=0.1$ and $c=0.5$, in accordance with Refs. \cite{Cote_2023, Wang_2025}. We fit the data with a power law in order to extract the scaling exponent, and we observe that our ansatz results in a sub-quadratic scaling of the required $T$, compared to the quadratic scaling of Ref.~\cite{Cote_2023}. With increasing difficulty, set by the parameter $c=0.1$, the exponent of the scaling is not substantially changed, but the required runtime increases by an $N$-independent factor.

\begin{figure}
    \centering
    \begin{tikzpicture}
        \node[anchor=north west, inner sep=0] (fig) at (0,0)
            {\includegraphics[width=0.92\linewidth]{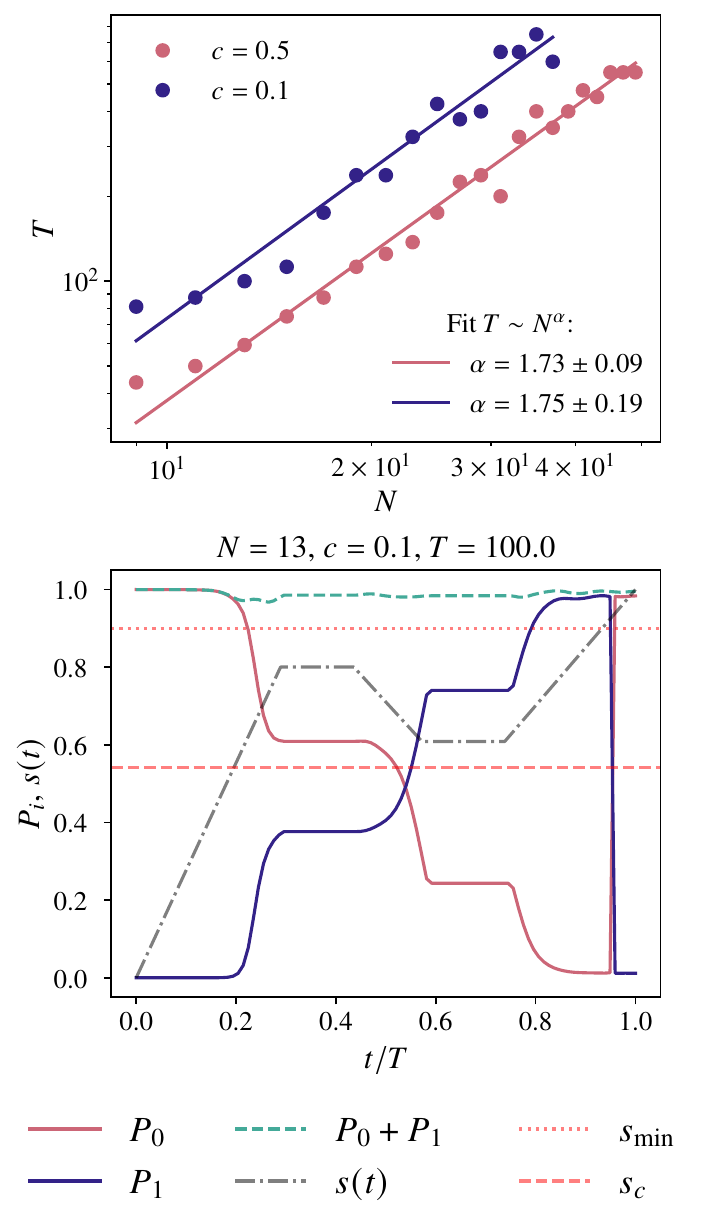}};
        \node[font=\large] at (fig.north west) {(a)};
        \node[font=\large] at ($(fig.north west)!0.45!(fig.south west)$) {(b)};
    \end{tikzpicture}
    \caption{(a) Optimum total time $T$ as found by running the optimization algorithm as in Ref.~\cite{Cote_2023} to prepare the ground state with at least $c=0.5, 0.1$ overlap.  Solid lines are the best fit curves of the form $b N^{\alpha}$.
    (b) The optimized schedule $s(t)$ for a small system size $N=13$, where state vector simulation is feasible. We show the populations of the instantaneous ground state $P_0$ and the first excited state $P_1$ along the anneal as well as their sum $P_0 + P_1$. 
    The horizontal dashed and dotted lines indicate the critical values $s_c$ where the Ising ring transitions from the paramagnetic to the ferromagnetic phase and the location of the minimal spectral gap due to the perturbative anti-crossing respectively . Note that for small system sizes, the small spectral gap associated with the PT is slightly above $s_c$, as shown in Figure \ref{fig:FrustratedIsingRingProblem}(b).}
    \label{fig:FrustratedIsingRingResult}
\end{figure}
For small system sizes, we further analyze the dynamics of the anneal by exact state vector simulation using QuTip \cite{qutip5}. We calculate the occupation of the instantaneous energy levels
\begin{equation}
    \begin{aligned}
        P_i(t) = |\langle E_i(t)| \psi(t) \rangle |^2 \ ,
    \end{aligned}
\end{equation}
for $i=0,1$, as shown in Figure \ref{fig:FrustratedIsingRingResult}(b). Here, we show the optimized schedule and the eigenstate populations $P_0$ and $P_1$ of the ground and first excited states, respectively, for $N=13$ and $c=0.1$ as a representative case. Additional data for $N=9, 11, 13$ and $c=0.5, 0.1$ can be found in Appendix~\ref{sec:AppendixIsingRingSchedules}. We see that the dynamics behaves in a similar way to that reported by C{\^o}t{\'e} et al. We see that the schedule populates the first excited state by sweeping across the paramagnetic transition, and then restores the population to the ground state when it crosses the perturbative crossing. The dynamics are also well restricted to the lowest two instantaneous energy levels, as can be seen by the sum $P_0 + P_1 \approx 1$ for most of the anneal.
Given that our ansatz is inspired by the intuition that the system dynamics is restricted to the ground and first excited states, we believe we have correctly identified LZS interference as the working mechanism behind the optimized annealing schedules observed in Ref.~\cite{Cote_2023}.

We believe that the smaller exponent using our ansatz compared to the quadratic scaling of $T$ found by C{\^o}t{\'e} et al.\@ can be explained by how the optimization of the schedule parameters was performed. For each $T$, C{\^o}t{\'e} et al.\@ used a fixed compute budget for the inner loop optimization of the schedule parameters. Once the budget is exhausted, the algorithm assumes that at the given time budget $T$, the ground state cannot be prepared with sufficient overlap, and $T$ is increased. We also choose this compute budget as a fixed value that also does not increase with system size. Thus, the number of bisection steps only increases with system size $N$ due to the bisection search having to search a larger interval. Bisection search finds the optimal $T$ in a number of bisection steps logarithmic in the optimal $T$, hence we conclude that the number of bisection steps in this experiment only increases logarithmically with system size for $T=\poly(N)$. However, under certain circumstances, we expect this strategy to skew the asymptotic scaling. As $N$ increases, we expect the inner-loop optimizer to fail to find the optimal schedule with increasing probability, as the fine-tuning of the schedule parameters likely becomes more difficult with growing $N$. Then, the bisection will assume some $T$ to be too short to solve the problem, even though it would be sufficient. This will lead to a systematic $N$-dependent over-estimation of the required $T$, thus increasing the scaling exponent. We believe that this could explain the difference between the exponents obtained by our ansatz compared to the ansatz by C{\^o}t{\'e} et al.

The frustrated Ising ring with the same choice of parameters was also studied by Wang et al.~\cite{Wang_2025}, where they find a linear scaling of the required time using discretized time evolution as an ansatz for QAOA~\cite{Farhi2014}. However, in their work the time evolution is not restricted to the low-energy subspace. Rather, the lowest four instantaneous eigenstates are only half occupied at small system sizes and almost completely vacated for large systems, except for the very beginning and end of the evolution (Figure 7 in Ref.~\cite{Wang_2025}). The approach of Wang et al. seems to work due to a yet unidentified principle, which we expect to be different from the mechanism of the schedules proposed here, as indicated by the almost complete evacuation of the low-energy subspace. We also note that the approach proposed by Wang et al.\@ requires the number of parameters to scale with the system size, while in our setting the number of parameters to optimize is fixed and does not scale with system size. 

\subsection{Toy model of first-order quantum phase transitions}
\label{sec:ToyModel}

Ref.~\cite{Werner_2023} proposes a simple toy model, where instances can be generated and classified according to the presence or absence of a first-order quantum phase transition in the localized ground state phase. These transitions correspond to avoided level crossings, also called perturbative anti-crossings, in the ferromagnetic phase and are a known slow-down mechanism for adiabatic quantum computation \cite{Altshuler_2010}. The target Hamiltonian of the toy model has a global and a local minimum and no other structure beyond that, allowing us to isolate the physical properties leading to speed-up conditions with our ansatz. Therefore, we test our proposed diabatic annealing schedule on this toy model with two instances: one with a perturbative anti-crossing  where we expect a substantial speedup, and another without an anti-crossing where we expect no substantial advantage.

The Hamiltonian of the toy model is given by
\begin{equation}
    H(s) = - (1-s) \sum_{n=1}^N X_n / N + \sum_{z\in \{ 0, 1\}^N} E_z \ket{z}\bra{z} \ ,
    \label{eq:ToyModelHamiltonian}
\end{equation}
where the sum over $z$ is over all bit strings with $N$ bits. The driver Hamiltonian is a rescaled transverse field Hamiltonian with ground state 
\begin{equation}
    \ket{E_0(s=0)} = 2^{-N/2} \sum_{z\in \{ 0, 1\}^N} \ket{z} \ .
\end{equation}
The target Hamiltonian $H_Z$ is diagonal and characterized by its eigenvalues $\left\{ E_z \right\}$. When taken as the adjacency matrix of a graph $G$ with $2^N$ nodes, the driver Hamiltonian defines a hypercube graph that is a $N$-regular graph. One node of the graph, which we can assume to be node $z=0$, is assigned the ground-state energy $E_0 = -1$. Far away in terms of distance on the graph $G$, the node $z=2^N-1$ is assigned the first excited-state energy $-1 < E_1 < 0$, where $E_1$ is sampled from the uniform distribution $\mathcal{U}(-1, 0)$. This local minimum is extended in size by randomly selecting a neighbor and also assigning it the energy $E_1$. The nodes with the energy $E_1$ define the set $V$. Depending on the depth of the local minimum, an instance might display an avoided level crossing in the localized phase. As discussed in Ref.~\cite{Werner_2023}, an instance exhibits a localized-localized transition if
\begin{equation}
    \frac{1}{1 + \langle E \rangle - E_0} \leq \frac{N-\phi(V)}{N-\phi(V) + N(E_1 - E_0)} \ ,
\end{equation}
where $\langle E \rangle = 2^{-N} \sum_{z \in \{0, 1 \}^N} E_z$ is the average energy over all eigenstates of $H_Z$. The quantity $\phi(V)$ is called the conductance of the local minimum $V$, and it is defined as
\begin{equation}
    \phi(V) = \frac{|\partial V|}{|V|} \ ,
\end{equation}
where $\partial V$ is the set of edges in $G$ with exactly one end in the set $V$. An example of such a spectrum is shown in Figure \ref{fig:ToyModelSpectra}(a). In contrast, if 
\begin{equation}
    \frac{1}{1 + \langle E \rangle - E_0} \geq \frac{d_{\max}(V)}{d_{\max}(V) + N(E_1 - E_0)} \ ,
\end{equation}
where $d_{\max} (V)$ is the maximum degree of the subgraph induced by $V$, the instance does not display a localized-localized level crossing but merely a delocalized-localized one, as depicted in Figure \ref{fig:ToyModelSpectra}(b). In the instances we consider here, where the local minimum $V$ contains $|V| = 2$ nodes, we find $d_{\max}(V) = 1$ and $|\partial V| = 2N-1$.

Since here we are investigating specific instances instead of a size-dependent scaling, we numerically optimize the schedule parameters for a fixed total time $T$, where the cost function to be minimized is again the expectation value of the target Hamiltonian $H_Z=\sum_{z\in \{ 0, 1\}^N} E_z \ket{z}\bra{z}$. For each value of $T$, the inner-loop optimizer runs $N_{\text{maxiter}} = 10^4$ calls to the quantum annealer. For each optimized schedule found at a given $T$, we calculate the normalized residual energy $\varepsilon_{\text{residual}}$ and fidelity $\mathcal{F}$,
\begin{eqnarray}
     \varepsilon_{\text{residual}}(T) &=& \frac{\langle \psi(T) | H_Z | \psi(T) \rangle - E_0}{ E_1 - E_0} \ ,
    \label{eq:DefResEnergy1} \\
     \mathcal{F}(T) &=& |\langle \psi(T) | 0 \rangle|^2 \ ,
    \label{eq:DefFidelity}
\end{eqnarray}
to compare these solutions against linear ramp annealing in Figure \ref{fig:ToyModelResult}.
We perform this optimization for both instances (Figure~\ref{fig:ToyModelSpectra}), and we distinguish between two types of optimized schedules found:  monotonic ($s_1 \leq s_2$) and non-monotonic ($s_1 > s_2$, compare Figure \ref{fig:DiagrammeIntuition}(e)) solutions.

In the case of the instance without the localized-localized crossing (Figure~\ref{fig:ToyModelSpectra}(b)), the variational schedules are able to reach the ground state at much lower $T$ compared to the linear ramp, but for sufficiently large $T$ both are able to find the ground state with high probability. The speedup of the variational schedules compared to the linear schedule is because the variational schedule has the flexibility to ramp fast to the minimal gap, pass it slowly, and then ramp quickly to the end. This is reflected by the monotonicity of the optimal schedules in Figure \ref{fig:ToyModelResult}(c,d).

In the case of the instance with the localized-localized crossing (Figure~\ref{fig:ToyModelSpectra}(a)), the linear schedule is almost unable to prepare the target ground state with high probability for the range of $T$ values explored. The optimal schedules almost exclusively assume a non-monotonic shape, and the target ground state fidelity reaches close to unity.
While some overlap is observed for short anneal times for the linear ramp, the fidelity is still much lower compared to the fidelity reached by the variational schedules. At these shorter times, the linear ramp dynamics are such that the system partially jumps to the first excited state at the delocalized-localized transition, which then leads to some population ending up in the final ground state after crossing the localized-localized transition.

Examples of the schedules and the populations of the instantaneous eigenstates during the evolution are shown in Figure~\ref{fig:ToyScheduleExample}, where we choose the value of $T$ that minimizes the target expectation value. For the instance with the localized-localized transition, we can clearly see that the optimized schedule behaves as expected, namely that the sweeps are used to populate the first excited state before the localized-localized anti-crossing. After sweeping past the localized-localized transition, most of the quantum state population ends up in the target ground state. In contrast, for the linear ramp schedule, the optimal time is the maximum $T$ that we tested, and even this time is insufficient since the target ground state population ends up very close to zero.

These results indicate that the localized-localized crossings, which often hinder linear quantum annealing schedules, can be overcome using the LZS-ansatz. Next, we will test this finding on an instance of an NP-hard problem.

\begin{figure}
    \centering
    \begin{tikzpicture}
        \node[anchor=north west, inner sep=0] (fig) at (0,0)
            {\includegraphics[width=0.9\linewidth]{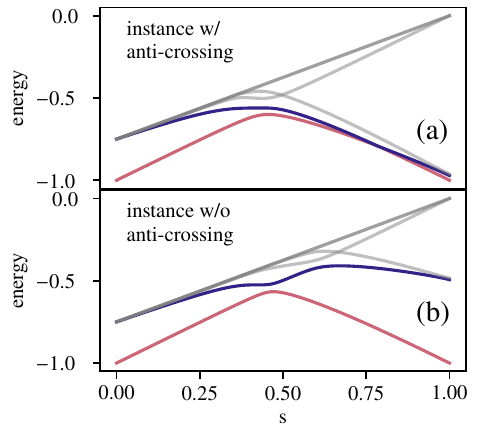}};
    \end{tikzpicture}
    \caption{The lowest six eigenvalues of the Hamiltonian $H_s$ for two instances of the toy model defined by Eq.~\eqref{eq:ToyModelHamiltonian}: (a) an instance with a localized-localized crossing and (b) an instance without a localized-localized crossing. The ground state energy is shown in red, the first excited state energy in blue and the higher excitations in gray.}
    \label{fig:ToyModelSpectra}
\end{figure}

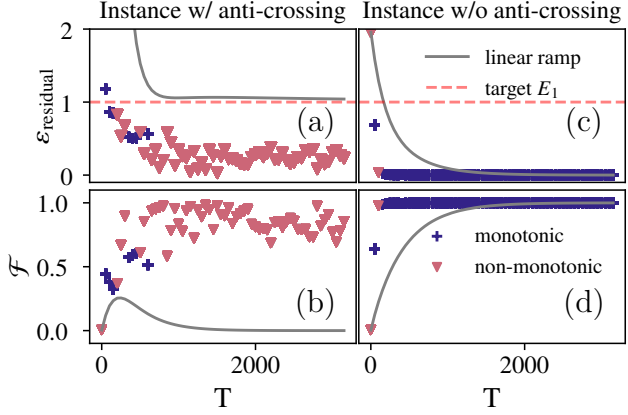
\begin{figure}
    \centering
    \begin{tikzpicture}
        \node[anchor=north west, inner sep=0] (fig) at (0,0)
            {\includegraphics[width=\linewidth]{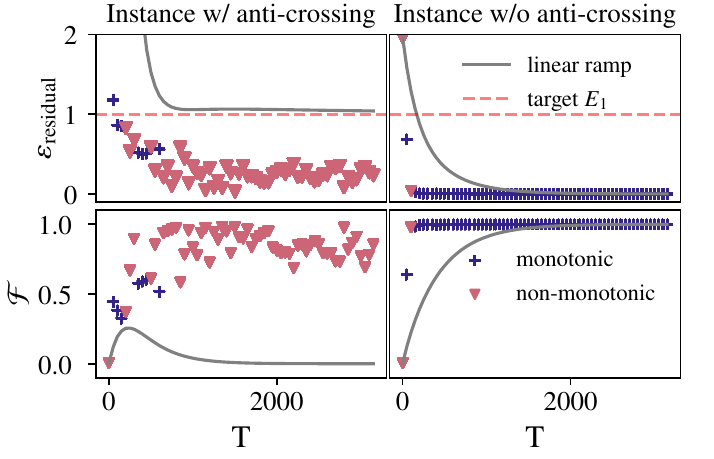}};
        \node[font=\large] at ($(fig.north west)!0.48!(fig.north east) + (fig.north west)!0.3!(fig.south west)$) {(a)};
        \node[font=\large] at ($(fig.north west)!0.48!(fig.north east) + (fig.north west)!0.73!(fig.south west)$) {(b)};
        \node[font=\large] at ($(fig.north west)!0.89!(fig.north east) + (fig.north west)!0.3!(fig.south west)$) {(c)};
        \node[font=\large] at ($(fig.north west)!0.89!(fig.north east) + (fig.north west)!0.73!(fig.south west)$) {(d)};
    \end{tikzpicture}
    \caption{(a) Normalized residual energy $\varepsilon_{\text{residual}}$  (Eq.~\eqref{eq:DefResEnergy1}) 
    and (b) target ground state fidelity $\mathcal{F}$ (Eq. \eqref{eq:DefFidelity}) as functions of the time budget $T$ for an instance of the toy model with a perturbative anti-crossing in the localized phase. (c) and (d) show the same quantities, but for a problem instance without the anti-crossing. The results for linear annealing schedules are plotted in gray, while the results for the optimized ansatz with  non-monotonic schedules are denoted with red triangles and those with monotonic schedules are denoted with blue crosses.}
    \label{fig:ToyModelResult}
\end{figure}

\begin{figure}
    \centering
    \begin{tikzpicture}
        \node[anchor=north west, inner sep=0] (fig) at (0,0)
            {\includegraphics[width=\linewidth]{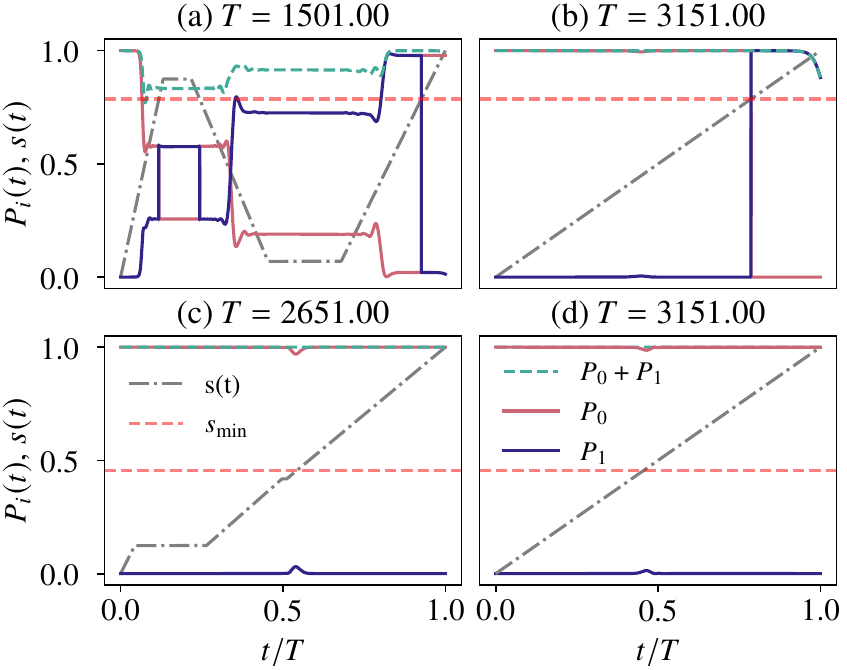}};
    \end{tikzpicture}
    \caption{Instantaneous eigenstate populations $P_i(t)$ and optimal schedules $s(t)$ for the optimal time budget $T$ for (a,c) the parameterized schedule and (b,d) the linear ramp. (a,b) The first row corresponds to the instance with the loc-loc transition, and  (c,d) the second row corresponds to the instance without the loc-loc transition.}
    \label{fig:ToyScheduleExample}
\end{figure}
\subsection{Observed speed-up on a MAXCUT instance}
\label{sec:Instance1}
In this section, we investigate an instance of the MAXCUT problem that has been used to benchmark quantum annealing and other methods in the literature \cite{Zhou_2020, Pecci_2024}. A MAXCUT problem is defined by a weighted graph $G$, and this can be converted to a target Hamiltonian given by
\begin{equation}
    H_Z = \sum_{(m,n) \in \mathcal{E}_G} \frac{J_{mn}}{2} (Z_mZ_n-1) \ ,
    \label{eq:MAXCUTHamiltonian}
\end{equation}
where $\mathcal{E}_G$ is the set of edges of the graph $G$ and the matrix $J_{mn}$ contains its weights. We focus on the problem referred to in the literature as ``MAXCUT instance 1,'' whose graph $G$ with $N=14$ nodes and weights can be found in Ref.~\cite{Pecci_2024}. When converted to the interpolating Hamiltonian in Eq.~\eqref{eq:DefFullH}, the spectrum displays an avoided level crossing between the ground and first excited state energies around $s_{\min} \approx 0.92$, as shown in Figure \ref{fig:Instance1Result}(a). In the neighborhood of the avoided crossing, the lowest two energy levels seem to be comparatively well separated from the rest of the spectrum.

\begin{figure*}
    \centering
    \begin{tikzpicture}
        \node[anchor=north west, inner sep=0] (top) at (0,0)
            {\includegraphics[width=0.97\linewidth]{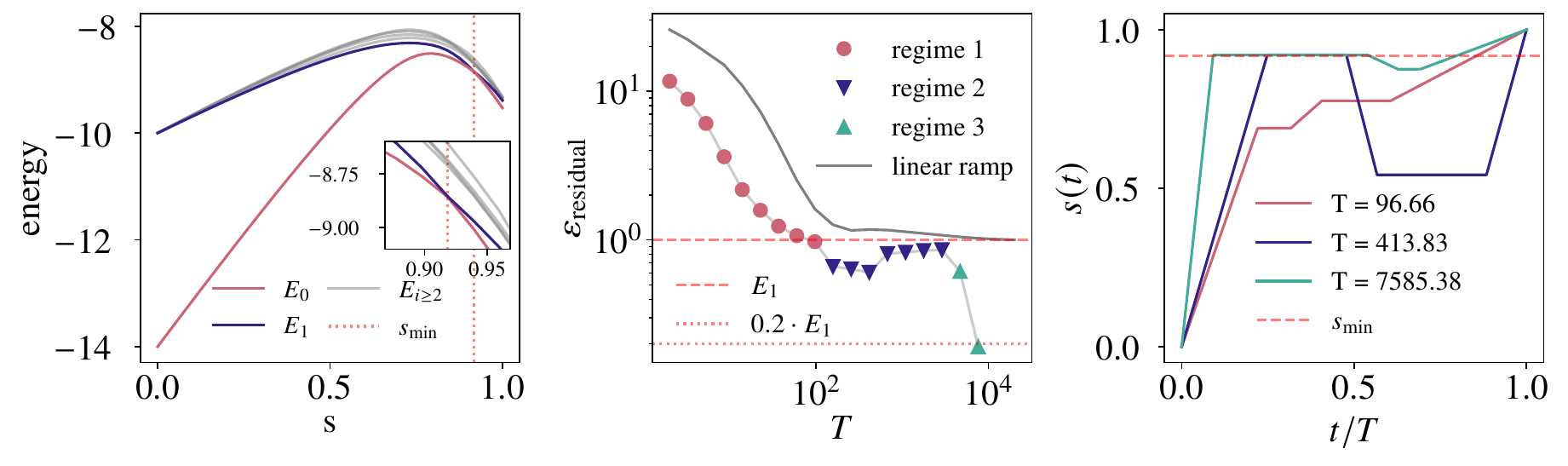}};
        \node[anchor=north west, inner sep=0] (bottom) at (0.69cm,-5cm)
            {\includegraphics[width=0.9\linewidth]{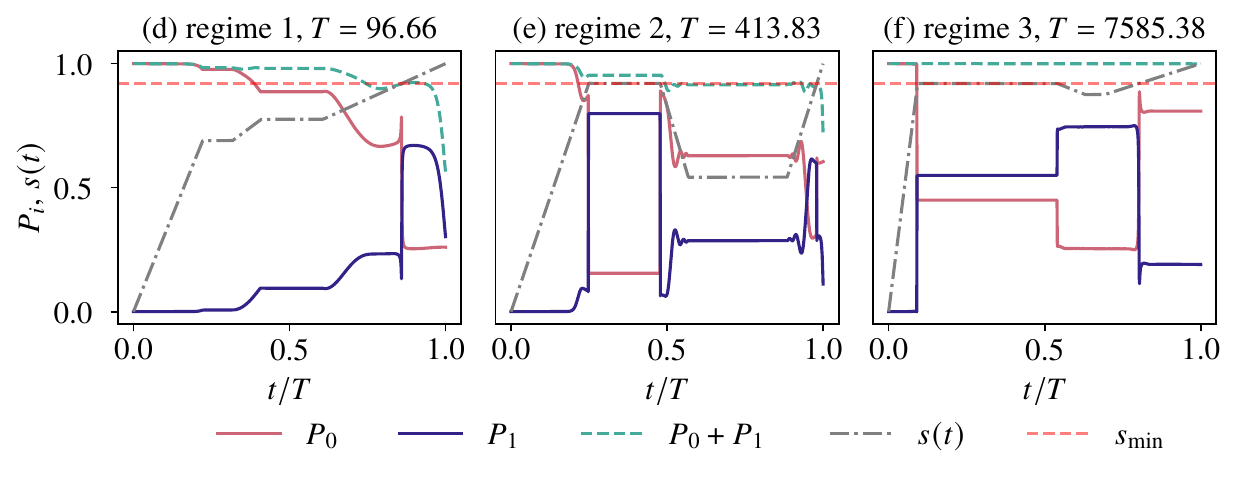}};
        \node[font=\large] at ($(top.north west)!0.05!(top.north east) + (0, 0.2cm)$) {(a)};
        \node[font=\large] at ($(top.north west)!0.38!(top.north east) + (0, 0.2cm)$) {(b)};
        \node[font=\large] at ($(top.north west)!0.7!(top.north east) + (0, 0.2cm)$) {(c)};
    \end{tikzpicture}
    \caption{Results for MAXCUT instance 1. (a): The lowest six energy eigenvalues of $H(s)$ Eq. \eqref{eq:DefFullH}, where the target Hamiltonian is given by MAXCUT instance 1 Eq. \eqref{eq:MAXCUTHamiltonian}. The ground state energy is shown in red, the first excited state in blue, while higher energy levels are gray; (b): Normalized residual energy (Eq. \eqref{eq:DefResEnergy1}) for MAXCUT instance 1 for the optimized LZS-schedules and the linear ramp for comparison. For the optimized schedule, the three regimes discussed in the text are shown in different colored markers (red $\bigcirc$, blue $\triangledown$, green $\triangle$), and the linear ramp results are shown in grey. For reference, the horizontal line corresponds to the first excited state energy; (c): Optimal LZS-schedules corresponding to the smallest residual energies for each regime shown in (b), where we used the same color coding. The red dashed line indicates the location of the minimal gap. (d,e,f) Eigenstate populations $P_i$, $i=0, 1$, for the best performing schedule in each regime identified in (b).}
    \label{fig:Instance1Result}
\end{figure*}
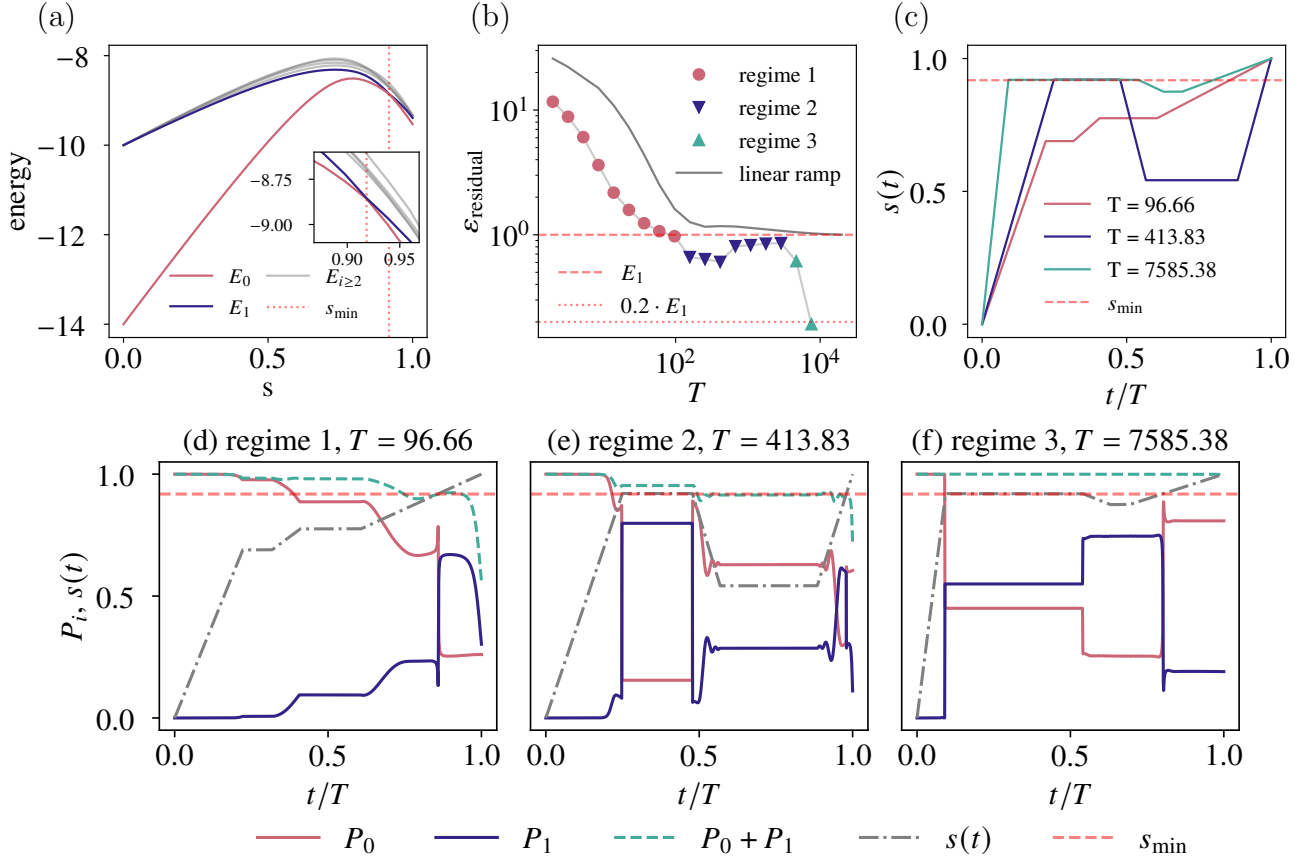

As before, we optimize the parameters of the LZS-schedule with the expectation value of the target Hamiltonian as the cost function to be minimized. However, here we increase the time $T$ exponentially. For each $T$, the inner-loop optimizer uses $N_{\text{maxiter}} = 4 \times 10^3$ calls to the quantum annealer. An exception is the optimization for the largest $T=7585.38$, where $N_{\text{maxiter}} = 2 \times 10^3$ in order to keep the wall clock computation time low. Since this still results in the best performing schedule, we conclude that this is not a significant impairment. 

We show the normalized residual energy $\varepsilon_{\text{residual}}$ (Eq. \eqref{eq:DefResEnergy1}) as a function of $T$ in Figure \ref{fig:Instance1Result}(b). For comparison, we also show the residual energies obtained using a linear ramp schedule. As the total annealing time increases, we can distinguish three regimes, where the optimal schedules display qualitatively distinct behavior. In Figure \ref{fig:Instance1Result}(b), the three regimes are distinguished by different colors and markers for the normalized residual energy points. In the first regime, we see that, as the annealing time $T$ increases, the normalized residual energy of both the optimized schedule and the linear ramp decay, with the optimized schedule outperforming the linear schedule by a constant factor.

In the second regime, the two cases begin to show a qualitative difference. Both schedule classes manage to reach the first excited state energy on average. However, the linear ramp approaches it much more slowly and even with the maximum time budget we allow does not achieve an average energy below the first excited state. On the other hand, the optimized schedules reach average energies substantially lower than the first excited state, indicating that they obtain substantial overlap with the target ground state. Increasing the time budget $T$ further, the optimized schedules enter a third regime, where the average energy decreases even further and eventually seems to predominantly reach the target ground state.

The three regimes of the normalized residual energy also correspond to qualitatively different optimal schedules. In Figure \ref{fig:Instance1Result}(c) we show the schedule that minimizes $\varepsilon_{\text{residual}}$ in each regime, but the examples are representative for their respective regimes. The optimal schedules for all other $T$ are shown in Appendix \ref{sec:AppendixDataInstance1}. In the first regime (red), the schedules first ramp quickly while the gap is large, just to then slow down and ramp more slowly when the gap is small. However, the optimal schedules still ramp monotonically from $s=0$ to $s=1$. In the second regime (blue) the optimal schedule makes use of the ability to go back and forth, such that the populations can interfere. Here, all schedules are non-monotonic. In the last regime (green), the optimal schedule seems to follow a different strategy. The schedule seems to quickly ramp to the minimal gap and then halts.

The fact that these three different regimes are distinct is further substantiated when considering the populations of the instantaneous eigenstates when evolving under the optimal schedules for each regime, shown in Figure \ref{fig:Instance1Result} (d-f). In the first regime, the schedule is monotonic (Figure \ref{fig:Instance1Result}(d)), and the system remains mainly in the ground state until the schedules crosses the minimal gap, where the populations of the first excited state and the ground state invert. This results in a low population for the final Hamiltonian ground state.

In the second regime, the schedule is non-monotonic (Figure \ref{fig:Instance1Result}(e)), and the schedule transfers population from the ground state to the first excited state with each sweep. This increasingly populates the first excited state, and this population transitions into the target ground state after crossing the minimum gap. In contrast to the third regime, the net population transfer from the ground state to the first excited state is limited in this regime.

In regime 3  (Figure \ref{fig:Instance1Result}(f)), the schedule is able to maximally transfer population to the first excited state. First, the schedule ramps adiabatically to the minimal gap, such that the populations remain close to constant. Then, the system halts at the minimal gap from $t/T \approx 0.1$ on-wards, while after moving away from the minimal gap around $t/T \approx 0.55$, a substantial amount of population has been moved to the first excited state. Therefore, in regime 3, the halting at the minimal gap contributes maximally to the population transfer, while the ramps contribute almost nothing to the transfer before crossing the minimal gap.

This behavior can be understood intuitively if we consider the effective Hamiltonian in the two-level subspace. Using the instantaneous eigenstates from sufficiently far away from the minimal gap $\ket{\psi_0^-} \approx \ket{\psi_1^+}$, $\ket{\psi_1^-} \approx \ket{\psi_0^+}$ as a basis, the two-level Hamiltonian in the vicinity of the minimal gap can be approximated as
\begin{equation}
    H(s) \approx C(s-s_{\min}) Z + \gamma_{\min} X / 2
\end{equation}
for some appropriately chosen constant $C$ such that $\gamma_{\min} / C = \mathcal{O}(\exp(-N))$. Approaching the minimal gap from $s < s_{\min}$, the system is in the ground state $\ket{\psi_0^-} \approx \ket{\psi_1^+}$ of $H(s) \approx C(s-s_{\min})Z$. Then, at the minimal gap, the Hamiltonian acts as $H(s_{\min}) = \gamma_{\min} X / 2$, such that after a time $t_i =\pi \gamma_{\min}^{-1}$, the system rotates to the state $\ket{\psi_1^-} \approx \ket{\psi_0^+}$, which is the excited state in $s < s_{\min}$ and the ground state at $s > s_{\min}$. This explains why the schedule moves adiabatically to $s_{\min}$, then halts there, where the system is in an almost equal superposition of the instantaneous eigenstates, and then continues adiabatically, up to the population inversion after crossing $s_{\min}$.

Based on this, we conclude that in regime 3 the optimal schedule solves the problem using a different mechanism than LZS interference. The rotation of the state we describe above requires a time that scales as $\mathcal{O}(\gamma_{\min})$, which we note is still at least quadratically better than an adiabatic ramp through the minimum gap, so this mechanism will typically be substantially slower than the exploitation of the LZS interference, whose time scale depends on the delocalized-to-localized transition. This difference in time scales is suggested by our simulation results in Figure \ref{fig:Instance1Result}(b), where residual energies below the first excited state energy are reached at times $T$ almost two orders of magnitude shorter in regime 2 than in regime 3.

In conclusion, we find that on this particular MAXCUT instance, the LZS-schedule provides a substantial speedup over linear ramps when minimizing the residual energy, and it reaches comparative performance to the results reported by Pecci et al.~(compare Figure 4 in Ref.~\cite{Pecci_2024}), albeit with far fewer parameters in the ansatz. As has been argued in Ref.~\cite{Pecci_2024}, the minimal spectral gap along the anneal path for MAXCUT instance 1 is exceptionally small compared to randomly generated instances. Hence, our simplified schedule provides a speedup on an instance, which is otherwise exceptionally hard to solve with conventional quantum annealing. 

\subsection{Open system simulation of a two-level system}
\label{sec:OpenSystem}
So far, we have only considered unitary dynamics, and in this section, we introduce relaxation to a two-level system to investigate the impact of decoherence numerically.
As we have discussed, the observed speedup in annealing time is achieved by preparing the system in the excited state in a targeted manner using a coherent interference effect prior to crossing the minimal spectral gap . In Section \ref{sec:DecoherenceTheory} we argued with a simplified model based on quantum channels that we expect the schedules ability to precisely prepare the excited state to be impaired by decoherence. We test this hypothesis numerically by considering an open two-level system. The model we consider is a qubit coupled to a thermal bath at inverse temperature $\beta$ described by an adiabatic master equation \cite{Albash2015}, which is discussed in more detail in Appendix \ref{sec:NoiseModel}. The two-level Hamiltonian is given by
\begin{equation}
    H(t) = -(1-s(t)) X + s(t) Z \ ,
    \label{eq:TwoLevelHamiltonian}
\end{equation}
and the system is initially prepared in $\ket{+}$ and $s(t)$ is our parameterized LZS-schedule. At the end of the anneal, the system's Hamiltonian is $Z$ with an excited state energy of 1. For simplicity, we investigate only the preparation of the excited state via our simplified diabatic annealing schedule in order to isolate the relevant dynamics.

As before, we linearly increase $T$ and optimize the annealing schedule for each $T$ using COBYLA \cite{Powell_1994} until convergence. We report the best result obtained from 10 random  of the optimizer. We keep the inverse temperature fixed at $\beta=0.2$ and increase the coupling strength to the environment $g$.

Figure \ref{fig:OpenSystemResult}(a) shows the target expectation value as a function of the coupling $g$ and the time $T$. We observe three regimes: in the bottom half of the image, corresponding to couplings $g \gtrapprox 0.1$, the target expectation values do not substantially deviate from zero, indicating that the optimizer is unable to prepare the excited state. In this regime, the coupling to the environment, i.e. the decoherence, is too strong to solve the task at all.

In the top half of the plot for $g \lessapprox 0.05$, there are two qualitative regimes, one for short time budgets ($T \lessapprox \pi$) where the system also struggles to prepare the excited state and one for larger $T$, where the system prepares the excited state with the maximal probability. The respective maximum depends on the coupling $g$.

The white dots in the heatmap indicate that the respective LZS-schedule does not sweep back and forth but increases monotonically from $s=0$ to $s=1$, while the gray dots indicate non-monotonic schedules. For long enough time budgets ($T \gtrapprox \pi$) and sufficiently coherent ($g \lessapprox 0.1$), the optimal schedules reach values close to one, i.e. the system is able to completely reach the the excited state by sweeping back and forth. This regime can be better understood by plotting the target expectation value over the time budget, as seen in Figure \ref{fig:OpenSystemResult}(b). From this plot, we can see two things, namely that the time to optimally solve the problem coincides with $T=\pi$ and that the target expectation values decrease as the coupling $g$ increases. This agrees with our theoretical prediction that in a system with decoherence there is a maximum population that can be transferred to the excited state.

The fact that the optimal time is $T=\pi$ can be explained intuitively. Given the Hamiltonian in Eq.~\eqref{eq:TwoLevelHamiltonian}, since $0 \leq s \leq 1$ the state of the system on the Bloch sphere can rotate either around the $X$-axis with negative orientation, around the $Z$-axis with positive orientation, or any convex combination of the two. Initially, the system is in the state $\ket{+}$, and we need to reach the state $\ket{1}$. The shortest path between the two would be to rotate around the $Y$-axis. However, $Y$-rotations are not in our toolbox. Consider instead the case of pure Z- and -X-rotations. The initial state $\ket{+}$ evolves into $\ket{1}$ by evolving under $\exp(-itZ)$ for time $t=\pi/2$, resulting in a $\pi/4$-rotation on the Bloch sphere. This is followed by a evolution $\exp(itX)$, also for $t=\pi/2$ for another $\pi/4$-rotation, giving an evolution time of $T=\pi$.

\begin{figure}
    \centering
    \begin{tikzpicture}
        \node[anchor=north west, inner sep=0] (fig) at (0,0)
            {\includegraphics[width=\linewidth]{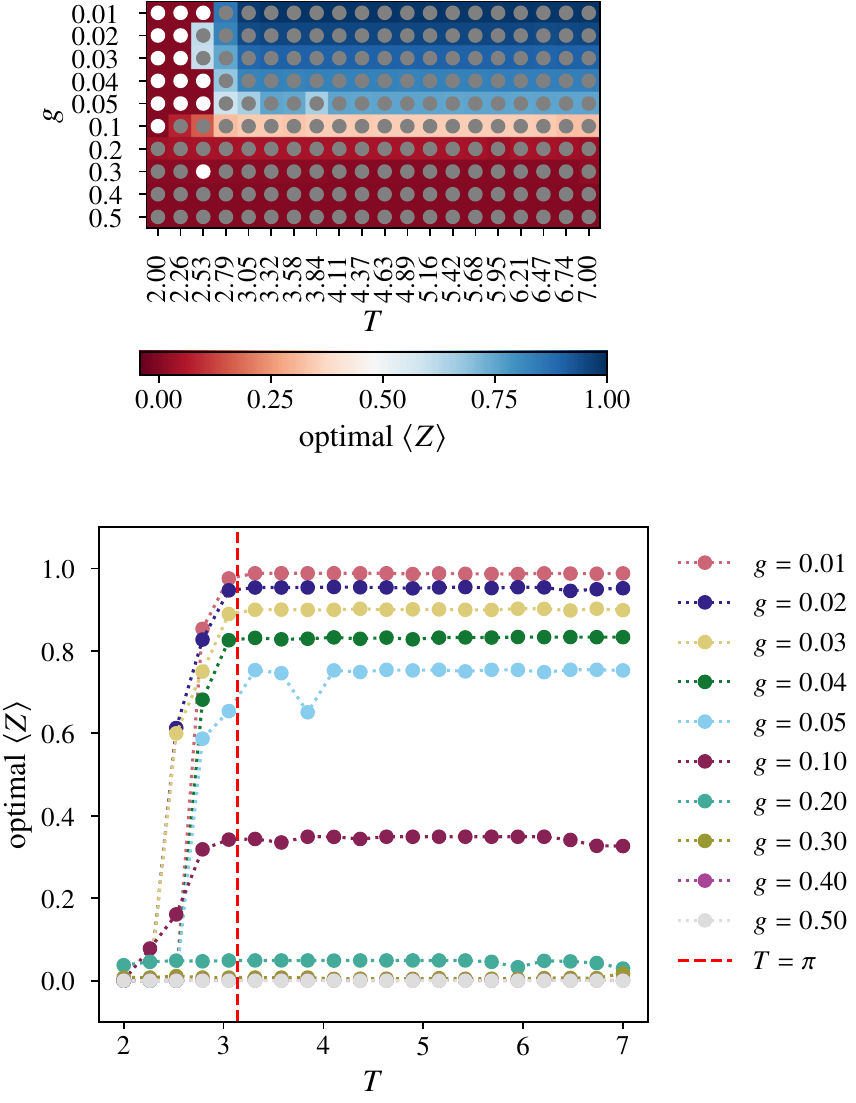}};
        \node[font=\large] at (fig.north west) {(a)};
        \node[font=\large] at ($(fig.north west)!0.45!(fig.south west)$) {(b)};
    \end{tikzpicture}
    \caption{Target expectation value of a two level system (Eq. \eqref{eq:TwoLevelHamiltonian}) after optimizing the LZS-schedule. The goal of this optimization is to prepare the excited state of the system, i.e. to reach $\langle Z \rangle = 1$. (a): Optimal target expectation value as a function of the coupling to the environment $g$ and the time budget $T$, the gray dots indicate that the optimal schedule assumes a back-and-forth shape, while white dots indicate monotonic schedules. (b): Optimal target expectation value as a function of time budget $T$.}
    \label{fig:OpenSystemResult}
\end{figure}

\subsection{Cubic interpolation improves scaling}
\label{sec:CubicSplines}
In this section, we discuss a potential improvement of our ansatz.
The class of schedules we propose is defined by six points in the $t-s$-plane, where one point is fixed at $(0, 0)$ and another at $(T, 1)$. The other points are determined by the remaining six parameters, since we enforce the existence of the plateaus in the schedule. So far, we have used linear interpolation between the points to construct the full schedule. However, we may consider a cubic spline interpolation between the points defined by our schedule definition. For the simulations on the frustrated Ising ring, we use the same parameters and the same optimization algorithm as in Section \ref{sec:IsingRingSimulation}, except that the schedule $s(t)$ is decoded using cubic splines instead of linear interpolation. 

An example of such a schedule is shown in Figure \ref{fig:FrustratedIsingRingCubicIntuition}(a). We observe that the change to a cubic spline interpolation significantly improves the scaling of the annealing time $T$ as a function of $N$, as shown in Figure \ref{fig:FrustratedIsingRingResultCubic}(a), reducing the exponent from $\alpha \approx 1.75$ to $\alpha \approx 1.2$ or $\alpha \approx 1.4$ for $c=0.5$ and $c=0.1$ respectively. At the same time, the dynamics of the system remains qualitatively similar to the linear interpolation case, at least for the small system sizes where exact state vector simulation is feasible. This can be seen by comparing Figure \ref{fig:FrustratedIsingRingResult}(b) and Figure \ref{fig:FrustratedIsingRingResultCubic}(b).

We do not have a conclusive theoretical explanation for this observation. However, we speculate that smoother turn-around enables the optimizer to effectively evolve locally adiabatically \cite{Roland_2002}, i.e. change the Hamiltonian rapidly when the gap is large and slowly when the gap is small. While the linear interpolation with the plateaus needs to cross the phase transition at a global rate for at least two out of the three sweeps, it does not allow to effectively move fast to the phase transition and then cross it slowly. On the other hand, the cubic interpolation results in schedules that smoothly turn around, so that the gradient also smoothly changes sign. This allows to effectively control the gradient at the phase transition, while maintaining a fast gradient where the spectral gap is large. We illustrate this reasoning in Figure \ref{fig:FrustratedIsingRingCubicIntuition}(b). We speculate that local adiabaticity could improve upon this and achieve a substantially sub-quadratic scaling. It is also conceivable that smoother schedules make the optimization landscape easier. Further research is needed to investigate these conjectures.

\begin{figure}
    \centering
    \begin{tikzpicture}
        \node[anchor=north west, inner sep=0] (a) at (0, 0)
            {\includegraphics[width=0.8\linewidth]{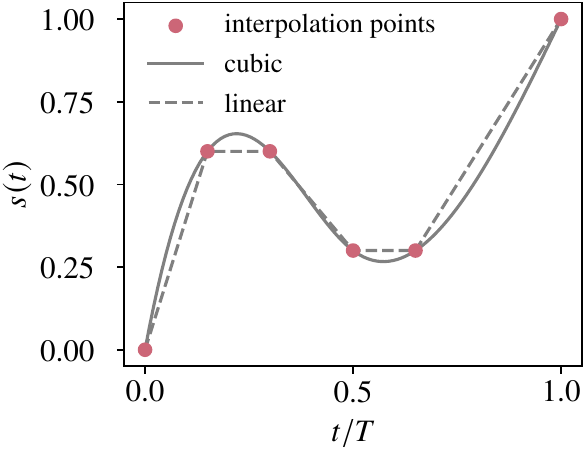}};
        \node[anchor=north west, inner sep=0] (b) at (1.15cm,-5.8cm)
            {\includegraphics[width=0.7\linewidth]{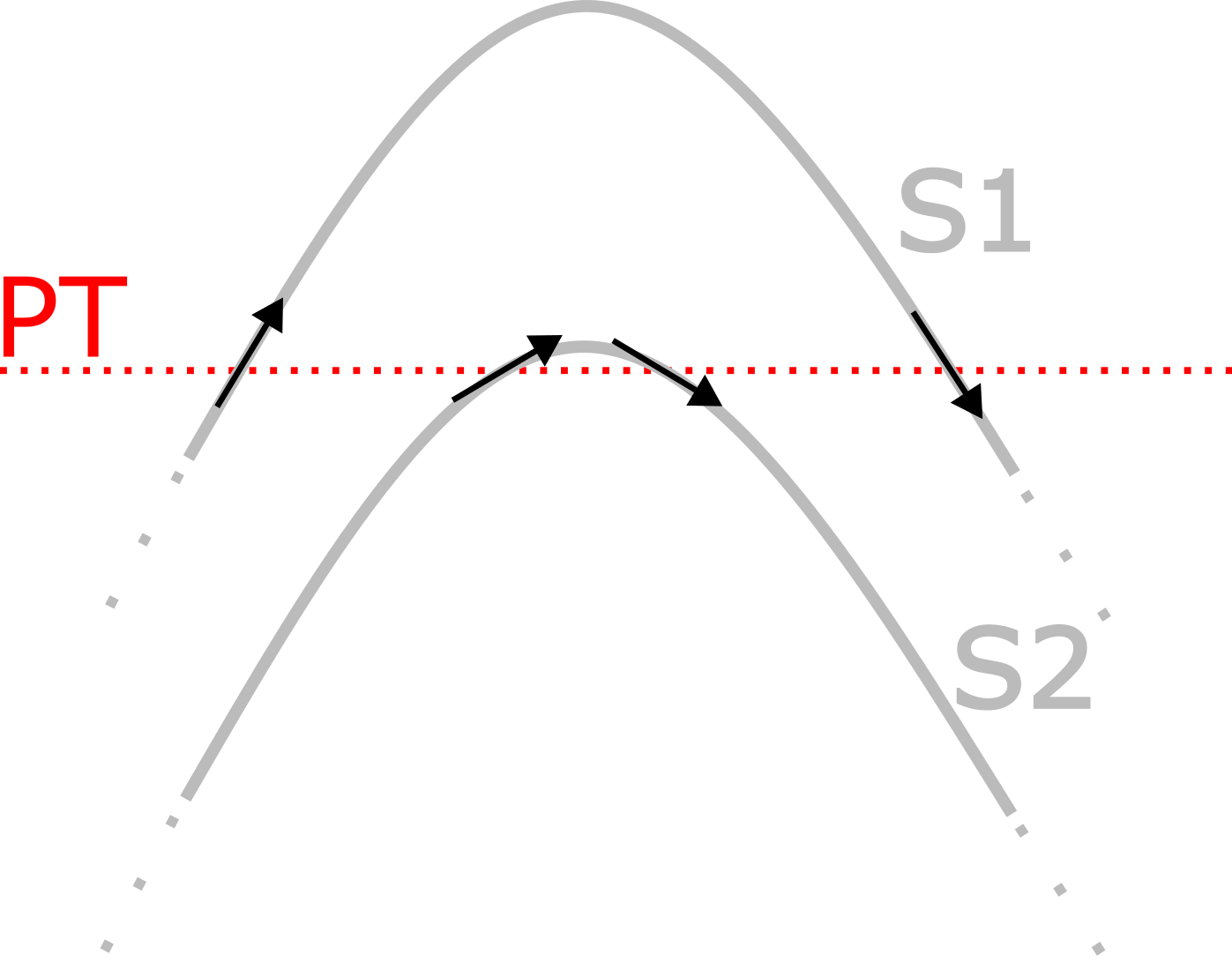}};
        \node[font=\large] at (a.north west) {(a)};
        \node[font=\large] at ($(b.north west) + (-1.15cm, 0)$) {(b)};
    \end{tikzpicture}
    \caption{(a) Schedules obtained by linear interpolation versus cubic spline interpolation for the same set of parameters. (b) Illustration of the effective reduction of the rate of change at the phase transition. The gradient (black arrows) at the phase transition is smaller for the schedules S2 in comparison to the schedule S1.}
    \label{fig:FrustratedIsingRingCubicIntuition}
\end{figure}

\begin{figure}
    \centering
    \begin{tikzpicture}
        \node[anchor=north west, inner sep=0] (fig) at (0,0) {\includegraphics[width=0.92\linewidth]{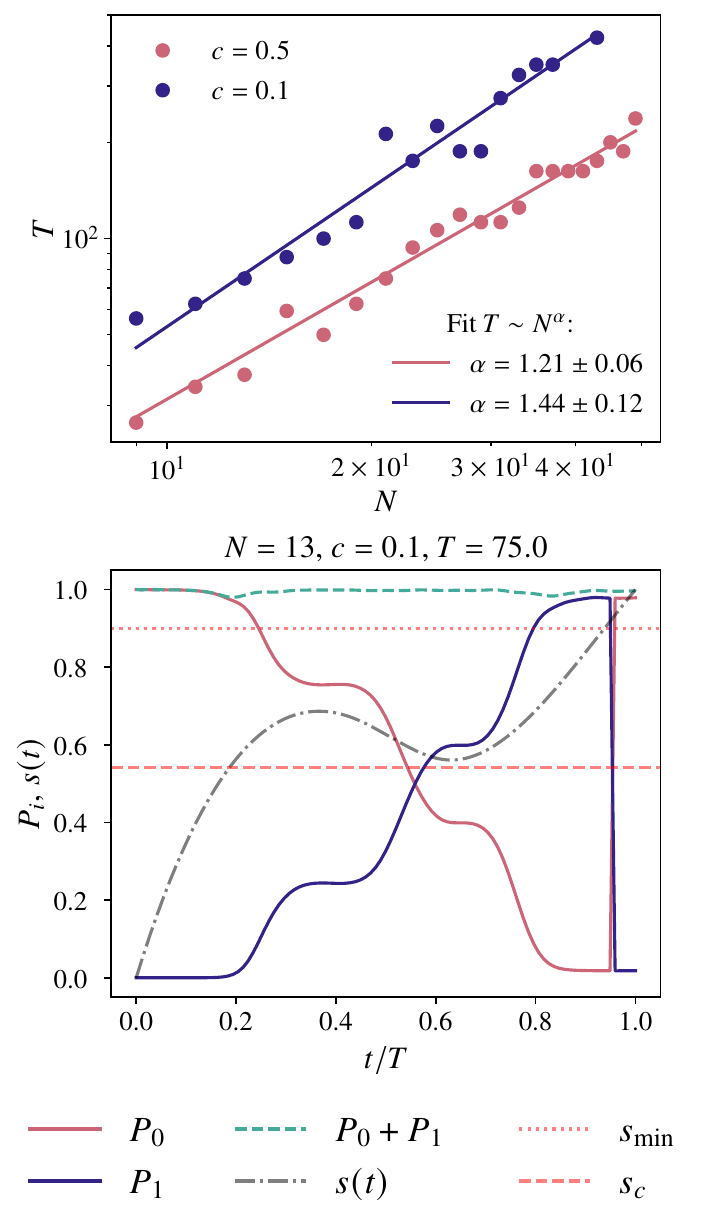}};
        \node[font=\large] at (fig.north west) {(a)};
        \node[font=\large] at ($(fig.north west)!0.45!(fig.south west)$) {(b)};
    \end{tikzpicture}
    \caption{Same plots as in Figure \ref{fig:FrustratedIsingRingResult}, but using cubic interpolation when decoding the schedule parameters. The dynamics shown in (b) is qualitatively similar to the linear interpolation. However, the exponent of the power law dependence of $T$ on $N$ is significantly smaller.}
    \label{fig:FrustratedIsingRingResultCubic}
\end{figure}
\section{Conclusion and outlook}
\label{sec:Conclusion}
In this work we proposed a simplified ansatz for variational diabatic quantum annealing that is characterized by only a small number of parameters independent of the system size. This is in stark contrast to most variational ans{\"a}tze, where the number of parameters grows with system size. The parameter efficiency is enabled by making explicit use of Landau-Zener-St{\"u}ckelberg interference (LZS), which allows for efficient population transfer to the first excited state, assuming the existence of a polynomially closing gap along the anneal. The inversion of the population in turn does not require the optimized annealing schedule to slow down exponentially at exponentially small spectral gaps, commonly due to perturbative anti-crossings. In fact, we are able to prove that sufficiently many consecutive sweeps can prepare any arbitrary state in a two-level subspace, as stated in Proposition \ref{lemma:UniversalSU2}, which in principle can overcome any obstruction to the anneal that is restricted to the low-energy subspace. It will be instructive to explore which problems require a more complex state preparation at the control interval to be optimally solved, potentially at the cost of choosing a more expressive ansatz with more sweeps.

The small number of parameters additionally allows us to show that efficient optimization of the ansatz is possible (Proposition \ref{prop:efficiency}) assuming that this schedule family is able to solve a given problem with sufficient precision. However, further research is required to shed light on whether the efficient trainability of this particular variational algorithm implies tractable classical simulation~\cite{Cerezo_2025}, a question which is extensively discussed in the literature on digital variational quantum circuits. Additionally, we give some sufficient conditions for when we expect the proposed schedule class to deliver a speedup over adiabatic annealing. It remains to be seen if these conditions are too strict. This is relevant to determine if the observed speedups over conventional quantum annealing are genuine, or merely an artifact of the specific problem instance's simplicity.

We test our proposal on various models. We reproduce the exponential speedup over linear ramps observed in Ref.~\cite{Cote_2023}. We even observe a slightly sub-quadratic scaling of the required annealing time $T$ for our ansatz. Subsequently, we test the ansatz on a toy model of Hamiltonians that display a phase transition from an extended ground state in the computational basis to a localized one, as well as perturbative anti-crossings in the localized phase. These anti-crossings are known to result in exponentially vanishing spectral gaps, making the problem instance particularly challenging for adiabatic computation. We show that with our ansatz, such instances can still be solved, thus indicating that the ansatz can systematically solve this failure case, when a spectral gap with sufficiently favorable scaling can be found along the anneal.

We also test the ansatz on a particularly pathological instance of MAXCUT, where we find that the identified mechanism reduces the optimization cost significantly faster compared to linear ramp quantum annealing. Our ansatz also identifies a third mechanism to accurately prepare the ground state, namely, progressing adiabatically to the minimal spectral gap and halting there. In this strategy, the Hamiltonian at the minimal gap acts as a $X$-gate (potentially up to phases) when pausing for a time proportional to $\gamma_{\min}^{-1}$. Therefore, while quadratically faster than an adiabatic transition on the same problem, this third mechanism requires substantially more annealing time than the LZS-mechanism.

While our ansatz seems to substantially improve the annealing result on the MAXCUT instance with an exponentially closing gap due to a perturbative anti-crossing, we only observe a mild improvement of the optimized ansatz compared to the linear ramp annealing on 20 randomly generated MAXCUT instances, as we show in Appendix \ref{sec:AppendixRandomMaxcut}. While this seems to indicate that our ansatz might not provide a qualitative benefit over linear quantum annealing on randomly generated MAXCUT instances, our observation could be due to the fact that random instances lack the structure to likely have perturbative anti-crossings. There may still exist problem classes of practical relevance where this phenomenon is more common, and thus our ansatz can potentially improve performance.

One aspect to highlight here is that the optimized schedule performs significantly better than linear ramps when considering the expected energy of the target Hamiltonian. For the MAXCUT instance 1 and the toy model, we observe that the linear ramp achieves a much smaller but non-zero fidelity with the ground state. This suggests that in some instances, the optimization task of observing the ground state at least once could also be solved by sufficiently many repetitions of the linear ramp, without having to invest computational efforts into optimizing the schedule parameters. This would be reflected when considering the time-to-solution metric~\cite{Ronnow2014}. However, if the task is the best possible ground state preparation, our ansatz does significantly better under conditions where linear ramps would fail.

We argue that decoherence will impair the population inversion mechanism using LZS interference. To verify our theoretical argument, we test the ansatz in a simple two-level system coupled to a thermal bath. The simulation results align with what we have argued theoretically, namely that the speedup mechanism is inhibited under dephasing noise. This establishes the LZS-mechanism as an inherently coherent effect, hinting that coherence can be a useful resource for quantum optimization and essential for the LZS-mechanism. At the same time, this limits the applicability of our ansatz to devices with sufficiently long coherence times.

We also observe that an equally simple ansatz using cubic spline interpolation results in qualitatively very similar dynamics for the frustrated Ising ring, with the significant difference being that the scaling exponents for the required anneal time drop to values closer to linear than quadratic. We speculate that this is due to the fact that the smooth non-monotonic schedules encoded by this method result in a similar mechanism as the LZS, but where it is implicitly combined with local adiabatic search. However, we did not conclusively answer this question, leaving it for future research. While we believe our method exploits a different mechanism compared to Ref.~\cite{Wang_2025}, Wang et al.~demonstrate that the problem is in principle solvable in linear time by a quantum annealing-inspired method. This suggests that, potentially, the exponent can be further reduced by a more careful ansatz design or an improved optimization algorithm. Furthermore, by Trotterizing the continuous quantum evolution, our proposed schedule class could be used as a parameter-efficient ansatz for QAOA and thus our findings are also relevant for gate-based quantum optimization algorithms.
\begin{acknowledgments}
The authors thank Jeremy C{\^o}t{\'e} for many stimulating discussions. 

M.W. and A.R. thankfully acknowledge RES resources provided by Barcelona Supercomputing Center in Marenostrum 5 to INNO-2025-3-0004.

This work was performed, in part, at the Center for Integrated Nanotechnologies, an Office of Science User Facility operated for the U.S. Department of Energy (DOE) Office of Science.

This article has been co-authored by an employee of National Technology \& Engineering Solutions of Sandia, LLC under Contract No. DE-NA0003525 with the U.S. Department of Energy (DOE). The employee owns all right, title and interest in and to the article and is solely responsible for its contents. The United States Government retains and the publisher, by accepting the article for publication, acknowledges that the United States Government retains a non-exclusive, paid-up, irrevocable, world-wide license to publish or reproduce the published form of this article or allow others to do so, for United States Government purposes. The DOE will provide public access to these results of federally sponsored research in accordance with the DOE Public Access Plan \url{https://www.energy.gov/downloads/doe-public-access-plan}.
\end{acknowledgments}
\bibliography{references}

\onecolumngrid

\appendix
\section{Proof of Proposition \ref{lemma:UniversalSU2} - Universality of LZS-gates}
\label{sec:ProofUniversalSU2}
Consider two 2-level Hamiltonians $H_{i=1, 2}$ and the time-dependent interpolation
\begin{equation}
    H(t,T) = \left( 1-\frac{t}{T} \right) H_1 + \frac{t}{T} H_2 \ ,
\end{equation}
for $t = [0, T]$. We define ramp-gates as the time-evolution induced by $H(t)$ as
\begin{equation}
    U_{\text{ramp-up}}(T) = \mathcal{T}\exp \left[ -i \int_0^T H(t, T) dt \right] \ ,
\end{equation}
\begin{equation}
    U_{\text{ramp-down}}(T) = \mathcal{T}\exp \left[ -i \int_T^0 H(t, T) dt \right] \ ,
\end{equation}
and the wait gates
\begin{equation}
    U_{\text{wait}, i}(T) = \exp \left[ -i H_i T\right] \ .
\end{equation}
We will first present the argument under the assumption that both ramp directions induce transitions with the same probability $r$, while the phases can be absorbed into the tunable phases of the wait-gates. This is generally only true for real Hamiltonians, but with a few considerations, it is straightforward to extend the argument to the case of distinct $r$. After applying one of the ramp-gates, the system will transition between the ground and the excited state with some probability $r$, where $r$ depends on $T$. We will consider $r$ to be a given parameter with $0 < r < 1$ and show that universality can be reached with sufficiently many concatenations of ramp and wait gates.

In the instantaneous eigenbasis of $H(t, T)$, the ramp gates $U_{\text{ramp-up}}$ and $U_{\text{ramp-down}}$ can be described as
\begin{equation}
    \begin{aligned}
    U_{\text{ramp}}(r) &= \begin{pmatrix}
        \sqrt{1-r} e^{i(\phi_r + \theta_r)/2} & \sqrt{r} e^{-i(\phi_r - \theta_r)/2} \\
        -\sqrt{r} e^{i(\phi_r - \theta_r)/2} & \sqrt{1-r} e^{-i(\phi_r + \theta_r)/2}
    \end{pmatrix} \\
    &= \begin{pmatrix}
        e^{i\theta_r/2} & 0 \\ 0 & e^{-i\theta_r/2}
    \end{pmatrix}
    \begin{pmatrix}
        \sqrt{1-r} & \sqrt{r} \\ -\sqrt{r} & \sqrt{1-r}
    \end{pmatrix}
    \begin{pmatrix}
        e^{i\phi_r/2} & 0 \\ 0 & e^{-i\phi_r/2}
    \end{pmatrix} \ ,
    \end{aligned}
\end{equation}
where the phases $\phi_r$ and $\theta_r$ depend on $r$. The wait-gates in the instantaneous eigenbasis are
\begin{equation}
    U_{\text{wait}}(\phi) = \begin{pmatrix}
        e^{i\phi/2} & 0 \\ 0 & e^{-i\phi/2}
    \end{pmatrix} \ , 
\end{equation}
where the phase $\phi$ is tunable via the waiting time. Combining the ramp-gate with two wait-gates, one gets
\begin{equation}
    \begin{aligned}
    U(r, \phi, \theta) &= U_{\text{wait}}(\theta) U_{\text{ramp}}(r) U_{\text{wait}}(\phi) \\
    &= \begin{pmatrix}
        e^{i(\theta - \theta_r)/2} & 0 \\ 0 & e^{-i(\theta - \theta_r)/2}
    \end{pmatrix}
    \begin{pmatrix}
        \sqrt{1-r} & \sqrt{r} \\ -\sqrt{r} & \sqrt{1-r}
    \end{pmatrix}
    \begin{pmatrix}
        e^{i(\phi - \phi_r)/2} & 0 \\ 0 & e^{-i(\phi - \phi_r)/2}
    \end{pmatrix} \\
    &= \begin{pmatrix}
        e^{i\theta/2} & 0 \\ 0 & e^{-i\theta/2}
    \end{pmatrix}
    \begin{pmatrix}
        \sqrt{1-r} & \sqrt{r} \\ -\sqrt{r} & \sqrt{1-r}
    \end{pmatrix}
    \begin{pmatrix}
        e^{i\phi/2} & 0 \\ 0 & e^{-i\phi/2}
    \end{pmatrix}
    \end{aligned} \ , 
    \label{eq:DefLZSGate}
\end{equation}
where in the last equality we absorb the $r$-dependent phases into the tunable phases from the wait-gates to simplify the notation.
We can already see that if $r$ could be freely chosen between $r_l = 0$ and $r_u = 1$, the LZS-gate (Eq.~\eqref{eq:DefLZSGate}) would decompose into an $Y$-rotation on the Bloch sphere flanked by two $Z$-rotations, which is a universal single-qubit gate. We will show here that for restricted $r$, a sequence of LZS-gates Eq.~\eqref{eq:DefLZSGate} is still universal.

In fact, it suffices to have control of the phase parameters $\phi$ and $\theta$ and keeping $r \neq 0, 1$ fixed. In the case $r=0,1$ no interference would be possible, and the LZS-gate could not function. If there is a sequence of LZS-gates that transports any population $r$ between the energy levels (and the required phases are added to the first and the last wait-gates in the sequence), then any SU(2)-gate can be realized.

We can re-write the LZS-gate Eq.~\eqref{eq:DefLZSGate} as
\begin{equation}
    U(\lambda, \phi, \theta) = \begin{pmatrix}
                                    e^{i\theta/2} & 0 \\ 0 & e^{-i\theta/2}
                                \end{pmatrix}
                                \begin{pmatrix}
                                    \cos \lambda & \sin \lambda \\ -\sin \lambda & \cos \lambda
                                \end{pmatrix}
                                \begin{pmatrix}
                                    e^{i\phi/2} & 0 \\ 0 & e^{-i\phi/2}
                                \end{pmatrix}
    \label{eq:DecompositionAngles}
\end{equation}
for $\lambda$ such that $\cos \lambda = \sqrt{1-r}$. As a convention, we will choose $0 < \lambda < \pi/2$. To simplify the notation in the next step, we will define $\phi' := \phi_1 + \theta_2$, which is freely tunable. Applying two gates consecutively, we obtain
\begin{equation}
    \begin{aligned}
    &U(\lambda, \phi_1, \theta_1) U(\lambda, \phi_2, \theta_2) \\
    &=
    \begin{pmatrix}
        e^{i\theta_1/2} & 0 \\ 0 & e^{-i\theta_1/2}
    \end{pmatrix}
    \begin{pmatrix}
        \cos \lambda & \sin \lambda \\ -\sin \lambda & \cos \lambda
    \end{pmatrix}
    \begin{pmatrix}
        e^{i\phi'/2} & 0 \\ 0 & e^{-i\phi'/2}
    \end{pmatrix}
    \begin{pmatrix}
        \cos \lambda & \sin \lambda \\ -\sin \lambda & \cos \lambda
    \end{pmatrix}
    \begin{pmatrix}
        e^{i\phi_2/2} & 0 \\ 0 & e^{-i\phi_2/2}
    \end{pmatrix} \\
    &=
    \begin{pmatrix}
        e^{i\theta_1/2} & 0 \\ 0 & e^{-i\theta_1/2}
    \end{pmatrix}
    \begin{pmatrix}
        e^{i\phi'/2} \cos^2 \lambda - e^{-i\phi'/2} \sin^2 \lambda & \cos \lambda \sin \lambda (e^{i\phi'/2} + e^{-i\phi'/2}) \\
        -\cos \lambda \sin \lambda (e^{i\phi'/2} + e^{-i\phi'/2}) & e^{-i\phi'/2} \cos^2 \lambda - e^{i\phi'/2} \sin^2 \lambda
    \end{pmatrix}
    \begin{pmatrix}
        e^{i\phi_2/2} & 0 \\ 0 & e^{-i\phi_2/2}
    \end{pmatrix} \\
    &=
    \begin{pmatrix}
        e^{i\theta_1/2} & 0 \\ 0 & e^{-i\theta_1/2}
    \end{pmatrix}
    \begin{pmatrix}
        \cos(2\lambda) \cos(\phi'/2) + i \sin(\phi'/2) & \sin (2\lambda) \cos(\phi'/2) \\
        -\sin (2\lambda) \cos(\phi'/2) & \cos(2\lambda) \cos(\phi'/2) - i \sin(\phi'/2)
    \end{pmatrix}
    \begin{pmatrix}
        e^{i\phi_2/2} & 0 \\ 0 & e^{-i\phi_2/2}
    \end{pmatrix}
    \end{aligned} \ ,
    \label{eq:DoubleLZSGate1}
\end{equation}
where we applied double-angle identities in the last equality. Note that we can factorize any unitary $U \in \text{SU}(2)$ with real off-diagonals as
\begin{equation}
    U = 
    \begin{pmatrix}
        e^{i\alpha} \cos \lambda' & \sin \lambda' \\ -\sin \lambda' & e^{-i\alpha} \cos \lambda'
    \end{pmatrix}
    =
    \begin{pmatrix}
        e^{i\alpha/2} & 0 \\ 0 & e^{-i\alpha/2}
    \end{pmatrix}
    \begin{pmatrix}
        \cos \lambda' & \sin \lambda' \\ -\sin \lambda' & \cos \lambda'
    \end{pmatrix}
    \begin{pmatrix}
        e^{i\alpha/2} & 0 \\ 0 & e^{-\alpha/2}
    \end{pmatrix}
    \label{eq:UnitaryDecomposition}
\end{equation}
for some $\alpha$ and $\lambda'$. This allows to write the off-diagonal elements of the matrix in the middle of Eq.~\eqref{eq:DoubleLZSGate1} as
\begin{equation}
    \pm \sin \lambda' = \pm \sin(2\lambda) \cos(\phi'/2) \in \left[ -|\sin(2\lambda)|, |\sin(2\lambda)| \right] \ , 
    \label{eq:OffDiagElementEquation}
\end{equation}
where any value in the interval can be obtained by setting the free parameter $\phi'$ to an appropriate value. The diagonal elements are then fixed according to
\begin{equation}
    e^{\pm i\alpha} \cos \lambda' = \cos(2\lambda) \cos(\phi'/2) \pm i \sin(\phi'/2) \ .
    \label{eq:DiagElementEquation}
\end{equation}
The parameters $\lambda'$ and $\alpha$ exist but are not unique. Using an appropriate choice of $\phi'$, we can assume $0 \leq \lambda' \leq \pi/2$, such that $0 < \sin \lambda' \leq |\sin2 \lambda| \leq 1$. Let us define the effective lambda
\begin{equation}
    \lambda_{\text{eff}} = \min \left\{ \lambda, \frac{\pi}{2} - \lambda \right\} \ .
\end{equation}
Since by assumption we have $0 < \lambda < \pi / 2$, it follows that $0 < \lambda_{\text{eff}} \leq \pi /4$, which in turn implies $|\sin \lambda_{\text{eff}}| \leq |\sin (2\lambda_{\text{eff}})|$. Note that the introduction of $\lambda_{\text{eff}}$ is equivalent to switching the assignment of cosine and sine to the diagonal and off-diagonal elements in the decomposition in Eq.~\eqref{eq:DecompositionAngles} such that $0 < \lambda \leq \pi / 4$.

Inserting the decomposition Eq.~\eqref{eq:UnitaryDecomposition} into the expression for the two LZS-gates Eq.~\eqref{eq:DoubleLZSGate1}, we obtain
\begin{equation}
    \begin{aligned}
        &U(\lambda, \phi_1, \theta_1) U(\lambda, \phi_2, \theta_2) \\
        &=
        \begin{pmatrix}
        e^{i(\theta_1 + \alpha)/2} & 0 \\ 0 & e^{-i(\theta_1 + \alpha)/2}
        \end{pmatrix}
        \begin{pmatrix}
            \cos \lambda' & \sin \lambda' \\ -\sin \lambda' & \cos \lambda'
        \end{pmatrix}
        \begin{pmatrix}
            e^{i(\phi_2 + \alpha)/2} & 0 \\ 0 & e^{-i(\phi_2 + \alpha)/2}
        \end{pmatrix} \\
        &=
        \begin{pmatrix}
        e^{i\theta'/2} & 0 \\ 0 & e^{-i\theta'/2}
        \end{pmatrix}
        \begin{pmatrix}
            \cos \lambda' & \sin \lambda' \\ -\sin \lambda' & \cos \lambda'
        \end{pmatrix}
        \begin{pmatrix}
            e^{i\phi'/2} & 0 \\ 0 & e^{-i\phi'/2}
        \end{pmatrix} \\
        &=: U'(\lambda', \phi', \theta')
    \end{aligned} \ , 
    \label{eq:DoubleLZSGate3}
\end{equation}
where we used the fact that $\phi_1$ and $\theta_2$ are free parameters. Eq.~\eqref{eq:DoubleLZSGate3} tells us that two LZS-gates with free phase parameters and fixed transition parameter $\lambda$ are equivalent to a single LZS-gate with free phase parameters $\phi'$ and $\theta'$ and a transition parameter $\lambda'$ that can be chosen such that
\begin{equation}
    0 \leq \sin \lambda' \leq |\sin 2\lambda_{\text{eff}} | \ . 
\end{equation}
In order to further extend the range of accessible transition probabilities, one can fix $\kappa \in [0, 2]$ such that $\lambda' = \min \left\{ \kappa \lambda_{\text{eff}}, \pi / 4 \right\}$ and join two of the gates defined in Eq.~\eqref{eq:DoubleLZSGate3}. Repetitive application of this argument presented above after $n$ iterations leads to
\begin{equation}
    U^{(n)}(\lambda^{(n)}, \phi, \theta) = 
    \begin{pmatrix}
        e^{-i\theta/2} & 0 \\ 0 & e^{-i\theta/2}
    \end{pmatrix}
    \begin{pmatrix}
        \cos \lambda^{(n)} & \sin \lambda^{(n)} \\ -\sin \lambda^{(n)} & \cos \lambda^{(n)}
    \end{pmatrix}
    \begin{pmatrix}
        e^{-i\phi/2} & 0 \\ 0 & e^{-i\phi/2}
    \end{pmatrix} \ , 
\end{equation}
where $\phi$ and $\theta$ can be chosen freely, while $\lambda^{(n)}$ can be chosen such that
\begin{equation}
    0 \leq \sin \lambda^{(n)} = \sin\left(\kappa^{(n)}\lambda^{(n-1)}\right) = \sin\left(\lambda_{\text{eff}} \prod_{m=1}^n \kappa^{(m)} \right) \ , 
    \label{eq:LambdaConstraint}
\end{equation}
where all $\kappa^{(m)} \in [0 , 2]$ can be chosen so that any $\lambda^{(n)} \in [0, 2^n \lambda_{\text{eff}}] \subseteq [0, \pi / 4]$ can be realized, where equality is reached after $n = \lceil \log_2 \frac{\pi}{4\lambda_{\text{eff}}}\rceil$. From one more iteration of the argument, where now we allow $0 \leq \kappa \leq 2$, $\lambda^{(n+1)} \in [0, \pi / 2]$ is reached. Although the sequence of basic LZS-gates doubles in every application of the argument, the range of accessible $\lambda^{(n)}$ also grows exponentially. Therefore, we obtain a universal SU(2)-gate after $N_0 = 2^{n+1} = 2^{\lceil \log_2 \frac{\pi}{2\lambda_{\text{eff}}}\rceil}$ LZS-gates, which is the statement of Proposition \ref{lemma:UniversalSU2}.\\

Note that in the iterative argument, the choice of $\lambda' = \min \left\{ \kappa \lambda_{\text{eff}}, \pi / 4 \right\}$ is not strictly necessary, but it allows to consider a monotonously increasing sequence of $\lambda^{(n)}$, which simplifies the proof. Furthermore, while we assumed identical $r$ for the forward and backward ramps, the argument trivially extends to the case of distinct $0 < r_{i=1,2} < 1$. The main difference occurs only in the first iteration. Consider the concatenation of two ramps with distinct transition probabilities
\begin{equation}
    \begin{aligned}
    &U(\lambda_1, \phi_1, \theta_1) U(\lambda_2, \phi_2, \theta_2) \\
    &=
    \begin{pmatrix}
        e^{i\theta_1/2} & 0 \\ 0 & e^{-i\theta_1/2}
    \end{pmatrix}
    \begin{pmatrix}
        \cos \lambda_1 & \sin \lambda_1 \\ -\sin \lambda_1 & \cos \lambda_1
    \end{pmatrix}
    \begin{pmatrix}
        e^{i\phi'/2} & 0 \\ 0 & e^{-i\phi'/2}
    \end{pmatrix}
    \begin{pmatrix}
        \cos \lambda_2 & \sin \lambda_2 \\ -\sin \lambda_2 & \cos \lambda_2
    \end{pmatrix}
    \begin{pmatrix}
        e^{i\phi_2/2} & 0 \\ 0 & e^{-i\phi_2/2}
    \end{pmatrix}
    \end{aligned} \ ,
    \label{eq:DoubleLZSGateDistinctLambda}
\end{equation}
where, as before, we absorb the different phases from the ramps into the tunable phases of the wait gate. Using again trigonometric identities, it is easy to see that the norm of the off-diagonal elements in Eq.~\eqref{eq:UnitaryDecomposition} would now read
\begin{equation}
    |\sin \lambda'| = |\sin(\lambda_1 + \lambda_2) \cos\phi' -i \sin(\lambda_1 - \lambda_2) \sin \phi' | \ ,
\end{equation}
which are respectively lower and upper  bounded by the minimum and maximum of $|\sin(\lambda_1 - \lambda_2)|$ and $|\sin(\lambda_1 + \lambda_2)|$. If $\lambda_1$ and $\lambda_2$ are close, i.e. $\lambda_1 \approx \lambda_2 = \lambda$, then $\phi'$ can be chosen such that $|\sin \lambda'| \approx |\sin 2\lambda|$ and the argument above applies similarly. If one of the transition rates is close to $1$ while the other is close to $0$, then $|\sin \lambda'|$ will also be close to unity. We can then start the iterative argument from here, reaching universality after $N_0 = 2^{\lceil \log_2 \frac{\pi}{\lambda_{\text{eff}}} \rceil}$, where $\lambda_{\text{eff}} = \max \{ \pi/2 - \lambda_{\text{eff},1} + \lambda_{\text{eff},2} \}$. The doubling of the $N_0$ for when the two transition probabilities are apart reflects the fact that now after two sweeps the interference cannot be completely destructive and there is a net transfer. However, even if the two sweeps are asymmetric, after grouping them they become symmetric.
\section{Proof of Proposition \ref{prop:efficiency} - Efficient optimization}
\label{sec:EfficiencyProof}
We consider the problem of preparing the ground state of a target Hamiltonian $H_Z$ using diabatic quantum annealing (DQA). DQA works by preparing the system in the ground state of a simple initial Hamiltonian $H_X$ and then interpolation between $H_X$ and $H_Z$ according to
\begin{equation}
    H(s) = (1-s) H_X + s H_Z \ .
    \label{eq:Hamiltonian}
\end{equation}
While for adiabatic quantum annealing (AQA), the schedule $s=s(t)$ is simply a monotonic function that interpolates between $H_X$ and $H_Z$ in some (long enough) time $T$, DQA makes explicit use of diabatic transitions, i.e. excitations of the system into higher excited states. In the main text, we propose a schedule that makes use of Landau-Zener-St{\"u}ckelberg interferometry, hence we define the LZS-schedule as
\begin{equation}
    s(t) = \begin{cases}
        s_1 \frac{t}{t_1} & , \ 0 \leq t < t_1 \\
        s_1 & , \ t_1 \leq t < t_2 \\
        s_1 + (s_2 - s_1) \frac{t-t_2}{t_3-t_2} & , \ t_2 \leq t < t_3 \\
        s_2 & , \ t_3 \leq t  < t_4 \\
        s_2 + (1 - s_2)\frac{t-t_4}{t_5-t_4} & , \ t_4 \leq t \leq t_5
    \end{cases} \ .
    \label{eq:DefLZS}
\end{equation}
This schedule is parameterized by the seven parameters $t_{i=1...5}$ and $s_{i=1, 2}$. We will denote them collectively as the parameter vector $\theta$. We will first show that, under specific assumptions, there is an adaptive grid search algorithm that converges to the optimal solution in polynomial time. Then we show that these conditions are met in the setting we consider here.

\subsection{Global optimization of a differentiable function with affinely bounded gradient}
Let $f: \mathcal{D} \rightarrow \mathbb{R}$ be a bounded differentiable function with
\begin{equation}
    x^* = \argmin_{x \in \mathcal{D}} f(x) < \infty \ ,
\end{equation}
i.e. we require $f$ to have a global minimum in the parameter domain $\mathcal{D} \subseteq \mathbb{R}^d $. It simplifies the proof to assume $\mathcal{D} = \mathbb{R}_+^d$ be the strictly positive orthant of $\mathbb{R}^d$. Additionally, we require that the gradient of $f$ is affinely bounded, i.e.
\begin{equation}
    \| \grad f \| \leq K_1 \| x \| + K_2 \ , 
\end{equation}
for real numbers $K_1, K_2 > 0$. The computational task we are asked to solve is to find for any $\epsilon >0$
\begin{equation}
    x_{opt} \in \mathcal{D} : |f(x_{opt}) - f(x^*)| \leq \epsilon \ .
\end{equation}
We propose the following Algorithm \ref{alg:AGO}, which is a version of grid search adapted to respect the fact that $\mathcal{D}$ is unbounded. We will show that the bound on the gradient of $f$ allows to determine the asymptotic runtime.

The optimization of continuous functions tends to suffer from the curse of dimensionality (i.e. exponential dependence on $d$). For simplicity of the argument, we will consider adaptive grid search, as it suffices to show the existence of a polynomial optimization algorithm, assuming that the LZS-ansatz prepares the target ground state efficiently.\\
We will denote the discretization of the set $(0, \Delta ]^d \subset \mathbb{R}_+^d$ into a regular grid with a distance $\rho$ ($\Delta / \rho$ needs to be an integer) between neighboring points as
\begin{equation}
    G_{\Delta, \rho} := \left\{ \rho, ...,  \Delta - \rho, \Delta \right\}^d \ .
\end{equation}
$G_{\Delta, \rho}$ contains $\left( \frac{\Delta}{\rho} \right)^d$ points. Note that in the optimization problem we consider here, we only need to search for positive values, but in principle we could also define $G_{\Delta, \rho}$ so that negative values are covered as well.
\begin{algorithm}[H]
\caption{Adaptive grid search} \label{alg:AGO}
\begin{algorithmic}[1]
    \Require bounded differentiable function $f$ with $x^* = \argmin_x f(x) < \infty$
    \Require initial resolution $\rho > 0$
    \Require initial range $\Delta > 0$
    \Require number of iterations $N_{\text{iter}}$
    \State $f_{opt} \gets \infty, x_{opt} \gets 0$
    \For{$N_{\text{iter}}$ iterations}
    \For{each point $x_i$ in $G_{\Delta, \rho}$}
    \If{$f(x_i) < f_{opt}$}
    \State $f_{opt} \gets f(x_i)$
    \State $x_{opt} \gets x_i$
    \EndIf
    \EndFor
    \State $\Delta \gets 2\Delta$
    \State $\rho \gets \frac{\rho}{2}$
    \EndFor
    \State \Return $x_{opt}$
\end{algorithmic}
\end{algorithm}
First, we can prove the time complexity of Algorithm \ref{alg:AGO} for differentiable $f$ with an affinely bounded gradient, where the time complexity depends polynomially on the constants $K_1$ and $K_2$. In a second step, we show that for the loss function used for our variational schedules these constants only grow polynomially with system size.
\begin{lemma} (Time complexity of Algorithm \ref{alg:AGO})
    Let $f: \mathbb{R}^d_+ \rightarrow \mathbb{R}$ a bounded differentiable function with $x^* = \argmin_x f(x)$ and $\| \grad f \| \leq K_1 \| x \| + K_2$. Then Algorithm \ref{alg:AGO} finds an $x_{opt}$ with $|f(x_{opt}) - f(x^*)| < \epsilon$ for any $\epsilon > 1/\poly(N)$ in at most
    \begin{equation*}
        N_{\text{eval}} = O \left( \left( \frac{\sqrt{d} \|x^*\| (K_1 \|x^*\| + K_2) }{\epsilon} \right)^{d} \log \frac{\sqrt{d} (K_1 \|x^*\| + K_2) }{\epsilon} \right)
    \end{equation*}
    evaluations of $f$.
    \label{lemma:Complexity}
\end{lemma}
\begin{proof}
    First, let us determine a sufficiently large search range $\Delta$ and a sufficiently small precision $\rho$ in order to observe the desired $x_{opt}$. The search range $\Delta$ is easily determined, as it must be large enough to cover $ x^*$, and since $\| x^* \| $ is finite by assumption, it suffices to demand
    \begin{equation}
        \Delta > \| x^* \| \ .
        \label{eq:MinDelta}
    \end{equation}
    Now we will determine the sufficient resolution $\rho$. By assumption, the function $f$ is differentiable, hence according to the mean value theorem
    \begin{equation}
        |f(x) - f(x^*)| \leq \| \grad f(\xi)\| \| x-x^* \| \leq (K_1 \| \xi \| + K_2) \| x-x^* \| \ ,
        \label{eq:funcbound}
    \end{equation}
    for some $\xi \in \mathbb{R}^d$ on the line connecting $x$ and $x^*$, while the second inequality follows from the assumption that $\| \grad f \| \leq K_1 \| \xi \| + K_2$. Since $\xi$ is on a line connecting $x$ and $x*$, we can bound $\| \xi \|$ using the distance from the optimum $\delta := \| x - x^* \|$ as
    \begin{equation}
        \| \xi \| \leq \|x^* \| + \delta \ .
        \label{eq:xibound}
    \end{equation}
    A sufficient condition to ensure $|f(x) - f(x^*)| \leq \epsilon$ is to find an $x$ that $\delta$-close to $x^*$, where $\delta$ is the solution to
    \begin{equation}
        (K_1 \| x^* \| + K_2) \delta + K_1 \delta^2 \leq \epsilon \ ,
    \end{equation}
    which follows from combining Eq.~\eqref{eq:funcbound} and Eq.~\eqref{eq:xibound}. Solving for $\delta$ renders
    \begin{equation}
        \delta \leq \frac{K_1 \| x^* \| + K_2}{2K_1} \left( \sqrt{1+\frac{4K_1 \epsilon}{(K_1 \|x^*\| + K_2)^2}} - 1 \right) \ .
    \end{equation}
    In order to have at least one point of the uniform grid in the $\delta$-ball centered at $x^*$, it suffices to demand
    \begin{equation}
        \rho < \frac{\delta}{\sqrt{d}} \leq \frac{K_1 \| x^* \| + K_2}{2K_1 \sqrt{d}} \left( \sqrt{1+\frac{4K_1 \epsilon}{(K_1 \|x^*\| + K_2)^2}} - 1 \right) \ ,
    \end{equation}
    or equivalently
    \begin{equation}
        \frac{1}{\rho} > \frac{2K_1\sqrt{d}}{K_1 \| x^* \| + K_2} \left( \sqrt{1+\frac{4K_1 \epsilon}{(K_1 \|x^*\| + K_2)^2}} - 1 \right)^{-1} = O \left( \frac{\sqrt{d} (K_1 \|x^*\| + K_2) }{\epsilon} \right) \ .
        \label{eq:MaxRho}
    \end{equation}
    Having found the sufficient range and resolution, we can compute the number of function calls to $f$ by Algorithm~\ref{alg:AGO}. Since $\Delta$ and $\rho$ grow / shrink exponentially with each iteration of the outer for-loop, a sufficiently large range $\Delta$ is reached after $\sim \log \Delta$ iterations, while a sufficiently small resolution $\rho$ is reached after $\sim \log 1/\rho$ iterations. This means we require
    \begin{equation}
        N_{\text{iter}} = O\left( \max \{ \log \Delta, \log 1 / \rho \} \right)
        \label{eq:TMin}
    \end{equation}
    iterations of the outer loop. Since both the required $\Delta$ and $1 / \rho$ asymptotically scale linearly in $\| x^* \|$ and $K_1 \| x^* \| + K_2$, respectively, we can omit taking the maximum of the logarithms in Eq.~\eqref{eq:TMin} and in combination with Eq.~\eqref{eq:MinDelta} and Eq.~\eqref{eq:MaxRho}, we find
    \begin{equation}
        N_{\text{iter}} = O \left( \log \frac{\sqrt{d} (K_1 \|x^*\| + K_2) }{\epsilon} \right) \ .
    \end{equation}
    Note that the inner for-loop requires $\left(\frac{\Delta}{\rho} \right)^d$ calls to $f$ per iteration, and therefore it follows with the number of iterations that
    \begin{equation}
        N_{\text{eval}} = \mathcal{O} \left( \left( \frac{\sqrt{d} \|x^*\| (K_1 \|x^*\| + K_2) }{\epsilon} \right)^{d} \log \frac{\sqrt{d} (K_1 \|x^*\| + K_2) }{\epsilon}\right)
    \end{equation}
    which concludes the proof.
\end{proof}
Clearly, the complexity of Algorithm \ref{alg:AGO} scales exponentially in the dimension $d$, as could be expected. However, the dimension of the parameter space for the LZS-ansatz is fixed to $d=7$, hence the algorithm scales polynomially in $\|x^*\|$ and $K_1, K_2$. In the next section we will argue about the scaling of the optimal parameter vector $x^*$ and the constants $K_1, K_2$ and show that they are polynomial in the system size $N$.

\subsection{The loss landscape of the LZS-ansatz: affine gradient bound}
Assuming that there are parameters $\theta^*$ of the LZS-ansatz (Eq.~\eqref{eq:DefLZS}) that give rise to a state that closely approximates the ground state of $H_Z$ in $T = \poly(N)$, we need to show that the optimal parameters $\theta^*$ are polynomially far away from the origin and that the constants $K_1$ and $K_2$ from the affine bound are also $\poly(N)$, as then by Lemma \ref{lemma:Complexity} the optimal parameters can be found in $\poly(N)$ function calls. Furthermore, as by assumption $t_{i=1...5} = \poly(N)$, each function call takes $\poly(N)$ time.

Showing that $\| \theta^* \| = \poly(N)$ if $T=\poly(N)$ is the simplest step. Note that $\theta = (s_{1,2}, t_{1...5})$, $s_{1,2} \in [0,1]$ and hence $O(1)$, while $t_{1...5} \geq 0$, which implies $T = \poly(N)$ iff $t_{1...5} = \poly(N)$. Since all components of $\theta^*$ are either $O(1)$ or $O(\poly(N))$, the (L2-)norm of $\theta^*$ is also $O(\poly(N))$.

Bounding $K_1$ and $K_2$ requires more work. The function to be optimized here is the expected energy of the target Hamiltonian
\begin{equation}
    \mathcal{L}(\theta) =  \langle \psi(\theta) | H_Z | \psi(\theta) \rangle \ ,
\end{equation}
where
\begin{equation}
    | \psi(\theta) \rangle =  U_T | \psi_0\rangle = \mathcal{T} \exp\left[ -i \int_{0}^{T} H(s(t)) dt \right] | \psi_0\rangle \ , 
\end{equation}
with $T = \sum_{n=1}^5 t_i$ and $H(s)$ and $s(t)$ from Eq.~\eqref{eq:Hamiltonian} and Eq.~\eqref{eq:DefLZS} respectively. We need to show that
\begin{equation}
    \| \grad \mathcal{L}|_{\theta} \| \overset{!}{\leq} K_1 \|\theta \| + K_2 \ ,
    \label{eq:DerivativeBound}
\end{equation}
which can be accomplished by showing that
\begin{equation}
\begin{aligned}
    |\partial_{\theta_i}\mathcal{L}| &= 2 | \Re \langle \psi(\theta) | H_Z | \partial_{\theta_i} \psi(\theta) \rangle | \\
    &\leq 2 | \Re \langle \psi_0 | U_T^\dagger H_Z (\partial_{\theta_i} U_T) | \psi_0 \rangle | \\
    &\leq 2 \| H_Z \| \| \partial_{\theta_i} U_T \| \\
    &\overset{!}{\leq} K'_1 \| \theta \| + K'_2 \ ,
\end{aligned}
\label{eq:Derivative0}
\end{equation}
for any $\partial_{\theta_i} := \frac{\partial}{\partial \theta_i}$. Since target Hamiltonians are assumed to satisfy $\| H_Z \| = \poly(N)$ for practical systems, it suffices to investigate the norm of the derivatives of $U_T$.\\
The unitary $U_T$ can be decomposed into five parts, according to the pieces of the schedule $s(t)$
\begin{equation}
     U_T = U_5(s_2, t_4, t_5) U_4(s_2, t_3, t_4) U_3(s_1, s_2, t_2, t_3) U_2(s_1, t_1, t_2) U_1(s_1, t_1) \ ,
     \label{eq:UTStepwise}
\end{equation}
where $U_1$, $U_3$ and $U_5$ are the unitaries corresponding to the linear ramps in the schedule, and $U_2$ and $U_4$ are the evolutions with constant Hamiltonians.

Since some of the parameters occur in two unitaries in Eq.~\eqref{eq:UTStepwise}, the product rule applies and the derivative $\partial_{\theta_i} U_T$ will be the sum of at most two terms. Each term is the derivative of a unitary multiplied by the remaining unitaries. Applying the Cauchy-Schwartz and the triangle inequalities, the norm of the derivative $\partial_{\theta_i} U_T$ can be bounded by the sums of at most two of the derivatives of the individual steps. Thus, it suffices to show that the derivatives of the unitaries $U_{1...5}$ are bounded by a polynomial in $N$.

\subsubsection{Derivatives of the unitaries with constant Hamiltonian}
Let us start with computing the derivatives of the evolutions of the constant Hamiltonians as
\begin{equation}
    \begin{aligned}
        &\| \partial_{t_i} \exp \left[ -i (t_{i+1} - t_i) H(s_j) \right] \| \\
        = &\| \partial_{t_{i+1}} \exp \left[ -i (t_{i+1} - t_i) H(s_j) \right] \| \\
        = &\| H(s_j) \exp \left[ -i \Delta t H(s_j) \right] \| \\
        \leq &\| H(s_j) \| = \poly(N) \ ,
    \end{aligned}
    \label{eq:Derivative1}
\end{equation}
and since by assumption $\|H_X \|, \| H_Z \| = \poly(N)$, it follows also that $\| H(s) \| = \| (1-s)H_X + s H_Z \| = \poly(N)$ for any $s \in [0, 1]$. Furthermore, we find
\begin{equation}
    \begin{aligned}
        &\| \partial_{s_j}\exp \left[ -i(t_{i+1}-t_{i})H(s_j) \right] \| \\
        = &\| (H_Z - H_X) (t_{i+1}-t_{i}) \exp \left[ -i(t_{i+1}-t_{i})H(s_j) \right] \| \\
        \leq &\| H_Z - H_X \| (t_{i+1}-t_{i}) \\
        = &\poly(N)
    \end{aligned}
    \label{eq:Derivative2a}
\end{equation}
as long as $\| H_X \|, \|H_Z \| = \poly(N)$ and $t_i = \poly(N)$. Note that Eq.~\eqref{eq:Derivative2a} is why we required the affine bound on the gradient.

For the operator norms of the ramp unitaries, we will consider the derivatives of $U_3$, since $U_1$ and $U_5$ can be treated analogously. 

\subsubsection{Derivatives of the linear ramps with respect to $s_i$}
Note that $U_3(s_1, s_2, t_2, t_3)$ is the solution to the Schr{\"o}dinger equation
\begin{equation}
    i \frac{\partial}{\partial t} U(t) = H(t) U(t) \ ,
\end{equation}
integrated from $t_2$ to $t_3$ with initial $U(t=t_2) = \mathbb{I}$, the solution being the time-ordered exponential
\begin{equation}
    U_3(s_1, s_2, t_2, t_3) = \mathcal{T}\exp[-i\int_{t_2}^{t_3}H(t) dt] \ , 
\end{equation}
where $s_1$ and $s_2$ are parameters of the Hamiltonian $H(t)$. This time-ordered exponential can be discretized as
\begin{equation}
    \mathcal{T}\exp[-i\int_{t_2}^{t_3}H(t) dt] = \lim_{M \rightarrow \infty} \prod_{m=1}^M\exp[-iH(t_m) \Delta t] \ , 
\end{equation}
and with the product rule we obtain for the derivative with respect to the $s_i$
\begin{equation}
    \begin{aligned}
        \partial_{s_i} \mathcal{T}\exp[-i\int_{t_2}^{t_3}H(t) dt] =
        \lim_{M \rightarrow \infty} \sum_{m'=1}^M &\left( \prod_{m=m'+1}^M \exp\left[ -i H(t_m) \Delta t \right] \right) \\
        &\cdot \left( -i\Delta t \partial_{s_i} H(t_{m'}) \right) \\
        &\cdot \left( \prod_{m=1}^{m'} \exp\left[ -i H(t_m) \Delta t \right] \right) \ , 
    \end{aligned}
\end{equation}
which, after evaluating the limit, simplifies to
\begin{equation}
    \partial_{s_i} \mathcal{T}\exp[-i\int_{t_2}^{t_3}H(t) dt] = -i\int_{t_2}^{t_3}U(t_3, t)  ( \partial_{s_i} H(t)) U(t, t_2)  dt \ ,
\end{equation}
where $U(t, t')$ are the time-evolution operators from time $t'$ to time $t$. We can bound the operator norm according to
\begin{equation}
    \begin{aligned}
        \left\| \partial_{s_i} \mathcal{T}\exp[-i\int_{t_2}^{t_3}H(t) dt] \right\| &= \left\| \int_{t_2}^{t_3}U(t_3, t)  ( \partial_{s_i} H(t)) U(t, t_2)  dt \right\| \\
        &\leq \int_{t_2}^{t_3} \left\| U(t_3, t)  ( \partial_{s_i} H(t)) U(t, t_2) \right\|  dt \\
        &\leq \int_{t_2}^{t_3} \| \partial_{s_i} H(t) \|  dt \ .
    \end{aligned}
\end{equation}
Considering the definition of the Hamiltonian (Eq.~\eqref{eq:Hamiltonian}) and the schedule (Eq.~\eqref{eq:DefLZS}), the derivatives of $H(t)$ give
\begin{equation}
    \begin{aligned}
        &\partial_{s_1}H(t) = (H_Z - H_X) \left(1- \frac{t-t_2}{t_3 - t_2} \right) \ , \\
        &\partial_{s_2}H(t) = (H_Z - H_X) \frac{t-t_2}{t_3 - t_2} \ ,
    \end{aligned}
\end{equation}
and thus we can bound
\begin{equation}
    \begin{aligned}
        \| \partial_{s_1} U(s_1, s_2, t_2, t_3) \| \leq \| H_Z - H_X \| \int_{t_2}^{t_3} 1- \frac{t-t_2}{t_3 - t_2} dt = \| H_Z - H_X \| \frac{t_3 - t_2}{2} \ ,
    \end{aligned}
    \label{eq:Derivative2}
\end{equation}
and
\begin{equation}
    \| \partial_{s_2} U(s_1, s_2, t_2, t_3) \| \leq \| H_Z - H_X \| \int_{t_2}^{t_3} \frac{t-t_2}{t_3 - t_2} dt = \| H_Z - H_X \| \frac{t_3 - t_2}{2} \ , 
    \label{eq:Derivative3}
\end{equation}
which are polynomial in $N$ under the same assumptions discussed above.

\subsubsection{Derivatives of the linear ramps with respect to $t_i$}
We can transform the time coordinates and recognize that $U_3(s_1, s_2, t_2, t_3)$ is also the solution to the Schr{\"o}dinger equation
\begin{equation}
    i \frac{\partial}{\partial s} U(s) = \frac{t_3 - t_2}{s_2 - s_1} H(s) U(s) \ ,
\end{equation}
integrated from $s_1$ to $s_2$ with initial condition $U(s=s_1) = \mathbb{I}$. The solution is the time-ordered exponential
\begin{equation}
    U(s_1, s_2, t_2, t_3) = \mathcal{T} \exp\left[-i \frac{t_3 - t_2}{s_2 - s_1} \int_{s_1}^{s_2} H(s) ds \right] \ .
\end{equation}
In order to compute $\partial_{t_2} U$ we can discretize the time-ordered exponential as done before
\begin{equation}
    \mathcal{T} \exp \left[ -i\frac{t_3 - t_2}{s_2 - s_1} \int_{s_1}^{s_2} H(s) ds \right] = \lim_{M\rightarrow \infty} \prod_{m=1}^M \exp\left[ -i\frac{t_3 - t_2}{s_2 - s_1} H(s_m) \Delta s \right] \ ,
\end{equation}
such that after applying the product-rule and evaluating the limits we find
\begin{equation}
    \partial_{t_i} \mathcal{T} \exp \left[ -i\frac{t_3 - t_2}{s_2 - s_1} \int_{s_1}^{s_2} H(s) ds \right] = \frac{\pm i}{s_2 - s_1} \int_{s_1}^{s_2}U(s_2, s) H(s) U(s, s_1) ds \ ,
\end{equation}
where the sign on the right-hand side depends on $i=1$ or $i=2$.
The operator norm then gives
\begin{equation}
    \begin{aligned}
        &\left\| \partial_{t_i} \mathcal{T} \exp \left[ -i\frac{t_3 - t_2}{s_2 - s_1} \int_{s_1}^{s_2} H(s) ds \right] \right\| \\
        = \frac{1}{s_2 - s_1}&\left\| \int_{s_1}^{s_2}U(s_2, s) H(s) U(s, s_1) ds \right\| \\
        \leq \frac{1}{s_2 - s_1} &\int_{s_1}^{s_2} \| U(s_2, s) H(s) U(s, s_1) \| ds \\
        \leq \frac{1}{s_2 - s_1} &\int_{s_1}^{s_2} \| H(s) \| ds = \poly(N) \ .
    \end{aligned}
    \label{eq:Derivative4}
\end{equation}
In conclusion, all derivatives of $\mathcal{L}$ are bounded by $\poly(N)$, as long as $ \| H_X \| , \| H_Z \| = \poly(N)$, as can be assumed for physical systems and $T=\sum_{i=1}^5 t_i = \poly(N)$. Therefore, we can conclude that
\begin{equation}
    \| \partial_{\theta_i} U_T \| \leq \poly(N) \ ,
\end{equation}
if $T=\sum_{i=1}^5 t_i = \poly(N)$. Combining Eq.~\eqref{eq:Derivative0}, Eq.~\eqref{eq:Derivative1}, Eq.~\eqref{eq:Derivative2}, Eq.~\eqref{eq:Derivative3} and Eq.~\eqref{eq:Derivative4} shows the bound Eq.~\eqref{eq:DerivativeBound} with $K_1, K_2 = \poly(N)$. Note that Algorithm \ref{alg:AGO} optimizes over $\mathbb{R}_+^d$. Since the loss $\mathcal{L}$ continuously extends to $[0, \infty)^d$, the proof extends trivially to the boundaries. Similarly, if any or all of the parameters are bounded, as is the case with the $s_i$, the proof also holds. Hence, Algorithm \ref{alg:AGO} finds the optimal parameters in polynomially many evaluations of $\mathcal{L}$. Since all parameter points that are tested are with $t_i = \poly(N)$, each function evaluation on the annealer also requires at most polynomial time. Therefore, the optimization is efficient, which is the claim of Proposition \ref{prop:efficiency}.
\section{Details of numerical simulations}
%

\subsection{Noise models}
\label{sec:NoiseModel}
We assume a two level system for all the noise models we consider. Denote the computational basis by $\ket{0}, \ket{1}$ and the instantaneous energy eigenbasis by $\ket{E_0(t)}, \ket{E_1(t)}$. Some of these cases were explicitly studied in Ref.~\cite{Albash2015} in the context of slow annealing. When comparing between the different models, one has to be a little careful how to do this in a `fair' way since they describe different decoherence processes and the strength of decoherence for each is not equivalent.  All units assume $\hbar = 1$.
%
\subsubsection{Dephasing in the computational basis}
\begin{equation}
\frac{d}{dt} \rho(t) = -i \left[ H(t), \rho(t) \right] + \gamma \left( \sigma^z \rho(t) \sigma^z - \rho(t) \right) \ .
\end{equation}
Here the Lindblad operator is $L = \sigma^z$ and the strength of the dephasing is controlled by $\gamma$.
\subsubsection{Amplitude damping}
Assume that the fixed point is the $\ketbra{0}{0}$ state.  Then the Lindblad master equation is given by:
\begin{equation}
\frac{d}{dt} \rho(t) = -i \left[ H(t), \rho(t) \right] + \gamma \left(  \sigma_-  \rho(t) \sigma_+ - \frac{1}{2} \left\{ \sigma_+ \sigma_- , \rho(t) \right\} \right) \ ,
\end{equation}
where the Lindblad operator is $L = \sigma_- = \ketbra{0}{1}$.
\subsubsection{Isotropic depolarization}
The Lindblad master equation is given by:
\begin{equation}
\frac{d}{dt} \rho(t) = -i \left[ H(t), \rho(t) \right] + \gamma \sum_{\alpha \in \left\{x,y,z \right\}}\left(  \sigma_\alpha  \rho(t) \sigma_\alpha - \rho(t) \right) \ ,
\end{equation}
where the three Lindblad operators are $L_1 = \sigma^x$, $L_2 = \sigma^y$, $L_3 = \sigma^z$.
\subsubsection{Adiabatic Master Equation}
The Lindblad equation is given by:
\begin{equation}
\frac{d}{dt} \rho(t) = -i \left[ H(t), \rho(t) \right] + \sum_{\omega} \gamma(\omega) \left( L_{\omega}(t) \rho(t) L_{\omega}(t)^\dagger - \frac{1}{2} \left\{ L_{\omega}^{\dagger}(t) L_{\omega}(t), \rho(t) \right\} \right) \ ,
\end{equation}
where for convenience we have excluded the Lamb shift term.  The sum over $\omega$ is over all differences of energy eigenvalues, $\omega_1 = E_1(t) - E_0(t), \omega_0 = 0, \omega_{-1} = E_0(t) - E_1(t)$ and the time-dependent Lindblad operators are given by:
\begin{eqnarray}
L_{\omega_1} &=& \bra{E_0(t)} \sigma^z \ket{E_1(t)} \ketbra{E_0(t)}{E_1(t)} \\
L_{0} &=& \bra{E_0(t)} \sigma^z \ket{E_0(t)} \ketbra{E_0(t)}{E_0(t)} +  \bra{E_1(t)} \sigma^z \ket{E_1(t)} \ketbra{E_1(t)}{E_1(t)} \\
L_{\omega_{-1}} &=& \bra{E_1(t)} \sigma^z \ket{E_0(t)} \ketbra{E_1(t)}{E_0(t)} \ . 
\end{eqnarray}
The coefficient $\gamma(\omega)$ is taken to be of the form:
\begin{equation}
    \gamma(\omega) = 2 \pi g^2 \eta \frac{ \omega e^{-|\omega|/\omega_c}}{1- e^{-\beta \omega}} \ ,
\end{equation}
where $g^2 \eta$ controls the strength of the noise, $\beta$ is the inverse temperature, and $\omega_c$ is a cutoff usually taken to be $8 \pi$.

\section{Additional numerical results}
\subsection{Landau-Zener-St{\"u}ckelberg interference on the frustrated Ising ring in Ref.~\cite{Cote_2023}}
\label{sec:AppendixLZSDemo}
In this section we briefly demonstrate LZS interference in the results by C{\^o}t{\'e} et al. To this end, we take the optimized schedule for the frustrated Ising ring at $N=9$ from Ref.~\cite{Cote_2023} and insert a waiting time $\tau$ at the first peak of the non-monotonic annealing schedule. Note that we add the plateau while leaving the schedule otherwise unchanged, therefore the total anneal time of this modified schedule is $T_{\text{opt}} + \tau$, where $T_{\text{opt}}$ is the optimal anneal time found in Ref.~\cite{Cote_2023}. An example of the modified schedule is shown in Figure \ref{fig:LZSPattern}(a). Considering the residual energy $\varepsilon_{\text{residual}}$ as a function of the waiting time $\tau$ in Figure \ref{fig:LZSPattern}(b), we see that $\varepsilon_{\text{residual}}$ seems to consist of a slow oscillation with a large amplitude and fast oscillations with small amplitude. The large oscillation follows the population of the first excited state, as can be seen in Figure \ref{fig:LZSPattern}(c). Figure \ref{fig:LZSPattern}(c) also shows how the system oscillates between ground and first excited state as a function of $\tau$, as would be expected in LZS interference.
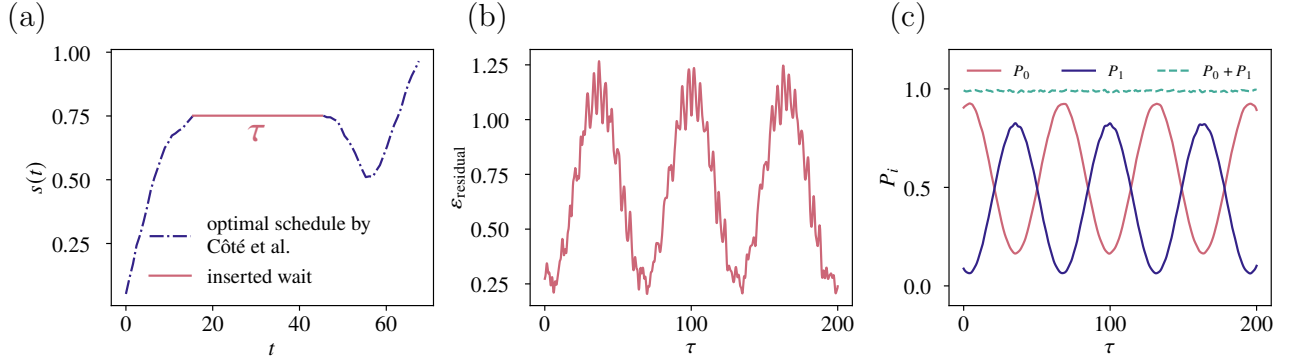
\begin{figure}[H]
    \centering
    \begin{tikzpicture}
        \node[anchor=north west, inner sep=0] (fig) at (0,0){\includegraphics[scale=0.55]{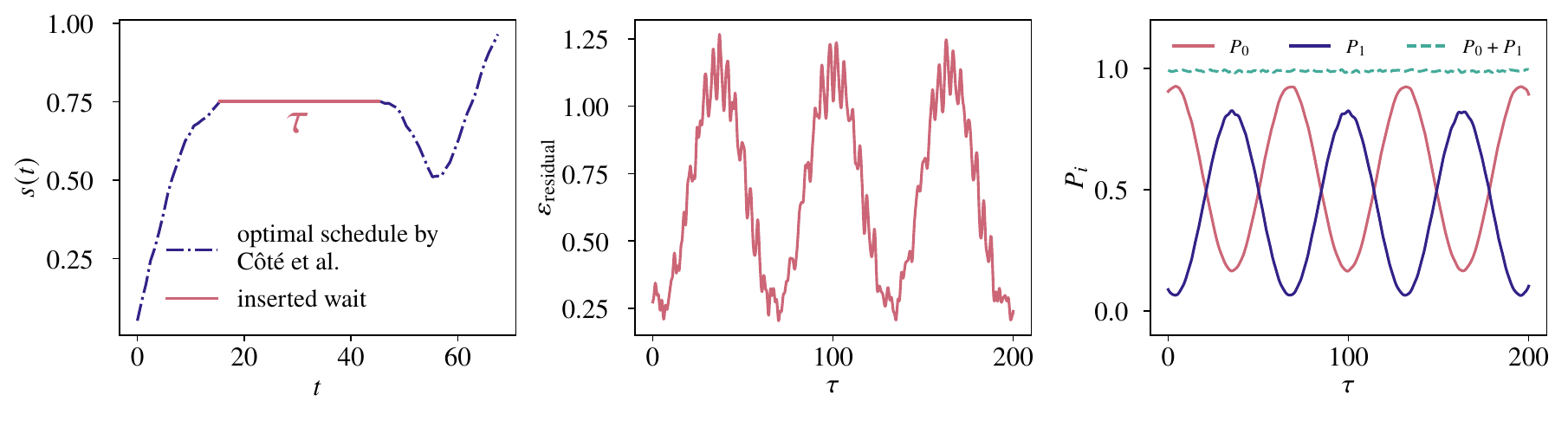}};
        \node[font=\large] at ($(top.north west)!0.01!(top.north east) + (0, 0.2cm)$) {(a)};
        \node[font=\large] at ($(top.north west)!0.36!(top.north east) + (0, 0.2cm)$) {(b)};
        \node[font=\large] at ($(top.north west)!0.68!(top.north east) + (0, 0.2cm)$) {(c)};
    \end{tikzpicture}
    \caption{(a) optimal annealing schedule for the frustrated Ising ring with $N=9$ identified by C{\^o}t{\'e} et al. in Ref.~\cite{Cote_2023} where we insert a waiting time $\tau$. Running the schedule with the waiting time and recording the expected target energy as a function of $\tau$ (b), we see that the cost function oscillates, as is expected in LZS interference. The signal consists of a large oscillation due to the system oscillating between ground and first excited state (c), while the fast oscillations of smaller amplitude are due to transitions between higher excitations.}
    \label{fig:LZSPattern}
\end{figure}

\subsection{Eigenstate populations for the frustrated Ising ring}
\label{sec:AppendixIsingRingSchedules}
Here we show the results for the simulation of the full state vector of the optimal annealing schedules for the frustrated Ising ring for $N=9, 11, 13$ and $c=0.5, 0.1$ for linear (Figure \ref{fig:FrustratedIsingRingAdditionalResultLinear}) and cubic (Figure \ref{fig:FrustratedIsingRingAdditionalResultCubic}) interpolation. As discussed in the main text in Sections \ref{sec:IsingRingSimulation} and \ref{sec:CubicSplines}, the dynamics is qualitatively the same in all cases, except the required time $T$ is shorter for cubic interpolation. The schedules sweep back and forth, populating the first excited state. This allows to find the ground state after ramping through the exponentially small gap close to the end of the anneal. At the same time, the dynamics is entirely contained in the ground and first excited states.

\begin{figure}[H]
    \centering
    \begin{tikzpicture}
        \node[anchor=north west, inner sep=0] (fig) at (0,0){
            \includegraphics[width=0.92\linewidth]{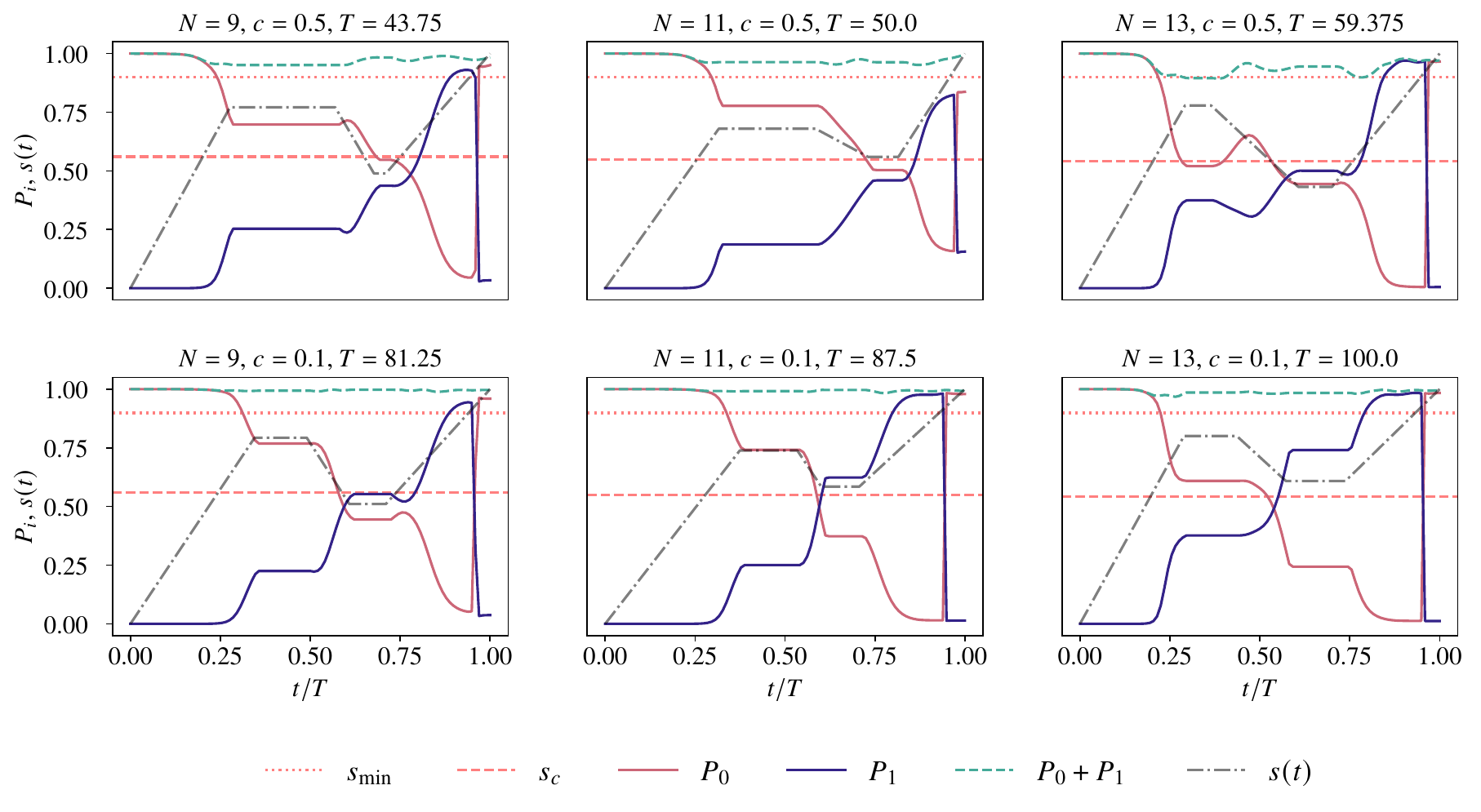}};
    \end{tikzpicture}
    \caption{Complementary data to the linear interpolation schedules as in Figure \ref{fig:FrustratedIsingRingResult}(b) for $N=9, 11, 13$ (left, middle, right) and $c=0.5, 0.1$ (top row, bottom row).}
    \label{fig:FrustratedIsingRingAdditionalResultLinear}
\end{figure}

\begin{figure}[H]
    \centering
    \begin{tikzpicture}
        \node[anchor=north west, inner sep=0] (fig) at (0,0){
            \includegraphics[width=0.92\linewidth]{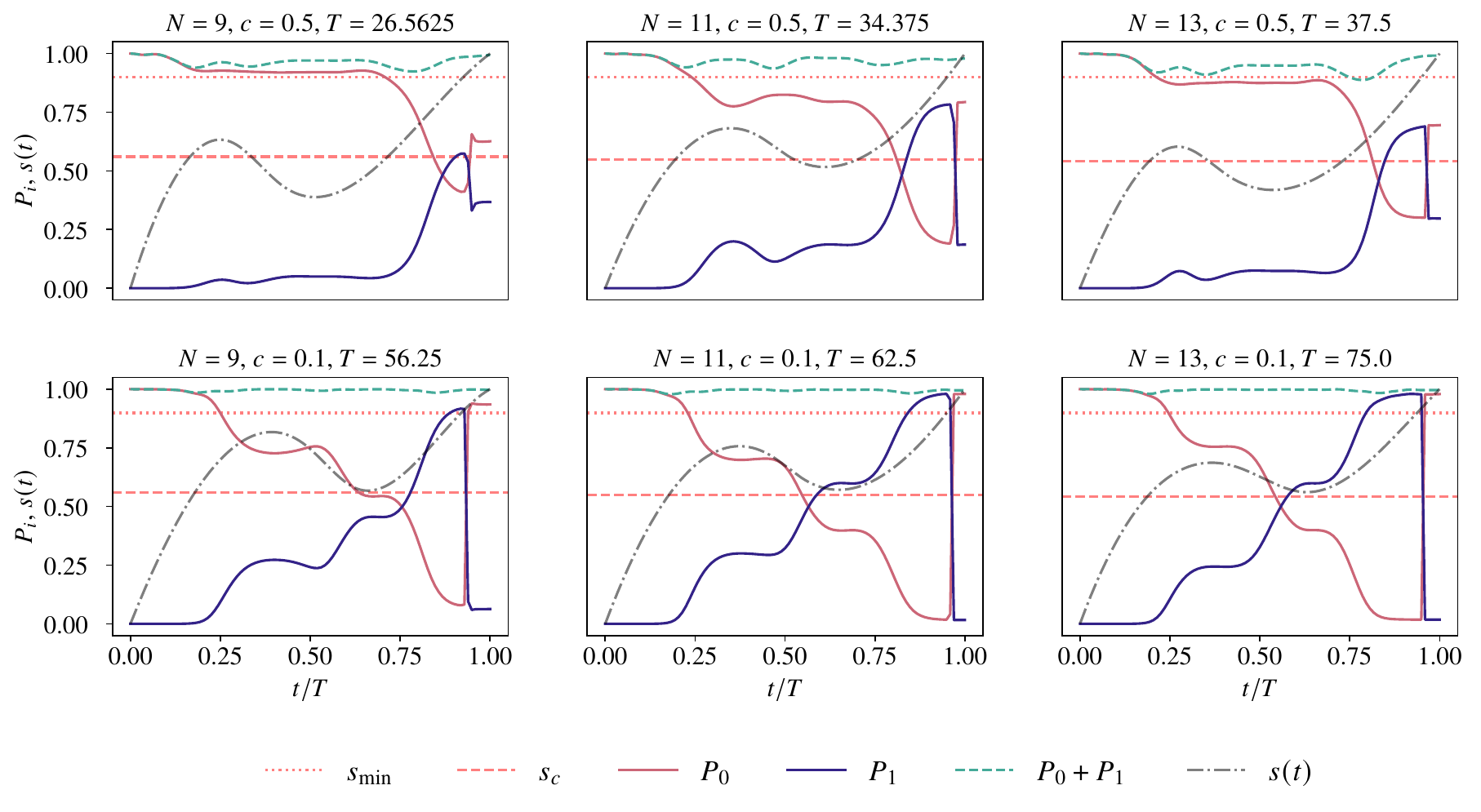}};
    \end{tikzpicture}
    \caption{Complementary data to the cubic spline interpolation schedules as in Figure \ref{fig:FrustratedIsingRingResultCubic}(b) for $N=9, 11, 13$ (left, middle, right) and $c=0.5, 0.1$ (top row, bottom row).}
    \label{fig:FrustratedIsingRingAdditionalResultCubic}
\end{figure}

\subsection{MAXCUT Instance 1 - optimal schedules}
\label{sec:AppendixDataInstance1}
As discussed in the main text, for MAXCUT Instance 1 we optimize the ansatz for an increasing time budget $T$. As $T$ increases, we can identify three regimes of distinct annealing strategies. In the main text, we show the best performing schedule for each regime, while here we show all of the schedules for completeness in Figure \ref{fig:INSTANCE1_AllSchedules} using the same color code as in the main text.
\begin{figure}[H]
    \centering
    \begin{tikzpicture}
        \node[anchor=north west, inner sep=0] (fig) at (0,0){\includegraphics[scale=0.7]{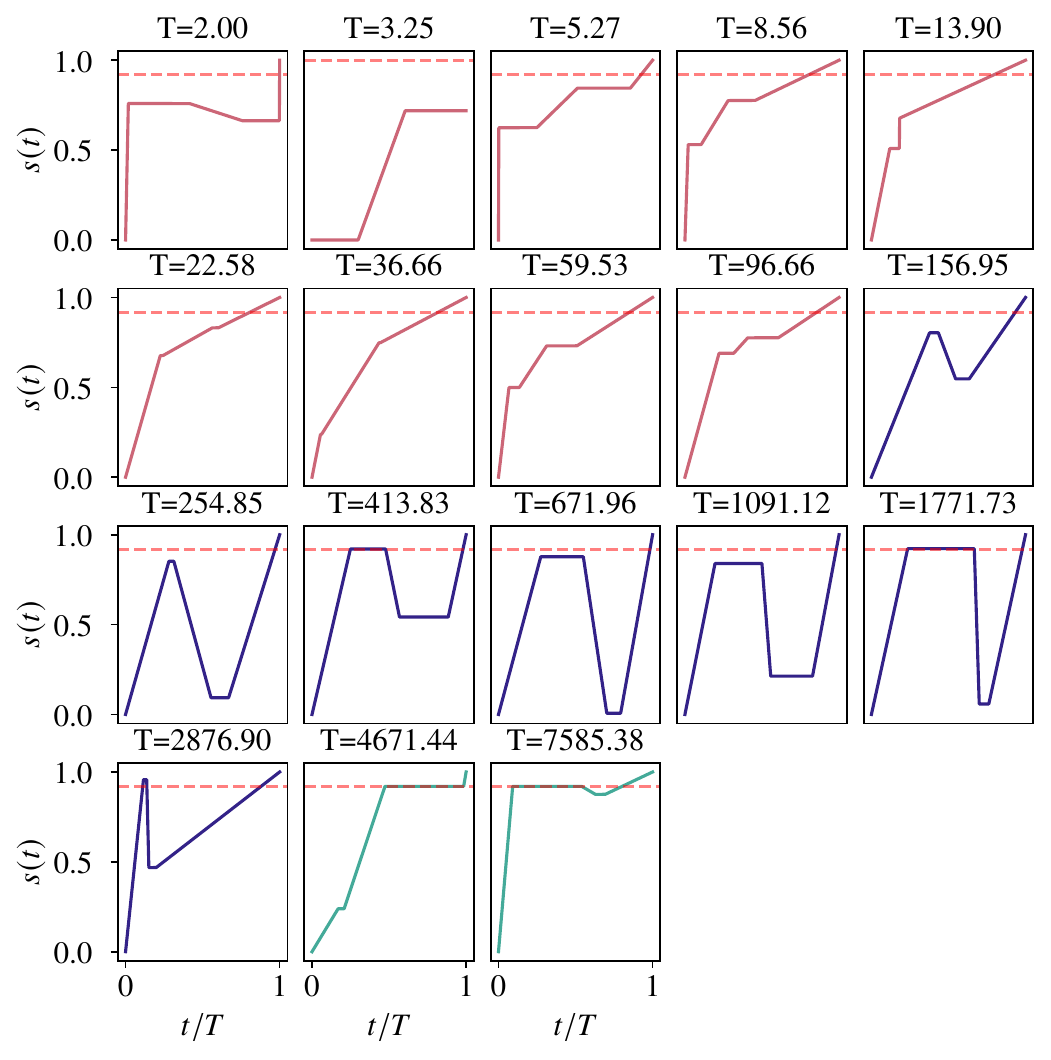}};
    \end{tikzpicture}
    \caption{Optimal schedules for MAXCUT Instance 1 with $N=14$ qubits. As the anneal time $T$ increases, we can identify regime 1 (red) of monotonic schedules, with the negligible exception of $T=2.0$, which is much too short of a time to be considered as a solution. For longer $T$, the schedules become non-monotonic with a pause around the minimum gap (regime 2, blue), while for even larger $T$, the schedules pause exactly at the minimal gap (dashed red), constituting regime 3 (green).}
    \label{fig:INSTANCE1_AllSchedules}
\end{figure}

\subsection{Random MAXCUT instances}
\label{sec:AppendixRandomMaxcut}
While we can observe a qualitative improvement of the residual energy on a particular MAXCUT instance with a perturbative anti-crossing, we investigate the generality of the result by applying our ansatz to 20 randomly generated instances of MAXCUT. The instances are constructed by generating random 3-regular graphs with $N=10$ nodes and sampling the weights of each edge i.i.d. uniformly from $[0, 1]$ For each instance, we show the residual energy $\varepsilon_{\text{residual}}$ for all the instances in Figure \ref{fig:MAXCUTRandomInstancesAll}. We show $\varepsilon_{\text{residual}}$ for quantum annealing with linear ramps (blue) and for the optimized ansatz (red). Most of the optimized schedules are monotonic. i.e. $s_2 \geq s_1$, but for a few instances the schedule is non-monotonic. These instances are denoted with blue crosses. For the instances we test here, the optimized schedules almost always perform only slightly better that the linear ramp, but both methods seem to prepare a low-energy state with a residual energy below the first excited state energy in a comparable amount of time.
\begin{figure}[H]
    \centering
    \begin{tikzpicture}
        \node[anchor=north west, inner sep=0] (fig) at (0,0){
            \includegraphics[width=\linewidth]{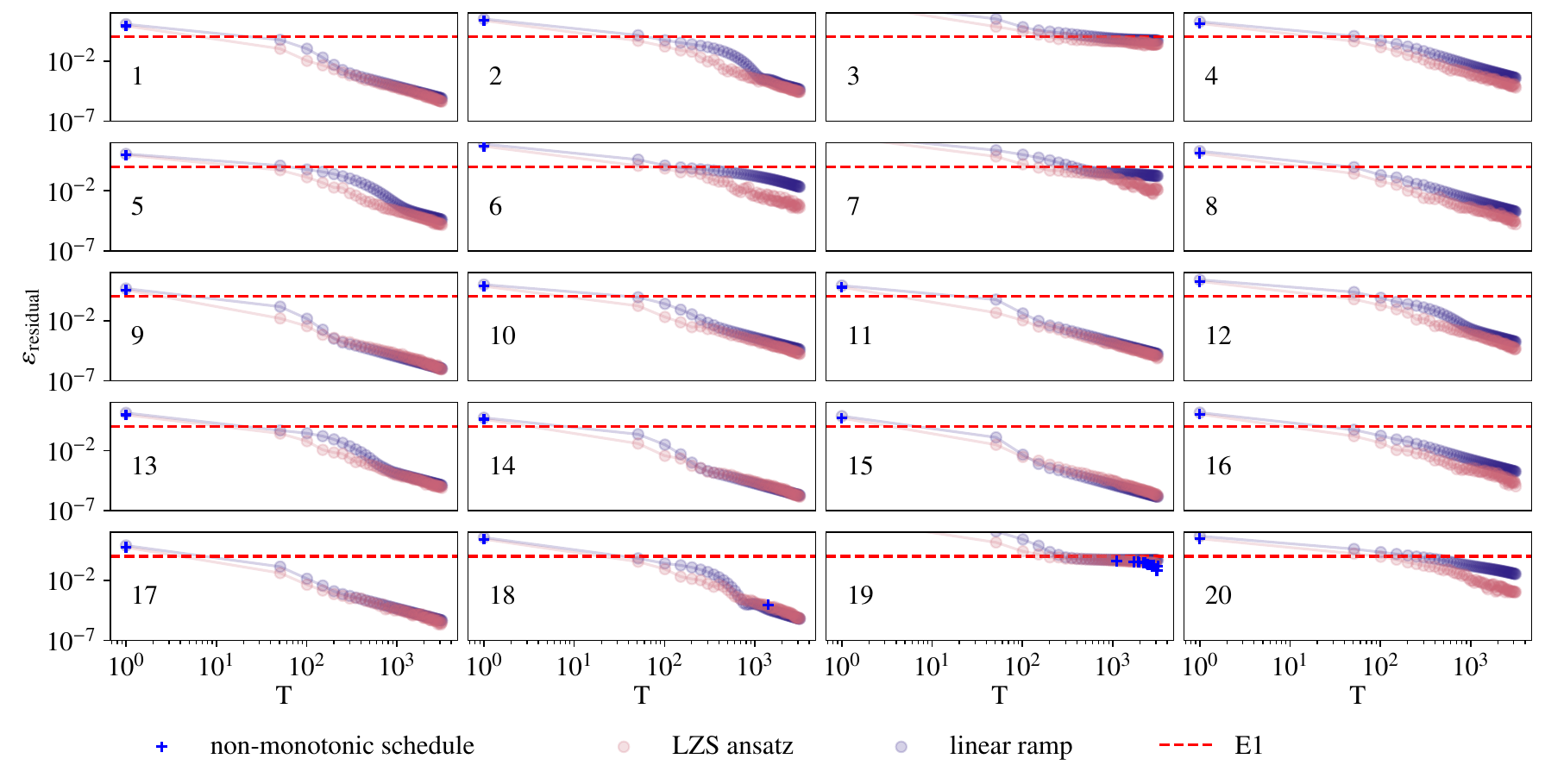}};
    \end{tikzpicture}
    \caption{Residual energy $\varepsilon_{\text{residual}}$ for 20 randomly generated MAXCUT instances on N=10 nodes as a function of annealing time $T$. The data for the optimized ansatz are shown as red dots if the optimal schedule is monotonic, or as blue crosses if the ansatz is non-monotonic. The residual energy of the linear ramp is shown as red dots, while the red dashed line corresponds to the first excited state energy of the target Hamiltonian.}
    \label{fig:MAXCUTRandomInstancesAll}
\end{figure}

\end{document}